%% file: reg_capital_credit_risk.tex
\documentclass[a4paper,10pt, reqno]{amsart}
\usepackage[a4paper, portrait]{geometry}	
		
\usepackage{amsfonts, amsmath, amssymb, amsthm}	
\usepackage{amsaddr}

\usepackage{bbm, dsfont, mathrsfs}		
\usepackage[T1]{fontenc}
\usepackage{url}

\usepackage{array, booktabs, multirow}
\usepackage[retainorgcmds]{IEEEtrantools}	
\usepackage{mathtools}
\usepackage{thmtools, thm-restate}

\usepackage{enumitem}
\usepackage{endnotes}				
	
\usepackage{footnote}

\usepackage{afterpage, pdflscape}
\usepackage[figuresright]{rotating}
\usepackage[usenames, svgnames, table]{xcolor}
	\definecolor{dark_gray}{gray}{0.05}
\usepackage{graphicx} 
	\usepackage{epstopdf}

\usepackage{abstract}

\usepackage{calc}			
\usepackage{xargs}			
\usepackage{setspace}		
\usepackage{comment}

\theoremstyle{plain}
\newtheorem{theorem}{Theorem}[section]
\newtheorem{lemma}[theorem]{Lemma}
\newtheorem{corollary}[theorem]{Corollary}
\newtheorem{proposition}[theorem]{Proposition}

\theoremstyle{definition}
\newtheorem{definition}[theorem]{Definition}

\theoremstyle{remark}
\newtheorem{remark}[theorem]{Remark}

\numberwithin{equation}{section}

\newcommand{\One}[1]{\mathds{1}_{\{#1\}}}
\newcommand{\Corr}[2]{\operatorname{Corr}\!\left(#1, #2\right)}

\newcommand{\dee}[1]{\operatorname{d}\!#1}

\newcommand{\Kap}[2][]{\operatorname{\widetilde{K}}_{#1}\!\left(#2\right)}
\newcommand{\normdist}[2]{\mathcal{N}\!\left(#1, #2\right)}
\newcommand{\q}[2]{\operatorname{q}_{#1}\!\left(#2\right)}

\newcommand{\sqrts}[2][]{\,\sqrt[#1]{#2}\,}
\newcommand{\Var}[1]{\operatorname{Var}\!\left(#1\right)}
\newcommand{\VaR}[2][]{\operatorname{VaR}_{#1}\!\left(#2\right)}

\def\E{\mathbb{E}}
\def\ead{\delta}
\def\EAD{\Delta}
\def\el{l}
\def\given{\,|\,}

\def\lb{\theta}
\def\lgd{\eta}
\def\N{\mathbb{N}}
\def\one{\mathds{1}}
\def\P{\mathbb{P}}
\def\R{\mathbb{R}}
\def\ub{\Theta}

\usepackage[all, bottom]{background}
\SetBgContents{\parbox{\textwidth}{This article is \textcopyright{} World Scientific Publishing and permission has been granted for this version to appear on \url{www.arxiv.org}. World Scientific does not grant permission for this article to be further copied/distributed or hosted elsewhere without the express permission from World Scientific Publishing Co.\ Pte.\ Ltd.\vspace{20pt}}}
\SetBgColor{dark_gray}
\SetBgScale{1.0}

\usepackage[authoryear, round]{natbib}
\newcommandx{\bracketcite}[3][1,2]{\brackettext{\cite[#2][#1]{#3}}}	

\begin{document}

\begin{titlepage}

\title[Regulatory Capital Modelling for Credit Risk]{}
\maketitle

\thispagestyle{empty}
\begin{center}

\vspace*{40pt}
\textsc{\textbf{\Large Regulatory Capital Modelling for Credit Risk}}
\par\vspace{20pt}
{\large Marek Rutkowski\textsuperscript{~a} and Silvio Tarca\textsuperscript{~a,$\ast$}}
\par\vspace{10pt}
{\small\textsuperscript{a~}\textit{School of Mathematics and Statistics F07, University of Sydney, NSW 2006, Australia.}}
\par\vspace{20pt}

\rule{\linewidth}{0.25mm}
\begin{abstract}
\vspace{4pt}\noindent The Basel~II internal ratings-based (IRB) approach to capital adequacy for credit risk plays an important role in protecting the banking sector against insolvency.  We outline the mathematical foundations of regulatory capital for credit risk, and extend the model specification of the IRB approach to a more general setting than the usual Gaussian case.  It rests on the proposition that quantiles of the distribution of conditional expectation of portfolio percentage loss may be substituted for quantiles of the portfolio loss distribution.  We present a more compact proof of this proposition under weaker assumptions.  Then, constructing a portfolio that is representative of credit exposures of the Australian banking sector, we measure the rate of convergence, in terms of number of obligors, of empirical loss distributions to the asymptotic (infinitely fine-grained) portfolio loss distribution.  Moreover, we evaluate the sensitivity of credit risk capital to dependence structure as modelled by asset correlations and elliptical copulas.  Access to internal bank data collected by the prudential regulator distinguishes our research from other empirical studies on the IRB approach.

\vspace{8pt}\noindent \textsc{Keywords:} credit risk, regulatory capital, internal ratings-based (IRB) approach, asymptotic single risk factor (ASRF) model, credit value-at-risk (VaR), one-factor Gaussian copula, Student's $t$ copula.
\end{abstract}
\rule{\linewidth}{0.25mm}

\vspace{40pt}
{\textnormal June 2016}
\end{center}

\vspace{110pt}
\noindent\rule{0.33\textwidth}{0.10mm}
\par{\footnotesize$^{\ast\,}$Corresponding author.  Telephone: +61 (0)8 8313 4178.  Email: \texttt{silvio.tarca@adelaide.edu.au}}

\end{titlepage}

\section{Introduction}\label{sect_intro}
Risk capital models serve management functions including capital allocation, performance attribution, risk pricing, risk identification and monitoring, strategic business planning, and solvency assessment (i.e., capital adequacy).  However, model characteristics best suited for different purposes may conflict.  For example, models for solvency assessment require precision in the measurement of absolute risk levels under stressed economic conditions, whereas models for capital allocation need only be accurate in the measurement of relative risk under ``normal'' economic conditions \citep{JF25}.  While economic capital models may serve several management functions, the sole purpose of regulatory capital models is solvency assessment.
 
Under the Basel~II Accord, authorised deposit-taking institutions (ADIs) are required to assess capital adequacy for credit, market and operational risks.  ADIs determine regulatory capital for credit risk using either the standardised approach or, subject to approval, the internal rating-based (IRB) approach.  The latter is more expensive to administer, but usually produces lower regulatory capital requirements than the former.  As a consequence, ADIs using the IRB approach may deploy their capital in pursuit of more (profitable) lending opportunities.  The IRB approach implements the so-called asymptotic single risk factor (ASRF) model, an asset value factor model of credit risk.  This paper examines the model specification of the IRB approach, outlining its mathematical foundations and evaluating its robustness to a relaxation of model assumptions.  In relation to the latter, we undertake an empirical analysis of the Australian banking sector.  Its findings, though, are pertinent to other banking jurisdictions where regulatory capital charges are assessed under the IRB approach.  

In the context of evaluating model robustness, we briefly comment on the adoption of the IRB approach by the Australian banking sector.  Upon implementation of Basel~II in the first quarter of 2008, the Australian Prudential Regulation Authority (APRA) had granted the four largest Australian banks, designated ``major'' banks, approval to use the IRB approach to capital adequacy for credit risk.  The market dominance of the major banks, when coupled with the concentration of their regulatory capital assessed under the IRB approach, is indicative of the significance of the ASRF model in protecting the Australian banking sector against insolvency.  It motivates our interest in the sensitivity of model output to parameter variations and model misspecification.  APRA's support for this research includes access to internal bank data, which allows us to evaluate model robustness on a portfolio that is representative of the credit exposures of the Australian banking sector.  It distinguishes our research from other empirical studies on the IRB approach.

We begin in Section~\ref{sect_asrf_model} by deriving the theoretical foundations, drawn from the literature, of the model specification of the IRB approach.  An asset value factor model of credit risk, it has its roots in the classical structural approach of \citet{MRC74}.  Adapting the single asset model of \citeauthor{MRC74} to a portfolio of credits, \citet{VO02} derived a function that transforms unconditional default probabilities into default probabilities conditional on a single systematic risk factor.  We extend Vasicek's model of conditional independence to a more general setting, one not restricted to Gaussian processes.  \citet{GMB03} established that conditional on a single systematic risk factor, the portfolio percentage loss converges to its conditional expectation as the portfolio approaches asymptotic granularity --- no single credit exposure accounts for more than an arbitrarily small share of total portfolio exposure.  Then, assuming conditional independence given a single systematic risk factor, we derive a limiting form of the portfolio loss distribution for the general case.  The model specification of the IRB approach employs an analytical approximation of credit value-at-risk (VaR).  It rests on the proposition, due to \citet{GMB03}, that quantiles of the distribution of conditional expectation of portfolio percentage loss may be substituted for quantiles of the portfolio loss distribution.  We present a more compact proof of this proposition starting from weaker assumptions.

In generating a portfolio loss distribution we are, in effect, combining marginal loss distributions of constituent credits into a multivariate distribution capturing default dependence between obligors.  An approach to modelling default dependence, popularised by \citet{LDX00}, uses copula functions, which combine marginal distributions into a multivariate distribution with a chosen dependence structure.  Section~\ref{sect_cop_appr} derives the single-factor copula model describing default dependence for the general case.  Then, we deal with the special case of the one-factor Gaussian copula, the most commonly applied copula function in credit risk modelling.  Recognising that Gaussian distributions in financial applications tend to underestimate tail risk, we proceed to outline procedures for generating empirical loss distributions described by elliptical copulas, including Gaussian and Student's $t$ copulas.  Moreover, we illustrate the dependence induced by elliptical copulas.

In its implementation of the IRB approach to capital adequacy for credit risk, APRA requires that ADIs set aside provisions for absorbing expected losses, and hold capital against unexpected losses.  Assuming that portfolios are infinitely fine-grained so that idiosyncratic risk is fully diversified away, and a single systematic risk factor explains dependence across obligors, Section~\ref{sect_model_spec_irb} describes an analytical approximation for assessing ratings-based capital charges.  While real-world portfolios are not infinitely fine-grained, as a practical matter, credit portfolios of large banks adequately satisfy the asymptotic granularity condition, so it need not pose an impediment to assessing ratings-based capital charges.  Again, our contribution extends the model specification of the IRB approach to a more general setting than the usual Gaussian case.

Section~\ref{sect_empir_data} describes the Basel~II capital adequacy reporting of ADIs that supplies data for our empirical analysis.  In Section~\ref{sect_cvrg_asymp_dist} we measure the rate of convergence, in terms of number of obligors, of empirical loss distributions to the distribution of conditional expectation of portfolio percentage loss representing an infinitely fine-grained portfolio.  In the process we demonstrate that Gordy's proposition, which underpins the IRB approach, holds for a representative credit portfolio that exhibits sufficient granularity.  The IRB approach applies the one-factor Gaussian copula, in which default dependence is described by the matrix of pairwise correlations between obligors' asset values.  Section~\ref{sect_cap_sens_depend_struct} proceeds to evaluate the sensitivity of credit risk capital to dependence structure, as modelled by asset correlations and elliptical copulas.  We conclude by outlining the direction of future related research. 

\section{Foundations of the Asymptotic Single Risk Factor Model}\label{sect_asrf_model}
In this section we derive the theoretical foundations of the Basel~II IRB approach to capital adequacy for credit risk, and extend its model specification to a more general setting, one not restricted to Gaussian processes.  The IRB approach implements the so-called \textit{asymptotic single risk factor} (ASRF) model, an asset value factor model of credit risk.  Asset value models posit that default or survival of a firm depends on the value of its assets at (the end of) a given risk measurement horizon.  If the value of its assets falls below a critical threshold, its default point, the firm defaults, otherwise it survives.  Asset value models have their roots in Merton's seminal paper published in \citeyear{MRC74}.  Factor models are a well established, computationally efficient technique for explaining dependence between variables.

Define set $D_{i}$, abstractly, as the event that firm~$i$ defaults, and denote by $p_{i} = \P(D_{i})$ the \textit{unconditional probability of default} (PD) assigned to firm~$i$.  The standard model of asset values is geometric Brownian motion.  Let $W_{i}(t)$ be a Brownian motion, ${W_{i}(t) \sim \normdist{0}{t}}$, describing the variability in asset values of firm~$i$.  Then, the value of assets of firm~$i$ at time~$t$ may be expressed in logarithmic form as 
\begin{equation}\label{eqn_gbm_log_ss_i}
\log A_{i}(t) = \log A_{i}(0) + \mu_{i}t - \tfrac{1}{2}\sigma_{i}^{2}t + \sigma_{i}\sqrts{t}W_{i},
\end{equation}
where latent random $W_{i}(1)$, which we write $W_{i}$, is standard Gaussian, ${W_{i} \sim \normdist{0}{1}}$.  Note that $W_{i}(t)$ is distributionally equivalent to $\sqrts{t}W_{i}$ by the self-similarity property of Brownian motion.  A more precise definition of the event that firm~$i$ defaults follows from the postulate of asset value models of credit risk:
\begin{equation}\label{eqn_dflt_event_gauss}
	D_{i} = \left\{W_{i} < \Phi^{-1}(p_{i})\right\},
\end{equation}
where $\Phi^{-1}$ is the inverse standard Gaussian distribution function.  In the sequel we assume that unconditional PDs are published as market data by the likes of Moody's KMV and RiskMetrics.

\subsection{Conditional Independence Model}\label{sect_cond_indep_model}
Adapting the single asset model of Merton to a portfolio of credits, \citet{VO02} derived a function that transforms unconditional PDs into PDs conditional on a single systematic risk factor.  This function is the kernel of the model specification of the IRB approach.  Let the sequence of random variables $\{L_{n}\}$ be the percentage loss on a credit portfolio comprising $n\in\N$ obligors over a given risk measurement horizon $[0, \tau], \tau > 0$.  We make the assumption that the number of credits in the portfolio equals the number of distinct obligors, which can be achieved by aggregating multiple credits of an individual obligor into a single credit.  Suppose that \textit{exposure at default} (EAD) and \textit{loss given default} (LGD) are deterministic quantities, and denote by ${\ead_{i} \in \R_{+}}$ and ${\lgd_{i} \in [0,1]}$ the EAD and LGD, respectively, assigned to obligor~$i$.  Also, let $D_{i}$ be the event that obligor~$i$ defaults during the risk measurement horizon.  Then, the \textit{portfolio percentage loss} is given by
\begin{equation}\label{eqn_port_loss}
	L_{n} = \sum_{i = 1}^{n} w_{i}\lgd_{i}\one_{D_{i}},
\end{equation}
where $\one_{D_{i}}$ is the default indicator function, and $w_{i} = \ead_{i} / \sum_{j=1}^{n}\ead_{j}$ is the \textit{exposure weight} of obligor $i$ with $\sum_{i = 1}^{n} w_{i} = 1$.  Clearly, $w_{i}$ depends on $n$ and could be denoted $w_{i}(n)$, but we adopt the more concise, and more common, notation for exposure weight. 

Assume that latent random variables ${W_{1}, \ldots, W_{n}}$ modelling the variability in obligors' asset values are standard Gaussian and conditionally independent.  Suppose, too, that $W_{i}$ may be represented as
\begin{equation}\label{eqn_cond_indep_gauss_rvs}
	W_{i} = \sqrts{\rho_{i}}Y + \sqrts{1 - \rho_{i}}Z_{i},
\end{equation}
where random variables $Z_{1}, \ldots, Z_{n}$ and $Y$ are standard Gaussian and mutually independent, and ${\rho_{1}, \ldots, \rho_{n} \in (0,1)}$ are correlation parameters calibrated to market data.  Thus, ${W_{1}, \ldots, W_{n}}$ are conditionally independent given random variable~$Y$, which is common to all obligors.  Systematic risk factor~$Y$ may be interpreted as an underlying risk driver or economic factor, with each realisation describing a scenario of the economy.  Random variables $Z_{1}, \ldots, Z_{n}$ represent idiosyncratic, or obligor specific, risk.  Representation~\eqref{eqn_cond_indep_gauss_rvs}, which assumes that asset values are positively correlated, is taken from \citet{VO02} and employed in the model specification of the IRB approach.

\begin{remark}\label{rem_corr_asset_values}
Let the variability in obligors' asset values be described by \eqref{eqn_cond_indep_gauss_rvs}.  Then, the pairwise correlation between obligors' asset values ${\Corr{W_{i}}{W_{j}} = \sqrts{\rho_{i}\rho_{j}}}$.
\end{remark}

Substituting representation~\eqref{eqn_cond_indep_gauss_rvs} into set~\eqref{eqn_dflt_event_gauss} representing the event of default, the PD of obligor~$i$ conditional on realisation~$y \in \R$ of systematic risk factor~$Y$, or \textit{conditional probability of default}, may be expressed as
\begin{IEEEeqnarray*}{rCl}
	p_{i}(y) = \P(D_{i} \given Y=y) & = & \P\left( W_{i} <  \Phi^{-1}(p_{i}) \given Y=y \right) \\
		& = & \P\left(\sqrts{\rho_{i}}y + \sqrts{1 - \rho_{i}}Z_{i} < \Phi^{-1}(p_{i})\right)	\\
		& = & \P\left(Z_{i} < \frac{\Phi^{-1}(p_{i}) - \sqrts{\rho_{i}}y}{\sqrts{1-\rho_{i}}}\right) \\
		& = & \Phi\left(\frac{\Phi^{-1}(p_{i}) - \sqrts{\rho_{i}}y}{\sqrts{1-\rho_{i}}}\right),\IEEEyesnumber\label{eqn_cond_pd_gauss}
\end{IEEEeqnarray*}
where $p_{i}$ is the unconditional PD of obligor $i$.  Equation~\eqref{eqn_cond_pd_gauss} transforms unconditional PDs into PDs conditional on a single systematic risk factor.

Let 
\begin{equation}\label{eqn_zeta_y_gauss}
	\zeta_{i}(y) = \frac{\Phi^{-1}(p_{i}) - \sqrts{\rho_{i}}y}{\sqrts{1-\rho_{i}}} = \Phi^{-1}\big(p_{i}(y)\big)
\end{equation}
for $i = 1, \dots, n$.  Then, given $Y=y$, the portfolio percentage loss is calculated as
\begin{equation}\label{eqn_cond_port_loss}
	L_{n} = \sum_{i = 1}^{n} w_{i}\lgd_{i}\One{Z_{i} < \zeta_{i}(y)}.
\end{equation}

In the sequel we refer to the model developed in this section as the \textit{Gaussian conditional independence model of a credit portfolio}. 

\subsection{A More General Setting}\label{sect_gen_set}
We proceed to extend the Gaussian conditional independence model to a more general setting.  In Section~\ref{sect_cond_indep_model} we assume that: (\textit{i})~defaults are conditionally independent given a single systematic risk factor; and (\textit{ii})~obligors' asset values are modelled as geometric Brownian motion.  The latter assumption implies that latent random variables modelling the variability in obligors' asset values, and their component systematic and idiosyncratic risk factors, are Gaussian.  We now relax this assumption to describe a more general setting.

\begin{definition}\label{def_cond_indep_model_port_loss}
A \textit{conditional independence model of a credit portfolio} comprising $n\in\N$ obligors over a given risk measurement horizon $[0,\tau]$, $\tau > 0$, takes the form:
\begin{enumerate}[label=(\arabic*)]
	\item\label{item_cim_params}  Let $\ead_{i}\in\R$ be the EAD assigned to obligor~$i$, and $w_{i} = \ead_{i} / \sum_{j=1}^{n}\ead_{j}$ its exposure weight.  Let $\lgd_{i}\in[0,1]$, $\gamma_{i}\in(-1,1)$ and $p_{i}\in(0,1)$ be the LGD, asset correlation and unconditional PD, respectively, assigned to obligor~$i$.
	\item\label{item_cim_rvs}  Suppose that latent random variables $W_{1}, \ldots, W_{n}$ are conditionally independent, and admit representation
	\begin{equation}\label{eqn_cond_indep_rvs}
		W_{i} = \gamma_{i}Y + \sqrts{1 - \gamma_{i}^{2}}Z_{i},\IEEEyesnumber
	\end{equation}
	where $Z_{1}, \ldots, Z_{n}$ and $Y$ are mutually independent random variables.  Systematic risk factor~$Y$ is common to all obligors, while $Z_{1}, \ldots, Z_{n}$ represent idiosyncratic, or obligor specific, risk.  Denote by $F_{1}, \ldots, F_{n}$, $G_{1}, \ldots, G_{n}$ and $H$ the continuous and strictly increasing distribution functions of $W_{1}, \ldots, W_{n}$, $Z_{1}, \ldots, Z_{n}$ and $Y$, respectively.  Clearly, $F_{i}$ depends on $G_{i}$ and $H$ for $i = 1, \ldots, n$.
	\item\label{item_cim_dflt}  The event that obligor~$i$ defaults during the time interval $[0,\tau]$ is defined by the set
	\begin{equation}\label{eqn_dflt_event}
		D_{i} = \left\{W_{i} < F_{i}^{-1}(p_{i})\right\}
	\end{equation}
	for $i = 1, \ldots, n$.
\end{enumerate}
\end{definition}

Portfolio percentage loss~$L_{n}$ is calculated by \eqref{eqn_port_loss} where, for the general case, $\one_{D_{i}}$ is the indicator function of the default event defined by~\eqref{eqn_dflt_event}.  Recasting argument~\eqref{eqn_zeta_y_gauss} of the default indicator function and conditional probability function~\eqref{eqn_cond_pd_gauss} for the general case, we deduce a formula for portfolio percentage loss conditional on realisation~$y\in\R$ of systematic risk factor~$Y$.  Thus, given $Y=y$, portfolio percentage loss under the conditional independence model of Definition~\ref{def_cond_indep_model_port_loss} is calculated as:
\begin{equation}\label{eqn_cond_port_loss}
	L_{n} = \sum_{i = 1}^{n} w_{i}\lgd_{i}\One{Z_{i} < \zeta_{i}(y)},
\end{equation}
where
\begin{equation}\label{eqn_zeta_y}
	\zeta_{i}(y) = \frac{F_{i}^{-1}(p_{i}) - \gamma_{i}y}{\sqrts{1-\gamma_{i}^{2}}} = G_{i}^{-1}\big(p_{i}(y)\big),
\end{equation}
and
\begin{equation}\label{eqn_cond_pd}
	p_{i}(y) = \P(D_{i} \given Y=y) = G_{i}\left(\frac{F_{i}^{-1}(p_{i}) -\gamma_{i}y}{\sqrts{1-\gamma_{i}^{2}}}\right)
\end{equation}
for $i = 1, \ldots, n$.

\begin{remark}\label{rem_inv_cond_pd}
In the abstract case where conditional PD, as well as unconditional PD and asset correlation, are known and the state of the economy is sought, we take the inverse of~\eqref{eqn_cond_pd}.  Conditional probability function ${p_{i}\colon\R\rightarrow(0,1)}$, given by~\eqref{eqn_cond_pd}, is continuous and strictly decreasing in~$y$ --- conditional PD falls (respectively, rises) as the economy improves (deteriorates).  Hence, its inverse ${p_{i}^{-1}\colon(0,1)\rightarrow\R}$ is strictly decreasing too.  In particular, 
\begin{equation}\label{eqn_inv_cond_pd}
	y = p_{i}^{-1}(x) = \frac{F_{i}^{-1}(p_{i}) - \sqrts{1-\gamma_{i}^{2}}G_{i}^{-1}(x)}{\gamma_{i}}
\end{equation}
for all $x \in (0,1)$.
\end{remark}

\begin{remark}\label{rem_cond_pd_alpha}
Let $y = H^{-1}(1\!-\!\alpha)$, where $\alpha \in (0,1)$.  Then, the PD of obligor~$i$ conditional on $Y=y$ may be interpreted as the probability of default of obligor~$i$ is no greater than
\begin{equation}\label{eqn_cond_pd_alpha}
	\P\big(D_{i} \given Y=H^{-1}(1\!-\!\alpha)\big) = p_{i}\big(H^{-1}(1\!-\!\alpha)\big) = G_{i}\left(\frac{F_{i}^{-1}(p_{i}) - \gamma_{i}H^{-1}(1\!-\!\alpha)}{\sqrts{1-\gamma_{i}^{2}}}\right)
\end{equation}
in $(\alpha \times 100)\%$ of economic scenarios.
\end{remark}

\subsection{Conditional Expectation of Portfolio Percentage Loss}\label{sect_cond_port_loss}
The model specification of the IRB approach calculates the expectation of portfolio credit losses conditional on realisation~$y\in\R$ of systematic risk factor~$Y$.  Taking expectations of~\eqref{eqn_port_loss} and~\eqref{eqn_cond_port_loss}, respectively, we define the \textit{expected portfolio percentage loss} as
\begin{equation}\label{eqn_exp_port_loss}
	\E[L_{n}] = \sum_{i = 1}^{n} w_{i}\lgd_{i}p_{i},
\end{equation}
and the \textit{conditional expectation of portfolio percentage loss} as
\begin{equation}\label{eqn_cond_exp_port_loss}
	\E[L_{n} \given Y=y] = \sum_{i = 1}^{n} w_{i}\lgd_{i}p_{i}(y).
\end{equation}

\begin{remark}\label{rem_cond_exp_decr_y}
Conditional expectation function ${\E[L_{n} \given Y]\colon\R\rightarrow(0,1)}$, given by~\eqref{eqn_cond_exp_port_loss}, is continuous and strictly decreasing in~$y$ --- conditional expectation of portfolio percentage loss falls (respectively, rises) as the economy improves (deteriorates).
\end{remark}

Key propositions underpinning the model specification of the Basel~II IRB approach are derived for an asymptotic portfolio, often described as infinitely fine-grained, in which no single credit exposure accounts for more than an arbitrarily small share of total portfolio exposure.  Accordingly, our derivation of the model specification of the IRB approach requires a mathematically more precise definition of asymptotic granularity.

\begin{definition}\label{def_asymp_port}
Let $\EAD = \sum_{k=1}^{\infty}\ead_{k}$ be an infinite series whose terms $\ead_{k} \in \R_{+}$ represent EAD assigned to obligors constituting a credit portfolio:
\begin{enumerate}[label=(\arabic*)]
	\item  The partial sums of $\EAD$ of order $n$, are defined for $n \in \N$ as
		\begin{equation*}
			\EAD_{n} = \sum_{k=1}^{n}\ead_{k}.
		\end{equation*}
	\item  An \textit{asymptotic portfolio} satisfies
		\begin{equation*}
			\sum_{n=1}^{\infty} \left(\frac{\ead_{n}}{\EAD_{n}}\right)^{2} < \infty.
		\end{equation*}
\end{enumerate}
\end{definition}

\begin{remark}\label{rem_ead_wgt}
An application of Kronecker's lemma \citep[Lemma~6.5.1]{GA05} shows that exposure weights of credits constituting an asymptotic portfolio shrink very rapidly as the number of obligors increases:
\begin{equation}\label{eqn_wgts_shrink_to_zero}
	\sum_{k=1}^{\infty}\left(\frac{\ead_{k}}{\EAD_{k}}\right)^{2} < \infty \Longrightarrow \frac{1}{\EAD_{n}^{2}}\sum_{k=1}^{n}\ead_{k}^{2} = \sum_{k=1}^{n} w_{k}^{2} \rightarrow 0 \textnormal{ as } n\rightarrow\infty.
\end{equation}
\end{remark}

\begin{remark}\label{rem_asymp_cvrg_p_series}
Suppose that ${a \leq\ead_{k}\leq b}$ where ${0 < a \leq b < \infty}$ for all $k\in\N$.  Then,
\begin{equation}\label{eqn_asymp_ead_dvrg}
	\EAD_{n} = \sum_{k=1}^{n}\ead_{k} \geq na\rightarrow\infty \text{ as } n\rightarrow\infty,
\end{equation}
and
\begin{equation}\label{eqn_asymp_cvrg_p_series}
	\sum_{n=1}^{\infty} \left(\frac{\ead_{n}}{\EAD_{n}}\right)^{2} = \sum_{n=1}^{\infty} \frac{\ead_{n}^{2}}{\left(\sum_{k=1}^{n}\ead_{k}\right)^{2}} \leq \sum_{n=1}^{\infty} \frac{b^{2}}{(na)^{2}} = \frac{b^{2}}{a^{2}}\sum_{n=1}^{\infty} \frac{1}{n^{2}} < \infty,
\end{equation}
where \eqref{eqn_asymp_cvrg_p_series} converges by the p-series test \citep[see, e.g.,][Corollary~6.13]{WWR04}.  Assuming that EAD assigned to individual obligors is bounded, an entirely uncontroversial claim, total EAD of an asymptotic portfolio diverges, but \eqref{eqn_asymp_cvrg_p_series} satisfies the definition of an asymptotic portfolio \citep[Example~2.5.3]{BOW10}.
\end{remark}

\citet[Proposition~1]{GMB03} established that, conditional on a single systematic risk factor, the portfolio percentage loss converges, almost surely, to its conditional expectation as the portfolio approaches asymptotic granularity.

\begin{restatable}{proposition}{condportloss}\label{prop_cond_port_loss}
Assume a conditional independence model of an asymptotic credit portfolio.  Then,
\begin{equation}\label{eqn_cond_port_loss_lim}
	\lim_{n\rightarrow\infty} \left(L_{n} - \sum_{i = 1}^{n} w_{i}\lgd_{i}p_{i}(Y)\right) = 0, \quad\P\text{-a.s.}
\end{equation}
\end{restatable}

\begin{proof}
See Appendix~\ref{appx_cond_port_loss}.
\end{proof}

\begin{remark}\label{rem_port_loss_may_not_cvrg}
The sequence $\{L_{n}\}_{n\in\N}$ is bounded, but it is not monotone.  Indicator function ${\One{Z_{i} < \zeta_{i}(y)}\in\{0,1\}}$ and ${\lgd_{i} \in [0,1]}$ for ${i = 1, \ldots, n}$, and ${\sum_{i=1}^{n} w_{i} = 1}$ for all $n\in\N$, where $w_{i}$ depends on $n$.  The dependence of exposure weights on the number of obligors is clear when expressed as $w_{i} = \ead_{i} / \sum_{j=1}^{n} \ead_{j}$.  Consequently, $L_{n}$, and hence ${\sum_{i = 1}^{n} w_{i}\lgd_{i}p_{i}(y)}$, may not converge as ${n\rightarrow\infty}$.
\end{remark}

\begin{remark}\label{rem_depend_only_sys_risk}
In an asymptotic portfolio idiosyncratic risk is fully diversified away, so portfolio percentage loss~$L_{n}$ depends only on systematic risk factor~$Y$.  
\end{remark}

\begin{remark}\label{rem_port_loss_near_asymp}
As a practical matter, credit portfolios of large banks are typically near the asymptotic granularity of Definition~\ref{def_asymp_port}.  Thus, given ${Y=y}$, \eqref{eqn_cond_exp_port_loss} provides a statistically accurate estimate of percentage loss on a portfolio containing a large number of credits without concentration in a few names dominating the rest of the portfolio.
\end{remark}

\subsection{Limiting Form of Portfolio Loss Distribution}\label{sect_port_loss_dist}
Risk capital for a credit portfolio is determined from its parametric or empirical loss distribution.  Assuming conditional independence given a single systematic risk factor, \citet{VO02} derived the parametric loss distribution function of an asymptotic, homogeneous credit portfolio.  In Section~\ref{sect_cond_port_loss} we define an asymptotic portfolio, here we define a homogeneous portfolio.  

\begin{definition}\label{def_homo_port}
Assume a conditional independence model of a credit portfolio.  A \textit{homogeneous portfolio} comprising $n$ obligors satisfies:
\begin{enumerate}[label=(\arabic*)]
	\item  Parameters ${\gamma_{i} = \gamma}$ for ${i = 1, \ldots, n}$ in representation~\eqref{eqn_cond_indep_rvs} of latent random variables ${W_{1}, \ldots, W_{n}}$.
	\item  Random variables ${Z_{1}, \ldots, Z_{n}}$ are drawn from the same distribution described by the continuous and strictly increasing distribution function~$G$.  Also, denote by~$F$ the common distribution function of ${W_{1}, \ldots, W_{n}}$.
	\item Obligors are assigned the same unconditional PD and LGD, that is, ${p_{i} = p}$ and ${\lgd_{i} = \lgd}$ for ${i = 1, \ldots, n}$.
 \end{enumerate}
 \end{definition}

\begin{remark}\label{rem_cond_pd_homo}
From the properties of a homogeneous credit portfolio we infer that  ${p_{i}(y) = p(y)}$ for $i = 1, \ldots, n$, and all realisations $y \in \R$ of systematic risk factor~$Y$.
\end{remark}

The following result, derived for the general case, is a corollary of Proposition~\ref{prop_cond_port_loss}.

\begin{restatable}{corollary}{limlossdist}\label{cor_lim_loss_dist}
Assume a conditional independence model of an asymptotic, homogeneous credit portfolio.  Then,
\begin{equation}\label{eqn_cond_lim_port_loss}
	\lim_{n\rightarrow\infty}L_{n} = \lgd p(Y), \quad\P\text{-a.s.} 
\end{equation}
Accordingly, the portfolio loss distribution satisfies
\begin{equation}\label{eqn_lim_loss_dist}
	\lim_{n\rightarrow\infty} \P(L_{n} \leq \el) = 1 - H\left(\frac{F^{-1}(p) - \sqrts{1-\gamma^{2}}G^{-1}(\el/\lgd)}{\gamma}\right)
\end{equation}
for all $\el \in (0,1)$.
\end{restatable}

\begin{proof}
See Appendix~\ref{appx_lim_loss_dist}.
\end{proof}

Corollary~\ref{cor_lim_loss_dist} generalises the \citet{VO02} formulation of the loss distribution function of an asymptotic, homogeneous portfolio, which models default dependence as a multivariate Gaussian process.

\subsection{Credit Value-at-Risk}\label{sect_cred_var}
In determining regulatory capital, the Basel~II IRB approach applies a risk measure to assign a single numerical value to a random credit loss.  The chosen risk measure is value-at-risk (VaR), one of the most widely used measures in risk management.  VaR is an extreme quantile of a loss, or profit and loss, distribution that is rarely exceeded.  Firstly, we define quantiles of a distribution. 

\begin{definition}\label{def_qntl}
Let $X$ be a random variable, and let $\alpha \in (0,1)$.  Then the $\alpha$ \textit{quantile} of the distribution of $X$ is
\begin{equation}\label{eqn_qntl}
	\q{\alpha}{X} = \inf\{x\in\R \colon \P(X \leq x) \geq \alpha\}.
\end{equation}
\end{definition}

\begin{remark}\label{rem_qntl_cont_incr_dist}
If $X$ has a continuous and strictly increasing distribution function $F$, then the $\alpha$ quantile of the distribution of X is given by
\begin{equation}\label{eqn_qntl_cont_incr_dist}
	\q{\alpha}{X} = F^{-1}(\alpha),
\end{equation}
where $F^{-1}(\alpha)$, the inverse distribution function evaluated at $\alpha$, is the number ${\q{\alpha}{X}\in\R}$ such that ${F(\q{\alpha}{X})=\alpha}$.
\end{remark}

Adopting the convention that a loss is a positive number, we now define credit VaR.

\begin{definition}\label{def_cred_var}
\textit{Credit VaR} at the confidence level $\alpha \in (0,1)$ over a given risk measurement horizon is the largest portfolio percentage loss~$\el$ such that the probability of a loss~$L_{n}$ exceeding~$\el$ is at most $(1\!-\!\alpha)$:
\begin{equation}\label{eqn_cred_var}
	\VaR[\alpha]{L_{n}} = \inf\{\el\in\R \colon \P(L_{n}>\el) \leq 1\!-\!\alpha\}.
\end{equation}
\end{definition}

In probabilistic terms, $\VaR[\alpha]{L_{n}}$ is simply the $\alpha$ quantile of the portfolio loss distribution.  Although computationally expensive, Monte Carlo simulation is routinely employed to generate the empirical loss distribution and determine VaR of a credit portfolio.  Suppose that we generate the loss distribution of a credit portfolio comprising $n$ obligors by simulation of \eqref{eqn_cond_port_loss} parameterised by~\eqref{eqn_zeta_y}.  Let Monte Carlo simulation perform $N$ iterations.  For each iteration we draw from their respective distributions random variable~$Y$ representing systematic risk, and random variables ${Z_{1}, \ldots, Z_{n}}$ representing obligor specific risks.  Then, conditional on realisation~$y_{k} \in \R$ of systematic risk factor~$Y$ describing a scenario of the economy, the portfolio percentage loss over the risk measurement horizon is computed as
\begin{equation}\label{eqn_sim_port_loss}
	L_{n,k} = \sum_{i=1}^{n} w_{i}\lgd_{i}\One{Z_{i,k} < \zeta_{i}(y_{k})}
\end{equation}
for iterations ${k = 1, \ldots, N}$.  Monte Carlo simulation computes $N$ portfolio percentage losses constituting the empirical loss distribution described by the function \citep[pp. 30--32]{BOW10}:
\begin{equation}\label{eqn_sim_loss_dist}
	F(\el) = \frac{1}{N}\sum_{k=1}^{N}\One{0 \leq L_{n,k} \leq \el}.
\end{equation}

$\VaR[\alpha]{L_{n}}$, the $\alpha$ quantile of the empirical loss distribution, is the maximum credit loss at the $\alpha$ confidence level over a given risk measurement horizon.  Expected loss is estimated by calculating the average portfolio percentage loss over $N$ iterations of the simulation:
\begin{equation}\label{eqn_sim_exp_loss}
	\E[L_{n}] = \frac{1}{N}\sum_{k=1}^{N}L_{n,k}.
\end{equation}

An analytical model of the portfolio loss distribution, on the other hand, facilitates the fast calculation of credit VaR.  In the limiting case of an asymptotic, homogeneous credit portfolio, $\VaR[\alpha]{L_{n}}$ may be determined analytically from distribution function~\eqref{eqn_lim_loss_dist}.  However, the assumptions of Corollary~\ref{cor_lim_loss_dist} are too restrictive for real-world credit portfolios.  The risk factor model for ratings-based capital charges derived by \citet{GMB03} relaxes the homogeneity assumption.  His analysis proceeds assuming that:
\begin{enumerate}[label=(\arabic*)]
	\item\label{item_port_invar_asymp}  Portfolios are infinitely fine-grained so that idiosyncratic risk is fully diversified away.
	\item\label{item_port_invar_ssrf}  A single systematic risk factor explains dependence across obligors. 
\end{enumerate}
Under these weaker assumptions, and subject to additional technical conditions, Gordy established that quantiles of the distribution of conditional expectation of portfolio percentage loss may be substituted for quantiles of the portfolio loss distribution.  The statement and proof of Proposition~5 of \citet{GMB03}, which leads to an analytical approximation of credit VaR, is relegated to Appendix~\ref{appx_gordy_prop_5}.  Here, we present a version of this proposition that relaxes the additional technical conditions imposed by Gordy, resulting in a more compact, or parsimonious, proof.  

\begin{proposition}\label{prop_qntl_loss_dist_subst}
Assume a conditional independence model of a credit portfolio comprising $n$~obligors.  Denote by $\varphi_{n}(y)$ the conditional expectation function ${\E[L_{n} \given y]\colon\R\rightarrow(0,1)}$ given by~\eqref{eqn_cond_exp_port_loss}, and assume that the sequence $\{\varphi_{n}\}_{n\in\N}$ of real-valued functions satisfies:
\begin{enumerate}[label=(\arabic*)]
	\item\label{item_mono_func}  For every $n\in\N$, function $\varphi_{n}$ is strictly monotonic.
	\item\label{item_func_intv_eps_delta}  For every $y\in\R$ and $\varepsilon > 0$ there is a $\overline{\delta}(\varepsilon)\in\R\setminus\{0\}$ and $N(\overline{\delta}, \varepsilon)\in\N$ such that $n > N(\overline{\delta}, \varepsilon)$ implies
	\begin{equation*}
		\big[\varphi_{n}(y)-\varepsilon, \varphi_{n}(y)+\varepsilon\big] \subset \big[\varphi_{n}\big(y-\overline{\delta}(\varepsilon)\big), \varphi_{n}\big(y+\overline{\delta}(\varepsilon)\big)\big],
	\end{equation*}
	where $\overline{\delta}$ depends on $\varepsilon$, and $N$ depends on $\overline{\delta}$ and $\varepsilon$, in general.  While $\overline{\delta}$ depends on $\varepsilon$, and may also depend on $y$, we assume that it is independent of $n$.
	\item\label{item_cvrg_zero_eps_delta}  For every $\xi > 0$ there is an $\varepsilon > 0$ such that $0 < \big|\overline{\delta}(\varepsilon)\big| < \xi$, that is, $\overline{\delta}(\varepsilon)$ tends to zero as $\varepsilon$ tends to zero.
	\item\label{item_func_intv_delta_eps}  For every $y\in\R$ and $\varepsilon > 0$ there is a $\underline{\delta}(\varepsilon)\in\R\setminus\{0\}$ and $N(\underline{\delta}, \varepsilon)\in\N$ such that $n > N(\underline{\delta}, \varepsilon)$ implies
	\begin{equation*}
		\big[\varphi_{n}\big(y-\underline{\delta}(\varepsilon)\big), \varphi_{n}\big(y+\underline{\delta}(\varepsilon)\big)\big] \subset \big[\varphi_{n}(y)-\varepsilon, \varphi_{n}(y)+\varepsilon\big],
	\end{equation*}
	where $\underline{\delta}$ is independent of $n$.
\end{enumerate}
Then,
\begin{equation}\label{eqn_qntl_loss_dist_subst}
	\lim_{n\rightarrow\infty}\big(L_{n}-\varphi_{n}(y)\big) = 0,\ \P\text{-a.s.} \quad\Rightarrow\quad \lim_{n\rightarrow\infty}\big|\q{\alpha}{L_{n}}-\varphi_{n}(\q{1\!-\!\alpha}{Y})\big| = 0
\end{equation}
for all $\alpha\in(0,1)$.
\end{proposition}

\begin{remark}\label{rem_bnds_pd_corr}
We argue that Conditions~\ref{item_mono_func}--\ref{item_func_intv_delta_eps} of Proposition~\ref{prop_qntl_loss_dist_subst} are quite reasonable assumptions for real-world credit portfolios.  Observe that the conditional expectation function~$\varphi_{n}$, given by~\eqref{eqn_cond_exp_port_loss}, is continuous and strictly decreasing in~$y$ by Remark~\ref{rem_cond_exp_decr_y}.  Therefore, it satisfies Condition~\ref{item_mono_func}, which is also a condition of Lemma~\ref{lem_cvrg_seq_func_inv} and Corollary~\ref{cor_cvrg_dist_seq_func}, both used in the proof of Proposition~\ref{prop_qntl_loss_dist_subst}.  

For Conditions~\ref{item_func_intv_eps_delta}--\ref{item_func_intv_delta_eps} to hold, the curve of $\varphi_{n}$ cannot have horizontal or vertical segments in the neighbourhood of $\q{1\!-\!\alpha}{Y}$.  This is guaranteed by constraints $\overline{\delta}(\varepsilon)\in\R\setminus\{0\}$, $0 < \big|\overline{\delta}(\varepsilon)\big| < \xi$, and $\underline{\delta}(\varepsilon)\in\R\setminus\{0\}$.  By inspection of~\eqref{eqn_cond_pd}, Conditions~\ref{item_func_intv_eps_delta}--\ref{item_func_intv_delta_eps} are satisfied if $p_{i}$ and $\gamma_{i}^{2}$ are bounded away from zero and one for ${i = 1, \ldots, n}$.  Otherwise, $\varphi_{n}$ would no longer depend on $y$, thus violating Condition~\ref{item_mono_func}.  As a practical matter, if $p_{i}$ were equal to zero, then the capital charge assessed on credit~$i$ would be zero; and if $p_{i}$ were equal to one, then the product of EAD and LGD assigned to obligor~$i$ would be charged against profit and loss.  
\end{remark}

\begin{remark}\label{rem_cond_exp_func_diffbl}
If on some open interval $I$ containing $\q{1\!-\!\alpha}{Y}$, $\varphi_{n}$ were also differentiable on $I$, then Conditions~\ref{item_func_intv_eps_delta}--\ref{item_func_intv_delta_eps}  would be satisfied if
\begin{equation*}
	-\infty < -\ub \leq \varphi_{n}^{\prime}(y) \leq -\lb < 0
\end{equation*}
for all $y \in I$, with $\lb > 0$ and $\ub > 0$ independent of $n$.  Indeed, Proposition~\ref{prop_gordy_prop_5} \citep[Proposition~5]{GMB03} assumes that this condition holds on an open interval $I$ containing $\q{1\!-\!\alpha}{Y}$.
\end{remark}

Recall that $\VaR[\alpha]{L_{n}} = \q{\alpha}{L_{n}}$.  So, Proposition~\ref{prop_qntl_loss_dist_subst} asserts that the $\alpha$ quantile of the distribution of $\E[L_{n} \given Y]$, which is associated with the $(1\!-\!\alpha)$ quantile of the distribution of $Y$, may be substituted for the $\alpha$ quantile of the distribution of $L_{n}$ (i.e., credit VaR at the $\alpha$ confidence level over a given risk measurement horizon).  The IRB approach rests on Proposition~\ref{prop_qntl_loss_dist_subst}.  Its proof, presented below, relies on the following lemmas and corollary.  Note that in this section, ${F_{1}, \ldots, F_{n}}$ denote a sequence of distribution functions, as distinct from the notation adopted in Section~\ref{sect_gen_set}.

\begin{lemma}\label{lem_cvrg_seq_func_inv}
Let $\{g_{n}\}_{n\in\N}$ be a sequence of real-valued functions ${g_{n}\colon\R\rightarrow\R}$ that satisfies Conditions~\ref{item_mono_func}--\ref{item_cvrg_zero_eps_delta} of Proposition~\ref{prop_qntl_loss_dist_subst}.  Then, for every $b\in\R$,
\begin{equation}\label{eqn_cvrg_seq_func_inv}
	\lim_{n\rightarrow\infty}\big(a_{n} - g_{n}(b)\big) = 0 \quad\Rightarrow\quad \lim_{n\rightarrow\infty}g_{n}^{-1}(a_{n}) = b.
\end{equation}
\end{lemma}

\begin{proof}
The sufficient condition for the conclusion in~\eqref{eqn_cvrg_seq_func_inv} states that for every $\varepsilon > 0$ there is an $N_{0}(\varepsilon)\in\N$ such that $n > N_{0}(\varepsilon)$ implies
\begin{equation*}
	\big| a_{n} - g_{n}(b) \big| \leq \varepsilon,
\end{equation*}
which may be expressed as
\begin{equation*}
		a_{n} \in [g_{n}(b)-\varepsilon, g_{n}(b)+\varepsilon] \subset \big[g_{n}\big(b-\overline{\delta}(\varepsilon)\big), g_{n}\big(b+\overline{\delta}(\varepsilon)\big)\big],
\end{equation*}
where the subset relation holds when Condition~\ref{item_func_intv_eps_delta} is satisfied.  Observe that since $g_{n}$ is strictly monotonic for every $n\in\N$ by hypothesis, $g_{n}$ is one-to-one and ${g_n^{-1}\big(g_{n}(b)\big)=b}$.  Also, $\overline{\delta}(\varepsilon)>0$ if $g_{n}$ is strictly increasing, and $\overline{\delta}(\varepsilon)<0$ if $g_{n}$ is strictly decreasing.  Then, an application of the inverse function~$g_{n}^{-1}$ yields
\begin{equation*}
	g_{n}^{-1}(a_{n}) \in \big[b-\big|\overline{\delta}(\varepsilon)\big|, b+\big|\overline{\delta}(\varepsilon)\big|\big] \subset [b-\xi, b+\xi],
\end{equation*}
where the subset relation holds when Condition~\ref{item_cvrg_zero_eps_delta} is satisfied.  Choosing $N(\overline{\delta}, \varepsilon) = N_{0}(\varepsilon)$ such that $\xi$ is arbitrarily close to zero establishes the necessary condition of the hypothesis in~\eqref{eqn_cvrg_seq_func_inv}.
\end{proof}

\begin{corollary}\label{cor_cvrg_dist_seq_func}
Let $X_{n}$ and $Y$ be random variables defined on a common probability space with distribution functions $F_{n}$ and $H$, respectively.  If a sequence $\{g_{n}\}_{n\in\N}$ of real-valued functions ${g_{n}\colon\R\rightarrow\R}$ satisfies Conditions~\ref{item_mono_func}--\ref{item_cvrg_zero_eps_delta} of Proposition~\ref{prop_qntl_loss_dist_subst}, then
\begin{equation}\label{eqn_cvrg_seq_func_inv_rvs}
	\lim_{n\rightarrow\infty}\big(X_{n}-g_{n}(y)\big)=0,\ \P\text{-a.s.} \quad\Rightarrow\quad \lim_{n\rightarrow\infty}g_{n}^{-1}(X_{n})=Y.
\end{equation}
Moreover, if $H$ is continuous, then for every realisation~$y\in\R$ of $Y$,
\begin{equation}\label{eqn_cvrg_dist_seq_func_incr}
	\lim_{n\rightarrow\infty}F_{n}\big(g_{n}(y)\big) = H(y)
\end{equation}
when functions $g_{n}$ are strictly increasing, and
\begin{equation}\label{eqn_cvrg_dist_seq_func_decr}
	\lim_{n\rightarrow\infty}F_{n}\big(g_{n}(y)\big) = 1-H(y)
\end{equation}
when functions $g_{n}$ are strictly decreasing.
\end{corollary}

\begin{proof}
Random variables are real-valued functions on some probability space, so~\eqref{eqn_cvrg_seq_func_inv_rvs} is an immediate consequence of Lemma~\ref{lem_cvrg_seq_func_inv}.

The almost sure convergence of the sufficient condition for the conclusion in~\eqref{eqn_cvrg_seq_func_inv_rvs} implies pointwise convergence \citep[see, e.g.,][Definition 7.1]{WWR04} of the sequence of distribution functions of $g_{n}^{-1}(X_{n})$ to the distribution function of $Y$ at every point of continuity of $H$.  If $H$ is continuous, then convergence occurs for every realisation~$y\in\R$ of $Y$.  It follows from the necessary condition of the hypothesis in~\eqref{eqn_cvrg_seq_func_inv_rvs} that
\begin{equation*}
	\P\big(g_{n}^{-1}(X_{n}) \leq y\big) = \P\big(X_{n} \leq g_{n}(y)\big) = F_{n}\big(g_{n}(y)\big)
\end{equation*}
converges to ${\P(Y \leq y) = H(y)}$ as $n\rightarrow\infty$ if functions $g_{n}$ are strictly increasing, which establishes~\eqref{eqn_cvrg_dist_seq_func_incr}.  Similarly,
\begin{equation*}
	\P\big(g_{n}^{-1}(X_{n}) \leq y\big) = \P\big(X_{n} \geq g_{n}(y)\big) = 1 - \P\big(X_{n} \leq g_{n}(y)\big) = 1 - F_{n}\big(g_{n}(y)\big)
\end{equation*}
converges to ${\P(Y \leq y) = H(y)}$ as $n\rightarrow\infty$ if functions $g_{n}$ are strictly decreasing, which establishes~\eqref{eqn_cvrg_dist_seq_func_decr}.
\end{proof}

\begin{lemma}\label{lem_qntl_mono_func_rvs}
Let $Y$ be a random variable with continuous and strictly increasing distribution function~$H$, and let $g\colon\R\rightarrow\R$ be a strictly monotonic function.  Then, the $\alpha$~quantile of the distribution function of $g(Y)$ is
\begin{equation}\label{eqn_qntl_incr_func_rvs}
	\q{\alpha}{g(Y)} = g\big(H^{-1}(\alpha)\big)
\end{equation}
if $g$ is strictly increasing, and
\begin{equation}\label{eqn_qntl_decr_func_rvs}
	\q{\alpha}{g(Y)} = g\big(H^{-1}(1-\alpha)\big)
\end{equation}
if $g$ is strictly decreasing.
\end{lemma}

\begin{proof}
By Definition~\ref{def_qntl}, the $\alpha$ quantile of $g(Y)$ is
\begin{equation}\label{eqn_qntl_y}
	\q{\alpha}{g(Y)} = \inf\big\{g(y)\in\R \colon \P\big(g(Y) \leq g(y)\big) \geq \alpha\big\}.
\end{equation}
Since $H$ is continuous and strictly increasing by hypothesis, the $\alpha$ quantile of the distribution of~$Y$ is given by
\begin{equation}\label{eqn_qntl_cont_incr_y}
	\q{\alpha}{Y} = H^{-1}(\alpha),
\end{equation}
where $H^{-1}(\alpha)$ is the inverse distribution function evaluated at $\alpha$ (Remark~\ref{rem_qntl_cont_incr_dist}).  Observe that if $g$ is strictly increasing, then
\begin{IEEEeqnarray*}{rCl}
	\P\big(g(Y) \leq g(\q{\alpha}{Y})\big) & = & \P\big(g(Y) \leq g(H^{-1}(\alpha))\big)	\\
	& = & \P\big(Y \leq H^{-1}(\alpha)\big) \\
	& = & \alpha,
\end{IEEEeqnarray*}
where the first equality is a consequence of~\eqref{eqn_qntl_cont_incr_y}, and the second equality is the result of an application of inverse function $g^{-1}$.  Hence, \eqref{eqn_qntl_incr_func_rvs} follows from~\eqref{eqn_qntl_y}.  By a parallel argument, if $g$ is strictly decreasing, then
\begin{IEEEeqnarray*}{rCl}
	\P\big(g(Y) \leq g(\q{1-\alpha}{Y})\big) & = &  \P\big(g(Y) \leq g(H^{-1}(1\!-\!\alpha))\big)	\\
	& = & \P\big(Y \geq H^{-1}(1\!-\!\alpha)\big) \\
	& = & 1 - \P\big(Y \leq H^{-1}(1\!-\!\alpha)\big) \\
	& = & \alpha,
\end{IEEEeqnarray*}
which establishes~\eqref{eqn_qntl_decr_func_rvs}.
\end{proof}

\begin{proof}[Proof of Proposition~\ref{prop_qntl_loss_dist_subst}]
Fix $\alpha\in(0,1)$, set $\varphi_{n}(y) = \q{\alpha}{\varphi_{n}(Y)}$, and denote by $F_{n}$ the distribution function of $L_{n}$.  By appealing to results due to Corollary~\ref{cor_cvrg_dist_seq_func} and Lemma~\ref{lem_qntl_mono_func_rvs} for strictly decreasing functions $\varphi_{n}$, observe that
\begin{IEEEeqnarray*}{rCl}
	\lim_{n\rightarrow\infty}F_{n}\big(\q{\alpha}{\varphi_{n}(Y)}\big) & = & \lim_{n\rightarrow\infty}F_{n}\big(\varphi_{n}(H^{-1}(1\!-\!\alpha))\big) = \lim_{n\rightarrow\infty}F_{n}\big(\varphi_{n}(\q{1\!-\!\alpha}{Y})\big)	\\
	& = & 1 - H\big(\q{1\!-\!\alpha}{Y}\big) = 1 - H\big(H^{-1}(1\!-\!\alpha)\big) = \alpha.\IEEEyesnumber\label{eqn_cvrg_dist_decr_func_alpha}
\end{IEEEeqnarray*}
The first equality follows from~\eqref{eqn_qntl_decr_func_rvs}, the second from~\eqref{eqn_qntl_cont_incr_y}, the third from~\eqref{eqn_cvrg_dist_seq_func_decr}, and the fourth from~\eqref{eqn_qntl_cont_incr_y} again.  Notice that $\underline{\delta}(\varepsilon)>0$ if $\varphi_{n}$ is strictly increasing, and $\underline{\delta}(\varepsilon)<0$ if $\varphi_{n}$ is strictly decreasing.  By Remark~\ref{rem_cond_exp_decr_y}, $\varphi_{n}(y)=\E[L_{n} \given y]$ is strictly decreasing, thus satisfying Condition~\ref{item_mono_func}.  Then, on the basis of~\eqref{eqn_cvrg_dist_decr_func_alpha} and subject to Condition~\ref{item_func_intv_delta_eps}, lower and upper bounds, respectively, on the $\alpha$ quantile of $L_{n}$ are deduced:
\begin{equation*}
	\lim_{n\rightarrow\infty}F_{n}\big(\varphi_{n}(\q{1\!-\!\alpha}{Y}-\underline{\delta}(\varepsilon))\big) = 1-H\big(\q{1\!-\!\alpha}{Y}+\big|\underline{\delta}(\varepsilon)\big|\big) < \alpha,
\end{equation*}
and
\begin{equation*}
	\lim_{n\rightarrow\infty}F_{n}\big(\varphi_{n}(\q{1\!-\!\alpha}{Y}+\underline{\delta}(\varepsilon))\big) = 1-H\big(\q{1\!-\!\alpha}{Y}-\big|\underline{\delta}(\varepsilon)\big|\big) > \alpha,
\end{equation*}
which may be expressed as
\begin{equation*}
	\q{\alpha}{L_{n}} \in \big[\varphi_{n}\big(\q{1\!-\!\alpha}{Y}-\underline{\delta}(\varepsilon)\big), \varphi_{n}\big(\q{1\!-\!\alpha}{Y}+\underline{\delta}(\varepsilon)\big)\big].
\end{equation*}
Finally, by Condition~\ref{item_func_intv_delta_eps}, for every $\varepsilon>0$ there is a $\underline{\delta}(\varepsilon)\in\R$ and $N(\underline{\delta}, \varepsilon)\in\N$ such that $n > N(\underline{\delta}, \varepsilon)$ implies
\begin{equation*}
	\q{\alpha}{L_{n}} \in \big[\varphi_{n}\big(\q{1\!-\!\alpha}{Y}\big)-\varepsilon, \varphi_{n}\big(\q{1\!-\!\alpha}{Y}\big)+\varepsilon\big],
\end{equation*}
which establishes the necessary condition of the hypothesis in \eqref{eqn_qntl_loss_dist_subst}.
\end{proof}

\section{Copula Approach to Modelling Default Dependence}\label{sect_cop_appr}
In generating a portfolio loss distribution we are, in effect, combining marginal loss distributions of constituent credits into a multivariate distribution capturing default dependence between obligors.  An approach to modelling default dependence, popularised by \citet{LDX00}, uses copula functions --- a statistical technique for combining marginal distributions into a multivariate distribution with a chosen dependence structure.  

\subsection{Single-Factor Copula Model}\label{sect_sfc_model}
Consider a credit portfolio comprising $n$ obligors, and let \eqref{eqn_dflt_event} define the event that obligor~$i$ defaults.  We may express the unconditional PD of obligor~$i$, for the general case introduced in Section~\ref{sect_gen_set}, as
\begin{equation}\label{eqn_pd_cop}
	\P(D_{i}) = \P\big(W_{i} < F_{i}^{-1}(p_{i})\big),
\end{equation}
and the joint default probability as
\begin{equation}\label{eqn_joint_pd_cop}
	\P\big(\one_{D_{1}}=1, \ldots, \one_{D_{n}}=1\big) = \P\big(W_{1} < F_{1}^{-1}(p_{1}), \ldots, W_{n} < F_{n}^{-1}(p_{n})\big).
\end{equation}

In the sequel, $\overline{\R}$ denotes the extended real number line $[-\infty, \infty]$.  Let ${(u_{1}, \ldots, u_{n}) = (p_{1}, \ldots, p_{n})}$ be a vector in $[0,1]^{n}$, and ${(W_{1}, \ldots, W_{n})}$ a vector of latent random variables with continuous and strictly increasing distribution functions $F_{1}, \ldots, F_{n}$, respectively.  Suppose that $F$ is an $n$-dimensional distribution function with margins $F_{1}, \ldots, F_{n}$.  Then, by Sklar's theorem \citep[see, e.g.,][Theorem~2.10.9]{NRB06}, there is a unique $n$-copula~$C$ such that for all $(w_{1}, \ldots, w_{n}) \in \overline{\R}^{n}$,
\begin{equation}\label{eqn_multi_dist_cop}
	F(w_{1}, \ldots, w_{n}) = C\big(F_{1}(w_{1}), \ldots, F_{n}(w_{n})\big).
\end{equation}
Now, for any ${(u_{1}, \ldots, u_{n}) \in [0,1]^{n}}$,
\begin{IEEEeqnarray*}{rCl}
	C(u_{1}, \ldots, u_{n}) & = & F\big(F_{1}^{-1}(u_{1}), \ldots, F_{n}^{-1}(u_{n})\big) \\
	& = & \P\big(W_{1} < F_{1}^{-1}(u_{1}), \ldots, W_{n} < F_{n}^{-1}(u_{n})\big),\IEEEyesnumber\label{eqn_dflt_cop}
\end{IEEEeqnarray*}
by a corollary to Sklar's theorem \citep[see, e.g.,][Corollary~2.10.10]{NRB06}.

Assuming that defaults are conditionally independent given systematic risk factor~$Y$, latent random variables ${W_{1}, \ldots, W_{n}}$ may be represented as in \eqref{eqn_cond_indep_rvs}:
\begin{equation*}
	W_{i} = \gamma_{i}Y + \sqrts{1 - \gamma_{i}^{2}}Z_{i},
\end{equation*}
where $Z_{1}, \ldots, Z_{n}$ and $Y$ are mutually independent random variables with continuous and strictly increasing distribution functions $G_{1}, \ldots, G_{n}$ and $H$, respectively, and $\gamma_{i} \in (-1,1)$ for $i = 1, \ldots, n$.  

\begin{restatable}{lemma}{sfcmodel}\label{lem_sfc_model}
Assume a conditional independence model of a credit portfolio comprising $n$ obligors.  Then, default dependence may be described by the single-factor copula function associated with ${(W_{1}, \ldots, W_{n})}$:
\begin{equation}\label{eqn_sfc_model}
	C(u_{1}, \ldots, u_{n}) = \int_{-\infty}^{\infty} \left(\prod_{i=1}^{n} G_{i}\left( \frac{F_{i}^{-1}(u_{i}) - \gamma_{i}y}{\sqrts{1-\gamma_{i}^{2}}} \right)\right) \dee{H(y)}
\end{equation}
for any $(u_{1}, \ldots, u_{n}) \in [0,1]^{n}$.
\end{restatable}

\begin{proof}
See Appendix~\ref{appx_sfc_model_ofgc}.
\end{proof}

\begin{remark}\label{rem_sfc_model_sim}
Monte Carlo simulation computes the portfolio percentage loss for a very large sample of realisations $y \in \R$ of systematic risk factor~$Y$.  For each realisation~$y$ we use \eqref{eqn_cond_port_loss} to determine the number of defaults $k$, and calculate portfolio percentage loss by summing the product of exposure weight and LGD for credits that have defaulted.  This simulation procedure is discussed in some detail in Section~\ref{sect_cvrg_asymp_dist}.  In generating the empirical loss distribution of a credit portfolio by simulation~\eqref{eqn_sim_port_loss}, we are implicitly applying the single-factor copula model using Monte Carlo methods.
\end{remark}

\subsection{One-Factor Gaussian Copula}\label{sect_ofgc}
Throughout this section, as the title implies, we restrict our attention to the special case in which default dependence is modelled as a multivariate Gaussian process.  Consider a credit portfolio comprising $n$ obligors.  Substituting the inverse standard Gaussian distribution function into~\eqref{eqn_pd_cop} and \eqref{eqn_joint_pd_cop}, the unconditional PD of obligor~$i$ becomes
\begin{equation}\label{eqn_pd_gauss_cop}
	\P(D_{i}) = \P\big(W_{i} < \Phi^{-1}(p_{i})\big),
\end{equation}
and the joint default probability is given by
\begin{equation}\label{eqn_joint_pd_gauss_cop}
	\P\big(\one_{D_{1}}=1, \ldots, \one_{D_{n}}=1\big) = \P\big(W_{1} < \Phi^{-1}(p_{1}), \ldots, W_{n} < \Phi^{-1}(p_{n})\big).
\end{equation}
Let ${(u_{1}, \ldots, u_{n}) = (p_{1}, \ldots, p_{n})}$ be a vector in $[0,1]^{n}$, and choose a dependence structure described by correlation matrix $\Gamma$.  Then, the unique Gaussian copula associated with ${(W_{1}, \ldots, W_{n})}$ is a particular case of~\eqref{eqn_dflt_cop}:
\begin{IEEEeqnarray*}{rCl}
	C_{\Gamma}(u_{1}, \ldots, u_{n}) & = & \Phi_{\Gamma}\big(\Phi^{-1}(u_{1}), \ldots, \Phi^{-1}(u_{n})\big)\IEEEyesnumber\label{eqn_gauss_cop} \\
	& = & \P\big(W_{1} < \Phi^{-1}(u_{1}), \ldots, W_{n} < \Phi^{-1}(u_{n})\big), 
\end{IEEEeqnarray*}
for any ${(u_{1}, \ldots, u_{n}) \in [0,1]^{n}}$, where $\Phi_{\Gamma}$ is the multivariate standard Gaussian distribution function with correlation matrix $\Gamma$.

Now suppose that defaults are conditionally independent given a single systematic risk factor.  Then, default dependence is described by correlation matrix
\begin{equation}\label{eqn_ofgc_gamma}
	\widehat{\Gamma} = 
	\begin{pmatrix*}[c]
		1					&  \sqrts{\rho_{1}\rho_{2}} & \cdots	& \sqrts{\rho_{1}\rho_{n}}  \\
		\sqrts{\rho_{1}\rho_{2}}	& 1					& \cdots	& \sqrts{\rho_{2}\rho_{n}} \\
		\vdots				& \vdots 				& \ddots	& \vdots \\
		\sqrts{\rho_{1}\rho_{n}}	& \sqrts{\rho_{2}\rho_{n}}	& \cdots	& 1 \\ 
	\end{pmatrix*},
\end{equation}
where $\sqrts{\rho_{i}\rho_{j}} = \Corr{W_{i}}{W_{j}}$ is the pairwise correlation between obligors' asset values (Remark \ref{rem_corr_asset_values}), and $\sqrts{\rho_{i}}$ is the exposure of obligor~$i$ to systematic risk factor~$Y$ in \eqref{eqn_cond_indep_gauss_rvs}.  This special case of the Gaussian copula is the so-called \textit{one-factor Gaussian copula}, the most commonly applied copula function in credit risk modelling \citep{MS12}.  The following result, a corollary of Lemma~\ref{lem_sfc_model}, provides an expression for the one-factor Gaussian copula.

\begin{restatable}{corollary}{ofgc}\label{cor_ofgc}
Assume a Gaussian conditional independence model of a credit portfolio comprising $n$ obligors with pairwise asset correlations defined by matrix~\eqref{eqn_ofgc_gamma}.  Then, default dependence may be described by the one-factor Gaussian copula associated with ${(W_{1}, \ldots, W_{n})}$: 
\begin{IEEEeqnarray*}{rCl}
	C_{\widehat{\Gamma}}(u_{1}, \ldots, u_{n}) & = & \Phi_{\widehat{\Gamma}}\big(\Phi^{-1}(u_{1}), \ldots, \Phi^{-1}(u_{n})\big) \\
	& = &\int_{-\infty}^{\infty} \left(\prod_{i=1}^{n} \Phi\left( \frac{\Phi^{-1}(u_{i}) - \sqrts{\rho_{i}}y}{\sqrts{1-\rho_{i}}} \right)\right) \phi(y) \dee{y}\IEEEyesnumber\label{eqn_ofgc}
\end{IEEEeqnarray*}
for any $(u_{1}, \ldots, u_{n}) \in [0,1]^{n}$.
\end{restatable}

\begin{proof}
See Appendix~\ref{appx_sfc_model_ofgc}.
\end{proof}

\begin{remark}\label{rem_irb_ofgc}
Appealing to Proposition~\ref{prop_qntl_loss_dist_subst}, the Basel~II IRB approach applies the one-factor Gaussian copula to calculate the $\alpha$ quantile of the distribution of $\E[L_{n} \given Y]$, an analytical approximation of the $\alpha$ quantile of the distribution of $L_{n}$, or $\VaR[\alpha]{L_{n}}$.
\end{remark}

\subsection{Elliptical Copulas}\label{sect_ellip_cop}
In applying the one-factor Gaussian copula to calculate regulatory capital for credit risk, the IRB approach implicitly assumes that a multivariate Gaussian distribution accurately models tail risk of credit portfolios.  But, it is generally acknowledged that models which assume that financial data follow a Gaussian distribution tend to underestimate tail risk.  For one, Gaussian copulas do not exhibit tail dependence --- the tendency for extreme observations (i.e., credit defaults) to occur simultaneously for all random variables.  Under a Gaussian copula defaults are said to be \textit{asymptotically independent} in the upper tail \citep{EMS02}.  Section~\ref{sect_sens_ellip_cop} examines the effect of tail dependence by measuring the sensitivity of credit risk capital to dependence structure as modelled by elliptical copulas, including Gaussian and Student's $t$ copulas.  We abbreviate the latter by $t$-copula.
 
Let ${(X_{1}, \ldots, X_{n})}$ be a vector of latent random variables modelling default dependence of a portfolio comprising $n$~obligors. We proceed to illustrate the dependence induced by a variety of elliptical copulas, and outline procedures for randomly generating observations drawn from the resultant multivariate distributions \citep[pp.~106--108]{BOW10}:
\begin{itemize}
	\item  \textit{One-factor Gaussian copula}.  Observations ${X_{1}, \ldots, X_{n}}$ are randomly generated with 
		\begin{equation}\label{eqn_gauss_cop_rvs}
			X_{i} = \sqrts{\rho_{i}}Y + \sqrts{1 - \rho_{i}}Z_{i},
		\end{equation}
		where random variables $Z_{1}, \ldots, Z_{n}$ and $Y$ are standard Gaussian and mutually independent, and correlation parameters ${\rho_{1}, \ldots, \rho_{n} \in (0,1)}$.  That is, we sample ${(X_{1}, \ldots, X_{n})}$ from the distribution induced by Gaussian copula~\eqref{eqn_gauss_cop} with correlation matrix~\eqref{eqn_ofgc_gamma}.  Note that  \eqref{eqn_gauss_cop_rvs} is simply the conditionally independent representation expressed in~\eqref{eqn_cond_indep_gauss_rvs}. 
	
	\item  \textit{Product copula with Gaussian margins}.  The product copula generates independent, and therefore uncorrelated, standard Gaussian random variables \citep[Theorem~8.2.5]{ELM03}.  So, observations ${X_{1}, \ldots, X_{n}}$ are independently drawn from the standard Gaussian distribution.  That is, we sample ${(X_{1}, \ldots, X_{n})}$ from the distribution induced by Gaussian copula~\eqref{eqn_gauss_cop} with correlation matrix $I_{n}$, the $n$-by-$n$ identity matrix.
	
	\item  \textit{$t$-copula with $\nu$ degrees of freedom and $t$-distributed margins}.  Observations ${X_{1}, \ldots, X_{n}}$ are randomly generated with
		\begin{equation}\label{eqn_t_cop_rvs}
			X_{i} = \sqrts{\frac{\nu}{V}} \left(\sqrts{\rho_{i}}Y + \sqrts{1 - \rho_{i}}Z_{i}\right),
		\end{equation}
where ${Z_{1}, \ldots, Z_{n}}$ and ${Y \sim \normdist{0}{1}}$, ${V \sim \chi^{2}(\nu)}$, and ${Z_{1}, \ldots, Z_{n}}$, $Y$ and $V$ are mutually independent.  Scaling~\eqref{eqn_cond_indep_gauss_rvs} by $\sqrts{\nu / V}$ transforms standard Gaussian random variables into $t$-distributed random variables with $\nu$ degrees of freedom.  Vector ${(X_{1}, \ldots, X_{n})}$ inherits correlation matrix~\eqref{eqn_ofgc_gamma}.

	\item  \textit{$t$-copula with $\nu$ degrees of freedom and Gaussian margins}.  Observations ${X_{1}, \ldots, X_{n}}$ are randomly generated with
		\begin{equation}\label{eqn_t_cop_gauss_rvs}
			X_{i} = \Phi^{-1}\left(\Phi_{\nu}\left(\sqrts{\frac{\nu}{V}} \left(\sqrts{\rho_{i}}Y + \sqrts{1 - \rho_{i}}Z_{i}\right)\right)\right),
		\end{equation}
where $\Phi^{-1}$ is the inverse standard Gaussian distribution function, and $\Phi_{\nu}$ is the Student's $t$ distribution function with $\nu$ degrees of freedom.
\end{itemize}

\begin{figure}[!t]
	\centering
	\caption{Bivariate scatter plots illustrate the dependence induced by a variety of elliptical copulas.  Each point corresponds to an ordered pair $(X_{1},X_{2})$.  Except for the product copula where $X_{1}$ and $X_{2}$ are uncorrelated, ${\rho_{1} = \rho_{2} = 0.170}$, the exposure-weighted average asset correlation of the representative credit portfolio described in Section~\ref{sect_empir_data}.}
	\label{fig_bivar_ellip_cop}
	\scalebox{0.50}{
	\begin{tabular}{ccc}
		\input{plot_gauss_cop}	& \input{plot_t30_cop_norm}	& \input{plot_t10_cop_norm}	\\
		\input{plot_prod_cop}	& \input{plot_t3_cop_norm}	& \input{plot_t10_cop}
	\end{tabular}
	}
\end{figure}
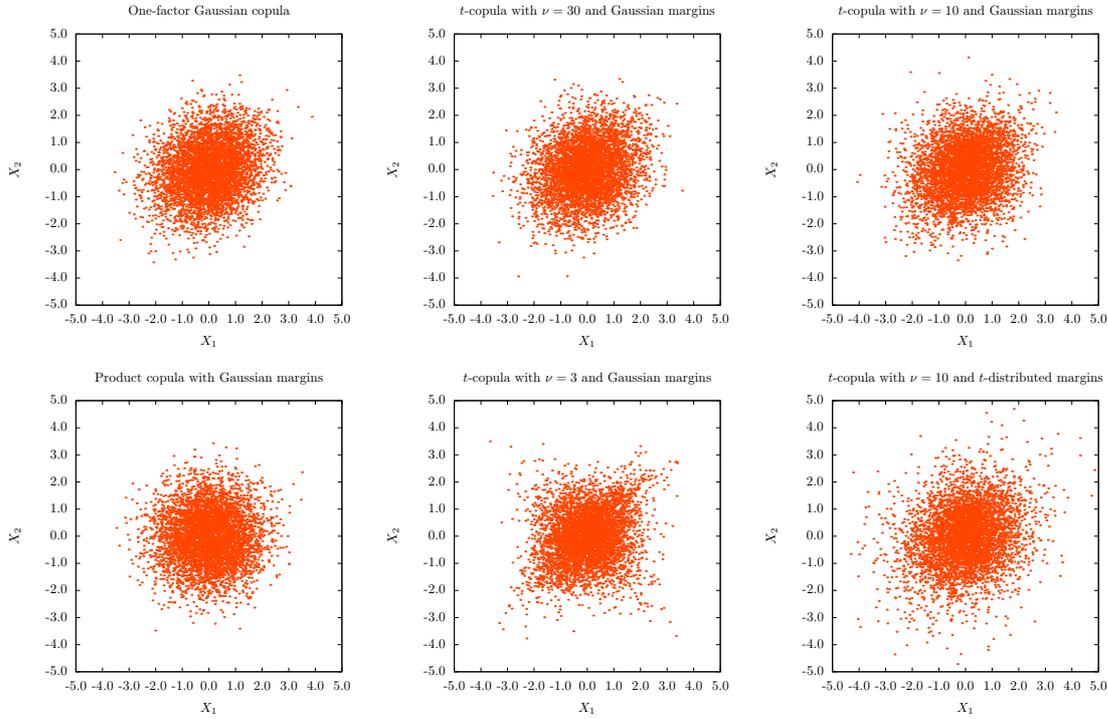

The bivariate scatter plots in Figure~\ref{fig_bivar_ellip_cop} illustrate the dependence induced by the elliptical copulas described above.  For each copula, except the product copula, we set ${\rho_{1} = \rho_{2} = 0.170}$, the exposure-weighted asset correlation of the representative credit portfolio described in Section~\ref{sect_empir_data}.

\begin{remark}\label{rem_ellip_cop}
The dependence exhibited by the one-factor Gaussian copula becomes apparent when compared with the product copula, which generates uncorrelated standard Gaussian random variables.  In contrast to Gaussian copulas, $t$-copulas admit tail dependence with fewer degrees of freedom producing stronger dependence.  When the $t$-copula is applied to combine Gaussian margins and $t$-distributed margins, respectively, the former is more tightly distributed.
\end{remark}

Assuming that asset values follow a log-normal distribution, the distribution of ${(X_{1}, \ldots, X_{n})}$ is determined by the copula function chosen to combine its margins.  Section~\ref{sect_cred_var} outlines the procedure for generating the loss distribution of a credit portfolio comprising $n$ obligors by simulation of \eqref{eqn_cond_port_loss}.  In an implementation of the one-factor Gaussian copula, the default indicator function of \eqref{eqn_cond_port_loss} is parameterised by~\eqref{eqn_zeta_y_gauss}.  Monte Carlo simulation performs $N$~iterations, \eqref{eqn_sim_port_loss} calculates the portfolio percentage loss for each iteration, and \eqref{eqn_sim_loss_dist} describes the empirical loss distribution. 

In relation to the $t$-copula with Gaussian margins, we continue to assume that unconditional PDs scaled to a given risk measurement horizon are published as market data.  Substituting~\eqref{eqn_t_cop_rvs} into~\eqref{eqn_dflt_event}, we define the event that obligor~$i$ defaults during the risk measurement horizon by the set
\begin{equation}\label{eqn_dflt_event_t_cop}
	D_{i} = \left\{ \sqrts{\frac{\nu}{V}} \left(\sqrts{\rho_{i}}Y + \sqrts{1 - \rho_{i}}Z_{i}\right) < \Phi_{\nu}^{-1}(p_{i}) \right\}.
\end{equation}
Then, the PD of obligor~$i$ conditional on $Y = y$ and $V = v$ is deducible \citep[pp.~109--111]{BOW10}:
\begin{IEEEeqnarray*}{rCl}
	p_{i}(y, v) & = & \P\left( \sqrts{\frac{\nu}{V}} \left(\sqrts{\rho_{i}}Y + \sqrts{1 - \rho_{i}}Z_{i}\right) < \Phi_{\nu}^{-1}(p_{i}) \given Y = y, V = v \right) \\
	& = & \P\left(Z_{i} < \frac{\sqrts{v/\nu}\Phi_{\nu}^{-1}(p_{i}) - \sqrts{\rho_{i}}y}{\sqrts{1-\rho_{i}}}\right) \\
	& = & \Phi\left( \frac{\sqrts{v/\nu}\Phi_{\nu}^{-1}(p_{i}) - \sqrts{\rho_{i}}y}{\sqrts{1-\rho_{i}}} \right).\IEEEyesnumber\label{eqn_cond_pd_t_cop}
\end{IEEEeqnarray*}
Let
\begin{equation}\label{eqn_zeta_y_t_cop}
	\zeta_{i}(y, v) = \frac{\sqrts{v/\nu}\Phi_{\nu}^{-1}(p_{i}) - \sqrts{\rho_{i}}y}{\sqrts{1-\rho_{i}}}
\end{equation}
for $i = 1, \dots, n$.  Now, given $Y = y$ and $V = v$, the portfolio percentage loss is calculated as
\begin{equation}\label{eqn_port_loss_t_cop}
	L_{n} = \sum_{i = 1}^{n} w_{i}\lgd_{i}\One{Z_{i} < \zeta_{i}(y, v)}.
\end{equation}

Suppose that we generate the loss distribution of a portfolio comprising $n$ obligors by simulation of \eqref{eqn_port_loss_t_cop}, an implementation of the $t$-copula with Gaussian margins.  Let Monte Carlo simulation perform $N$~iterations.  For each iteration we draw from the standard Gaussian distribution random variables ${Z_{1}, \ldots, Z_{n}}$ and $Y$, and from the chi-square distribution with $\nu$ degrees of freedom random variable~$V$.  Then, given $Y = y_{k}$ and $V = v_{k}$, the portfolio percentage loss over the risk measurement horizon is computed as
\begin{equation}\label{eqn_sim_port_loss_t_cop}
	L_{n,k} = \sum_{i=1}^{n} w_{i}\lgd_{i}\One{Z_{i,k} < \zeta_{i}(y_{k}, v_{k})}
\end{equation}
for iterations ${k = 1, \ldots, N}$.  Again, \eqref{eqn_sim_loss_dist} describes the empirical loss distribution.

\section{Model Specification of the Internal Ratings-Based Approach}\label{sect_model_spec_irb}
Under the Basel~II Accord \citep{BCBS128}, ADIs assess capital adequacy for credit risk using either the standardised approach or, subject to approval, the IRB approach.  Our concern is with the theoretical foundations and empirical analysis of the latter approach.  In keeping with the Basel~II IRB approach to capital adequacy for credit risk, the relevant prudential standard of APRA \citeyearpar{APS113} requires that ADIs set aside provisions for absorbing expected losses, and hold capital against unexpected losses.

\begin{definition} \label{def_unexp_loss}
\textit{Unexpected loss} on a credit portfolio at the $\alpha$ confidence level over a given risk measurement horizon is the difference between credit VaR (with the same confidence level and time horizon) and expected loss.
\end{definition}

Recall that $L_{n}$ denotes the portfolio percentage loss on a credit portfolio comprising $n$ obligors.  Then, $\VaR[\alpha]{L_{n}}$ denotes the portfolio percentage loss at the $\alpha$ confidence level over a given risk measurement horizon, and $\E[L_{n}]$ the expected portfolio percentage loss.  We define credit risk capital consistent with the IRB approach of the Basel~II Accord and the relevant prudential standard of APRA.

\begin{definition} \label{def_cap_cred_asymp}
Let \textit{credit risk capital} be held against unexpected losses.  Then,
\begin{equation}\label{eqn_cap_pct_asymp}
	K_{\alpha}(L_{n}) = \VaR[\alpha]{L_{n}} - \E[L_{n}]
\end{equation}
is the capital charge (at the $\alpha$ confidence level over a given risk measurement horizon) as a percentage of EAD on a credit portfolio comprising $n$ obligors.
\end{definition}

Appealing to Definition~\ref{def_cap_cred_asymp} and Proposition~\ref{prop_qntl_loss_dist_subst}, we deduce a function for calculating capital held against unexpected losses on an asymptotic credit portfolio.

\begin{restatable}{proposition}{capcredasymp}\label{prop_cap_cred_asymp}
Assume a conditional independence model of an asymptotic credit portfolio.  Then,
\begin{equation}\label{eqn_cap_cred_asymp}
	\lim_{n\rightarrow\infty}K_{\alpha}(L_{n}) = \lim_{n\rightarrow\infty}\sum_{i = 1}^{n} w_{i}\lgd_{i}G_{i}\left(\frac{F_{i}^{-1}(p_{i}) - \gamma_{i}H^{-1}(1\!-\!\alpha)}{\sqrts{1-\gamma_{i}^{2}}}\right) - \lim_{n\rightarrow\infty}\sum_{i = 1}^{n} w_{i}\lgd_{i}p_{i},
\end{equation}
assuming that the limits on the right-hand side of~\eqref{eqn_cap_cred_asymp} are well defined.
\end{restatable}

\begin{proof}
See Appendix~\ref{appx_cap_cred_asymp}.
\end{proof}

Proposition~\ref{prop_cap_cred_asymp} applies to the abstract case of an infinitely fine-grained credit portfolio for which total EAD diverges to $\infty$ (Remark~\ref{rem_asymp_cvrg_p_series}).  However, in the real world a credit portfolio contains a finite number of obligors, each assigned a positive and finite EAD.  Hence, portfolio EAD is positive and finite, and so is the assessed capital charge, which may be expressed as a percentage of EAD.  Remark~\ref{rem_port_loss_near_asymp} highlights that while real-world portfolios are not infinitely fine-grained, as a practical matter, credit portfolios of large banks are typically near the asymptotic granularity of Definition~\ref{def_asymp_port}.  For the practical application of Proposition~\ref{prop_cap_cred_asymp} we redefine credit risk capital for finite portfolios that exhibit ``sufficient'' granularity.

\begin{definition}\label{def_cap_cred_risk}
Let \textit{credit risk capital} be held against unexpected losses.  Then, in practice, the capital charge (at the $\alpha$ confidence level over a given risk measurement horizon) as a percentage of EAD on a credit portfolio comprising $n$ obligors that ``adequately'' satisfies the asymptotic granularity condition of Definition~\ref{def_asymp_port} may be assessed as
\begin{equation}\label{eqn_cap_pct}
	\Kap[\alpha]{L_{n}} = \E\big[L_{n} \given Y=H^{-1}(1\!-\!\alpha)\big] - \E[L_{n}].
\end{equation}
\end{definition}

Expanding~\eqref{eqn_cap_pct}, the capital charge as a percentage of EAD on a credit portfolio containing a finite number of obligors, $n$, that exhibits sufficient granularity is calculated as
\begin{IEEEeqnarray*}{rCl}
	\Kap[\alpha]{L_{n}} & = & \sum_{i = 1}^{n} w_{i}\lgd_{i}p_{i}\big(H^{-1}(1\!-\!\alpha)\big) - \sum_{i = 1}^{n} w_{i}\lgd_{i}p_{i}	\\
	& = &\sum_{i = 1}^{n} w_{i}\lgd_{i}G_{i}\left(\frac{F_{i}^{-1}(p_{i}) - \gamma_{i}H^{-1}(1\!-\!\alpha)}{\sqrts{1-\gamma_{i}^{2}}}\right) - \sum_{i = 1}^{n} w_{i}\lgd_{i}p_{i}.\IEEEyesnumber\label{eqn_cap_cred_risk}
\end{IEEEeqnarray*}
The ASRF model developed by the Basel Committee on Banking Supervision (BCBS) models default dependence as a multivariate Gaussian process.  Recasting~\eqref{eqn_cap_cred_risk} for the Gaussian case, the capital charge (at the $\alpha$ confidence level over a given risk measurement horizon) on a near asymptotic portfolio comprising $n$ obligors is calculated as
\begin{IEEEeqnarray*}{rCl}
	\Kap[\alpha]{L_{n}} & = & \sum_{i = 1}^{n} w_{i}\lgd_{i}\Phi\left(\frac{\Phi^{-1}(p_{i}) - \sqrts{\rho_{i}}\Phi^{-1}(1\!-\!\alpha)}{\sqrts{1-\rho_{i}}}\right) - \sum_{i = 1}^{n} w_{i}\lgd_{i}p_{i} \\
		& = & \sum_{i = 1}^{n} w_{i}\lgd_{i}\Phi\left(\frac{\Phi^{-1}(p_{i}) + \sqrts{\rho_{i}}\Phi^{-1}(\alpha)}{\sqrts{1-\rho_{i}}}\right) - \sum_{i = 1}^{n} w_{i}\lgd_{i}p_{i}.\IEEEyesnumber\label{eqn_cap_cred_gauss}
\end{IEEEeqnarray*} 
The second equality follows from the symmetry of the standard Gaussian density function.  Note that the kernel of ASRF model~\eqref{eqn_cap_cred_gauss} transforms unconditional PDs into PDs conditional on systematic risk factor~$Y$ using~\eqref{eqn_cond_pd_gauss}.

Under the Basel~II IRB approach regulatory capital is determined at the $99.9\%$ confidence level over a one-year horizon --- a $0.1\%$ probability that credit losses will exceed provisions and capital over the subsequent year.  In practice, it incorporates a maturity adjustment to account for the greater likelihood of downgrades for longer-term claims, the effects of which are stronger for claims with higher credit ratings.  We omit the maturity adjustment from our mathematical derivation and empirical study of the ASRF model.  Thus, regulatory capital for a near asymptotic credit portfolio comprising $n$ obligors is assessed as
\begin{IEEEeqnarray*}{rCl}
	\Kap[99.9\%]{L_{n}}
	& = & \sum_{i=1}^{n} w_{i}\lgd_{i}\Phi\left(\frac{\Phi^{-1}(p_{i}) + \sqrts{\rho_{i}}\Phi^{-1}(0.999)}{\sqrts{1-\rho_{i}}}\right) - \sum_{i=1}^{n} w_{i}\lgd_{i}p_{i}.\IEEEyesnumber\label{eqn_reg_cap_cred}
\end{IEEEeqnarray*}

\begin{remark}\label{rem_cond_pd_999}
In view of Remark~\ref{rem_cond_pd_alpha}, but with reference to~\eqref{eqn_cond_pd_gauss}, the default probability of obligor~$i$ is no greater than
\begin{equation}\label{eqn_cond_pd_999}
	\P\big(D_{i} \given Y=\Phi^{-1}(0.001)\big) = p_{i}(-3.090) = \Phi\left(\frac{\Phi^{-1}(p_{i}) + \sqrts{\rho_{i}}\Phi^{-1}(0.999)}{\sqrts{1-\rho_{i}}}\right)
\end{equation}
in $99.9\%$ of economic scenarios.
\end{remark}

\begin{remark}\label{rem_bcbs_alpha}
BCBS \citeyearpar{BCBS0507} claims that the IRB approach sets regulatory capital for credit risk at a level where losses exceed it, ``on average, once in a thousand years.''  Qualifying this informal statement of probability, BCBS cautions that the $99.9\%$ confidence level was chosen because tier 2 capital ``does not have the loss absorbing capacity of tier 1'', and ``to protect against estimation error'' in model inputs as well as ``other model uncertainties.''  With provisions and capital amounting to as little as 2.0--3.0\% of EAD under the IRB approach, perhaps the claim of protection against insolvency due to credit losses at the $99.9\%$ confidence level should be interpreted as providing a margin for misspecification of the ASRF model, and not literally protection against a ``one in a thousand year'' event.  The choice of confidence level for the ASRF model may also have been influenced by the desire to produce regulatory capital requirements that are uncontroversial vis-\`a-vis Basel~I.  These qualifying remarks warn against the complacency engendered by the high confidence level chosen for the IRB approach.
\end{remark}

\section{Empirical Data}\label{sect_empir_data}
Under its implementation of Basel~II, APRA requires ADIs to assess capital adequacy for credit, market and operational risks.  ADIs determine regulatory capital for credit risk using either the standardised approach or, subject to approval, the IRB approach.  The former applies prescribed risk weights to credit exposures based on asset class and credit rating grade to arrive at an estimate of RWA.  Then, the minimum capital requirement is simply 8\% of RWA.  The standardised approach, which is an extension of Basel~I, is straightforward to administer and produces a relatively conservative estimate of regulatory capital.  The IRB approach, which implements ASRF model~\eqref{eqn_reg_cap_cred}, is a more sophisticated method requiring more input data estimated at higher precision.  Its greater complexity makes it more expensive to administer, but usually produces lower regulatory capital requirements than the standardised approach.  As a consequence, ADIs using the IRB approach may deploy their capital in pursuit of more (profitable) lending opportunities.  

In evaluating the robustness of the ASRF model, we conduct an empirical analysis of the Australian banking sector.  Its findings, though, are pertinent to other banking jurisdictions where regulatory capital charges are assessed under the IRB approach.  In this context we provide background on the Australian banking sector.  Upon implementation of Basel~II in the first quarter of 2008, APRA had granted the four largest Australian banks, designated ``major'' banks, approval to use the IRB approach to capital adequacy for credit risk.  They include: Commonwealth Bank of Australia (CBA), Westpac Banking Corporation (WBC), National Australia Bank (NAB), and Australia and New Zealand Banking Group (ANZ).  WBC acquired St.~George Bank (SGB) on 1 December 2008, and CBA acquired Bank of Western Australia (BWA) on 19 December 2008.  Putting them in a global context, all four major Australian banks have been ranked in the top 20 banks in the world by market capitalisation, and top 50 by assets, during 2013.  We construct a portfolio that is representative of credit exposures, reported to APRA by the major banks, for which regulatory capital is assessed by the ASRF model.  We choose the reporting period ending 31 December 2012.  At this date, the major banks accounted for 78.1\% of total assets on the balance sheet of ADIs regulated by APRA.  Furthermore, of the regulatory capital reported by the major banks, in aggregate,\footnote{
Data are reported for the major Australian banks, in aggregate, so as not to violate confidentiality agreements.
} 85.5\% was assessed for credit risk, 9.4\% for operational risk and 5.1\% for market risk.\footnote{
Figures for the quarter ending 31 December 2012, are published by APRA \citeyearpar{AQPS13} in its quarterly issue of ADI performance statistics.
}

\begin{figure}[!b]
	\centering
	\caption{EAD and RWA, respectively, of IRB credit exposures of the major Australian banks at 31 December 2012 are decomposed by sector (business, government and household).\protect\\
		\small{Source: Australian Prudential Regulation Authority.}} 
	\label{fig_irb_ead_rwa_121231}
	\scalebox{0.70}{
	\begin{tabular}{cc}
		\input{plot_irb_ead_121231}	& \input{plot_irb_rwa_121231} 
	\end{tabular}
	}
\end{figure}

ADIs lodge their statutory returns with APRA using a secure electronic data submission system \citep{APRA1406}.  A return is a collection of related forms covering the same reporting period.  A form is a dataset containing information on a specific topic (e.g., capital, risk class/sub-class, financial statement, etc.).  For reference purposes, forms in spreadsheet format and instructions are available at \url{www.apra.gov.au}.  Given the volume of data processed in preparing statutory returns, the major Australian banks automate their electronic data submissions using XBRL (eXtensible Business Reporting Language) --- an XML-based language for preparing, publishing, extracting and exchanging business and financial information.  Once forms are validated and returns submitted, data are stored in APRA's data warehouse.  We use SQL programming to fetch and aggregate data reported to APRA by the major banks for the quarter ending 31 December 2012.  In particular, we construct the aforementioned representative credit portfolio on the basis data sourced from IRB credit risk forms. 

Banking book exposures are reported on credit risk forms\footnote{
Reporting forms and instructions for ADIs available at \url{www.apra.gov.au}: ARF\_113\_1A, ARF\_113\_1B, ARF\_113\_1C, ARF\_113\_1D, ARF\_113\_3A, ARF\_113\_3B, ARF\_113\_3C, and ARF\_113\_3D.
} by IRB asset class: corporate (non-financial), small- and medium-sized enterprises (SME), bank, sovereign, residential mortgages, retail qualified revolving, and other retail.  For presentation purposes we merge IRB asset classes, and report credit exposures as business, government or household.  Henceforth, we refer to these banking book exposures as IRB credit exposures.  Figure~\ref{fig_irb_ead_rwa_121231} decomposes EAD and RWA, respectively, of IRB credit exposures into business, government and household sectors.  At 31 December 2012, RWA for credit risk reported by the major Australian banks, in aggregate, was divided 72.2/27.8 between IRB credit exposures and other banking book exposures.  The market dominance of the major banks, coupled with the concentration of their regulatory capital held against unexpected losses on IRB credit exposures, convey the significance of the ASRF model in protecting the Australian banking sector against insolvency.  Accordingly, we contend that, while focussing exclusively on IRB credit exposures of the major banks, our empirical analysis draws a representative sample of the credit risk assumed by the Australian banking sector.

\begin{table}[!b]
\begin{minipage}{\textwidth}		
	\centering
	\caption{Characteristics (EAD and exposure-weighted LGD, unconditional PD and asset correlation) of a portfolio that is representative of the IRB credit exposures of the major Australian banks at 31 December 2012 are reported for each credit rating grade by sector (business, government and household) and for the whole portfolio.}
	\label{tbl_rep_port_char}
	\resizebox{\textwidth}{!}{
	\begin{tabular}{l l r r r r r r r r}
		\toprule
		&	&	& \multicolumn{7}{c}{Credit rating grade} \\
		&	& \multicolumn{1}{c}{Total}	& \multicolumn{1}{c}{AAA}		&  \multicolumn{1}{c}{AA}	& \multicolumn{1}{c}{A}	&  \multicolumn{1}{c}{BBB}	& \multicolumn{1}{c}{BB}	& \multicolumn{1}{c}{B}	& \multicolumn{1}{c}{C\footnote{CCC, CC and C credit ratings.}}	\\
		\midrule
		Business		& EAD, $\ead_{i}$ (\$)				& 3,552	& 7		& 549	& 1,001	& 874	& 728	& 327	& 66		\\
		\cmidrule{2-10}
					& LGD, $\lgd_{i}$					& 0.429	& 0.219	& 0.519	& 0.474	& 0.418	& 0.356	& 0.339	& 0.412	\\
		\cmidrule{2-10}
					& Unconditional PD, $p_{i}$	 (\%)		& 1.02	& 0.02	& 0.03	& 0.11	& 0.41	& 1.24	& 3.08	& 18.56	\\
		\cmidrule{2-10}
					& Asset correlation, $\rho_{i}$			& 0.198	& 0.239	& 0.238	& 0.231	& 0.206	& 0.159	& 0.112	& 0.091	\\
		\midrule
		Government	& Exposure at default, $\ead_{i}$ (\$)	& 785	& 538	& 203	& 18		& 10		& 13		& 3		&		\\
		\cmidrule{2-10}
					& Loss given default, $\lgd_{i}$			& 0.102	& 0.088	& 0.085	& 0.223	& 0.496	& 0.387	& 0.408	&		\\
		\cmidrule{2-10}
					& Unconditional PD, $p_{i}$	 (\%)		& 0.07	& 0.01	& 0.03	& 0.09	& 0.40	& 1.74	& 3.00	&		\\
		\cmidrule{2-10}
					& Asset correlation, $\rho_{i}$			& 0.237	& 0.239	& 0.238	& 0.235	& 0.218	& 0.171	& 0.149	&		\\	
		\midrule
		Household	& Exposure at default, $\ead_{i}$ (\$)	& 5,663	&		&		& 2,581	& 1,725	& 919	& 291	& 147	\\
		\cmidrule{2-10}
					& Loss given default, $\lgd_{i}$			& 0.244	&		&		& 0.233	& 0.225	& 0.264	& 0.357	& 0.324	\\
		\cmidrule{2-10}
					& Unconditional PD, $p_{i}$	 (\%)		& 0.99	&		&		& 0.08	& 0.39	& 1.10	& 3.68	& 17.80	\\
		\cmidrule{2-10}
					& Asset correlation, $\rho_{i}$			& 0.143	&		&		& 0.144	& 0.147	& 0.143	& 0.123	& 0.127	\\
		\midrule
		Portfolio		& Exposure at default, $\ead_{i}$ (\$)	& 10,000	& 545	& 752	& 3,600	& 2,609	& 1,660	& 621	& 213	\\
		\cmidrule{2-10}
					& Loss given default, $\lgd_{i}$			& 0.299	& 0.090	& 0.402	& 0.300	& 0.291	& 0.306	& 0.348	& 0.350	\\
		\cmidrule{2-10}
					& Unconditional PD, $p_{i}$	 (\%)		& 0.93	& 0.01	& 0.03	& 0.09	& 0.40	& 1.16	& 3.36	& 18.03	\\
		\cmidrule{2-10}
					& Asset correlation, $\rho_{i}$			& 0.170	& 0.239	& 0.238	& 0.169	& 0.167	& 0.150	& 0.117	& 0.116	\\
		\bottomrule
	\end{tabular}
	}
\end{minipage}
\end{table}

Under the IRB approach, ADIs assign their on- and off-balance sheet credit exposures to internally-defined obligor grades reflecting PD bands, and LGD bands.  At 31 December 2012, RWA of IRB credit exposures held in the banking book of the major Australian banks was divided 75.0/25.0 between on-balance sheet assets and off-balance sheet exposures.  EAD, RWA, expected loss, and exposure weighted LGD, unconditional PD and firm size are reported for each obligor grade.  We assign IRB credit exposures reported by the major banks to standardised PD bands (i.e., consistent across the major banks), and calculate risk parameters characterising each of these standardised obligor grades.  Asset correlation is constant for residential mortgages and retail qualified revolving credit exposures; a function of unconditional PD for corporate, bank, sovereign and other retail credit exposures; and a function of firm size and unconditional PD for SME credit exposures \citep{BCBS0507}.  

We construct a portfolio that is representative of the IRB credit exposures of the major Australian banks at 31 December 2012.  That is, exposure weight of each PD band within an IRB asset class is equal to the corresponding EAD reported by major banks as a percentage of IRB credit exposures.  Credits assigned the same obligor grade share the same risk characteristics: LGD, unconditional PD and asset correlation.  Table~\ref{tbl_rep_port_char} categorises credit exposures as business, government or household, and groups them into credit rating grades on the basis of obligors' unconditional PD.\footnote{
We use the mapping of S\&P credit rating grades to KMV expected default frequency values derived by \citet{LJA02}.
}  It reports EAD, and exposure weighted LGD, unconditional PD and asset correlation for each credit rating grade by sector.  Sections~\ref{sect_cvrg_asymp_dist} and \ref{sect_cap_sens_depend_struct}, respectively, use characteristics of this representative credit portfolio to measure the rate of convergence to the asymptotic portfolio loss distribution, and evaluate the sensitivity of credit risk capital to dependence structure as modelled by asset correlations and elliptical copulas.

\section{Rate of Convergence to the Asymptotic Distribution}\label{sect_cvrg_asymp_dist}
The Basel~II IRB approach to credit risk is premised on Proposition~\ref{prop_qntl_loss_dist_subst}.  Assuming a conditional independence model of an asymptotic credit portfolio, it asserts that quantiles of the distribution of conditional expectation of portfolio percentage loss may be substituted for quantiles of the portfolio loss distribution.  Before measuring the rate of convergence to the asymptotic portfolio loss distribution, we demonstrate that Proposition~\ref{prop_qntl_loss_dist_subst} holds for a representative credit portfolio within a static (i.e., single-period) framework.  Table~\ref{tbl_rep_port_char} describes the representative credit portfolio.  In order to construct a portfolio that exhibits sufficient granularity we impose the constraint that no credit accounts for more than one basis point exposure.  This exercise demonstrates that the $\alpha$ quantile of the distribution of $\E[L_{n} \given Y]$, which is associated with the $(1\!-\!\alpha)$ quantile of the distribution of $Y$, may be substituted for the $\alpha$ quantile of the distribution of $L_{n}$.  Firstly, capital held against unexpected losses at the $99.9\%$ confidence level over a one-year horizon, $\Kap[99.9\%]{L_{n}}$, is calculated analytically by ASRF model~\eqref{eqn_reg_cap_cred}.  The expectation of portfolio percentage loss over a one-year horizon conditional on a state of the economy that is worse than at most $99.9\%$ of economic scenarios, ${\E\big[L_{n} \given Y=\Phi^{-1}(0.001)\big]}$, is readily calculated by the first term on the right-hand side of~\eqref{eqn_reg_cap_cred}.  Expected loss, $\E[L_{n}]$, is given by the second term of the right-hand side of~\eqref{eqn_reg_cap_cred}.  And $\Kap[99.9\%]{L_{n}}$ is the difference between $\E\big[L_{n} \given Y=\Phi^{-1}(0.001)\big]$ and $\E[L_{n}]$.

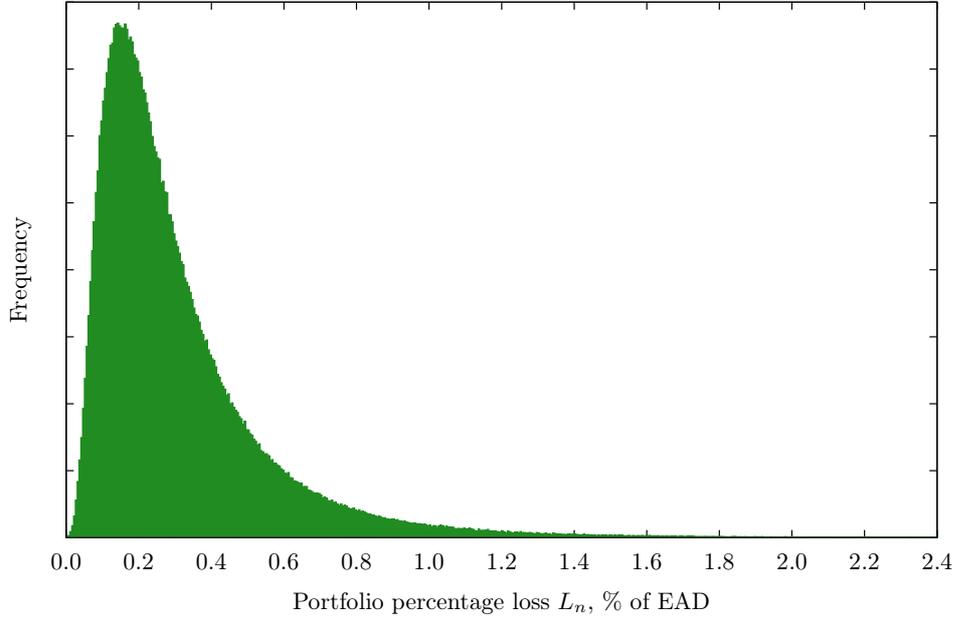
\begin{figure}[!t]
	\centering
	\caption{The empirical loss distribution of the representative credit portfolio described in Table~\ref{tbl_rep_port_char} is generated by Monte Carlo simulation.  The portfolio is constructed such that no credit accounts for more than one basis point exposure, and credit losses are reported as a percentage of EAD.}
	\label{fig_empir_loss_dist}
	\scalebox{0.90}{
		\input{plot_empir_loss_dist}
	}
\end{figure}

Next, $K_{99.9\%}(L_{n})$ is determined by computationally intensive simulation~\eqref{eqn_sim_port_loss}.  Credit VaR at the 99.9\% confidence level over a one-year horizon, $\VaR[99.9\%]{L_{n}}$, is determined from the empirical loss distribution of the representative credit portfolio, which is generated by simulation of \eqref{eqn_cond_port_loss} parameterised by \eqref{eqn_zeta_y_gauss}.\footnote{
Gaussian random variables are generated using GNU Scientific Library routine \texttt{gsl\_ran\_gaussian}.  It implements the Box-Muller algorithm, which make two calls to the MT19937 generator of Makoto Matsumoto and Takuji Nishimura. 
}  Monte Carlo simulation performs 1,000,000 iterations to generate the empirical loss distribution (Figure~\ref{fig_empir_loss_dist}).  Equation~\eqref{eqn_sim_port_loss} calculates the portfolio percentage loss for each iteration, and credit VaR is the 99.9\% quantile of the empirical loss distribution described by \eqref{eqn_sim_loss_dist}.  Expected loss, $\E[L_{n}]$, is given by \eqref{eqn_sim_exp_loss}.  And $K_{99.9\%}(L_{n})$ is the difference between $\VaR[99.9\%]{L_{n}}$ and $\E[L_{n}]$.

\newcolumntype{g}{>{\columncolor{Gold}}r}
\begin{table}[!t]
	\centering
	\caption{Capital charges, at the 99.9\% confidence level over a one-year horizon, assessed on the representative credit portfolio described in Table~\ref{tbl_rep_port_char}.  The portfolio is constructed such that no credit accounts for more than one basis point exposure.  Credit VaR at the 99.9\% confidence level is computed numerically by Monte Carlo simulation, and its analytical approximation is calculated by the ASRF model.}
	\label{tbl_asrf_sim_rslt}
	\begin{tabular}[t]{r g g l}
		\toprule
		& \multicolumn{2}{c}{\% of EAD}	&	\\
		& \multicolumn{1}{c}{ASRF} & \multicolumn{1}{c}{Simulation}	&	\\
		\cmidrule{2-3}
		$\E\big[L_{n} \given Y=\Phi^{-1}(0.001)\big]$	& 2.18	& 2.19	& $\VaR[99.9\%]{L_{n}}$	\\
		$\E[L_{n}]$							& 0.31	& 0.31	& $\E[L_{n}]$	\\
		$\Kap[99.9\%]{L_{n}}$					& 1.87	& 1.88	& $K_{99.9\%}(L_{n})$	\\
		\bottomrule
	\end{tabular}
\end{table} 

It remains to compare the estimates produced by the analytical and simulation models.  Given that this representative portfolio contains a large number of credits without concentration in a few names dominating the rest of the portfolio, we anticipate that as the number of simulation iterations increases, its empirical loss distribution will converge to the distribution of conditional expectation of portfolio percentage loss.  Confirming our intuition, Table~\ref{tbl_asrf_sim_rslt} reports that estimates from the ASRF model and Monte Carlo simulation are within one basis point of one another.  Recall that ASRF model~\eqref{eqn_reg_cap_cred} describes defaults as conditionally independent Gaussian random variables, and the Monte Carlo simulation of \eqref{eqn_cond_port_loss}, parameterised by \eqref{eqn_zeta_y_gauss}, draws random variables representing the single systematic risk factor and obligor specific risks from a Gaussian distribution.  So, it would be surprising if this convergence did not occur.  Indeed, for a number of iterations large enough and a portfolio exhibiting sufficient granularity, the simulation does little more than demonstrate that the expectation of indicator function~$\One{Z_{i} < \zeta_{i}(y)}$ is conditional PD $p_{i}(y)$.  Our findings provide empirical support for Proposition~\ref{prop_qntl_loss_dist_subst}.

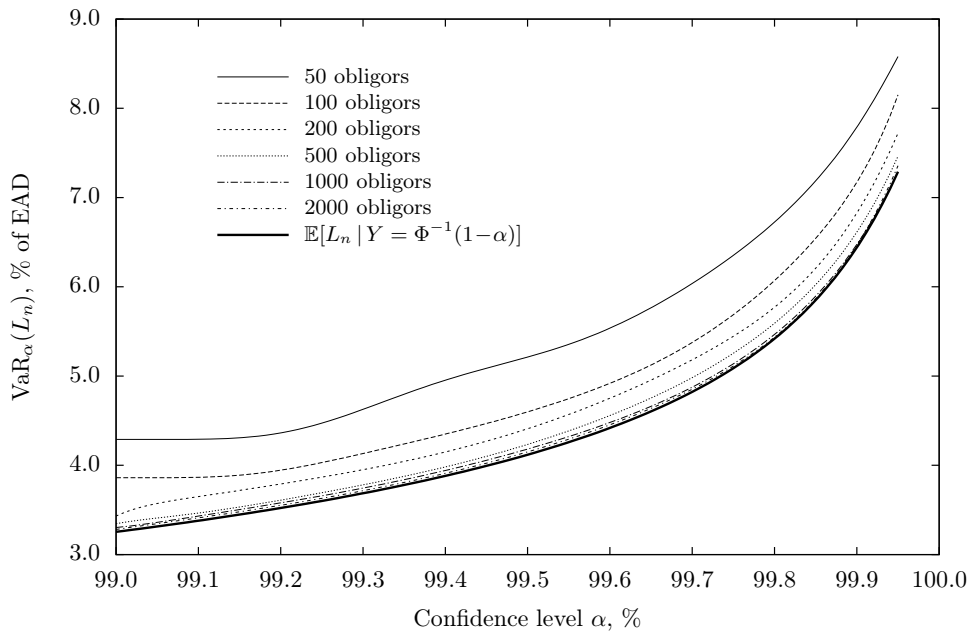
\begin{figure}[!b]
	\caption{Credit VaR at the $\alpha$ confidence level, ${99.0\% \leq \alpha < 100.0\%}$, for portfolios comprising obligors ranging in number from 50 to 2000.  The portfolios are representative of IRB credit exposures of the major Australian banks to the business sector reported in Table~\ref{tbl_rep_port_char}.  The curves illustrate the rate of convergence in the tail of empirical loss distributions, $\VaR[\alpha]{L_{n}}$, to the distribution of conditional expectation of portfolio percentage loss, ${\E[L_{n} \given Y = \Phi^{-1}(1\!-\!\alpha)]}$, representing the loss distribution of an infinitely fine-grained portfolio.}
	\label{fig_cvrg_asymp_dist}
	\scalebox{0.90}{
		\input{plot_cvrg_asymp_dist}
	}
\end{figure}

Evidently, the representative credit portfolio described in Table~\ref{tbl_rep_port_char}, with no credit accounting for more that one basis point exposure, adequately satisfies the asymptotic granularity condition of Definition~\ref{def_asymp_port}, an assumption of Proposition~\ref{prop_qntl_loss_dist_subst}.   But how large need the number of credits, or obligors, constituting a portfolio be for ${\E[L_{n} \given Y = \Phi^{-1}(1\!-\!\alpha)]}$ to produce a statistically accurate estimate of $\VaR[\alpha]{L_{n}}$?  We proceed to address this question by constructing portfolios comprising obligors ranging in number from 50 to 2000.  For each portfolio we assign its constituent obligors equal dollar EAD, and exposure weighted LGD, unconditional PD and asset correlation as calculated for exposures to the business sector of the representative credit portfolio described in Table~\ref{tbl_rep_port_char}.  We choose to conduct this exercise on portfolios representative of IRB credit exposures of the major banks to the business sector, because corporate loans are typically ``lumpier'' than residential mortgages.

We generate empirical loss distributions, one for each of the constituted portfolios containing a finite number of obligors, by simulation of \eqref{eqn_cond_port_loss} parameterised by \eqref{eqn_zeta_y_gauss}.  Monte Carlo simulation performs 1,000,000 iterations to generate an empirical loss distribution for each of the constituted portfolios.  The distribution of conditional expectation of portfolio percentage loss represents the loss distribution of an infinitely fine-grained portfolio.  Solvency assessment is concerned with the tail of the portfolio loss distribution, so we examine the confidence interval bounded by ${99.0\% \leq \alpha < 100.0\%}$.  $\VaR[\alpha]{L_{n}}$ for each of the constituted portfolios is the $\alpha$ quantile of the empirical loss distributions described by~\eqref{eqn_sim_loss_dist}.  ${\E[L_{n} \given Y = \Phi^{-1}(1\!-\!\alpha)]}$ is given by the first term on the right-hand side of ASRF model~\eqref{eqn_cap_cred_gauss}.  Figure~\ref{fig_cvrg_asymp_dist} plots the tail of each of the empirical loss distributions along with the tail of the distribution of conditional expectation of portfolio percentage loss, illustrating the rate of convergence in terms of number of obligors.  By inspection we argue that a statistically accurate estimate of $\VaR[\alpha]{L_{n}}$ is given by ${\E[L_{n} \given Y = \Phi^{-1}(1\!-\!\alpha)]}$ for portfolios comprising 1,000 obligors or more, without concentration in a few names dominating the rest of the portfolio.

\section{Sensitivity of Credit Risk Capital to Dependence Structure}\label{sect_cap_sens_depend_struct}
Regulatory capital models, which serve the sole purpose of solvency assessment, require precision in the measurement of absolute risk levels under stressed economic conditions.  In this section we evaluate the robustness of the model specification of the Basel~II IRB approach to a relaxation of model assumptions.  Firstly, in relation to the one-factor Gaussian copula, we evaluate the sensitivity of credit risk capital to dependence structure as described by asset correlations.  Then, we examine the effect of tail dependence by measuring the sensitivity of credit risk capital to dependence structure, as modelled by a variety of elliptical copulas.  As in Section~\ref{sect_cvrg_asymp_dist}, we conduct this empirical analysis on a portfolio, with characteristics described in Table~\ref{tbl_rep_port_char}, that is  representative of the IRB credit exposures of the major banks at 31 December 2012.

\subsection{Default Dependence Described by Asset Correlations}\label{sect_sens_asset_corr}
The IRB approach, in effect, applies the one-factor Gaussian copula of Corollary~\ref{cor_ofgc} with matrix~\eqref{eqn_ofgc_gamma} describing pairwise correlations between obligors' asset values.  Choosing the one-factor Gaussian copula to model default dependence, we evaluate the sensitivity of credit risk capital to dependence structure as described by asset correlations.

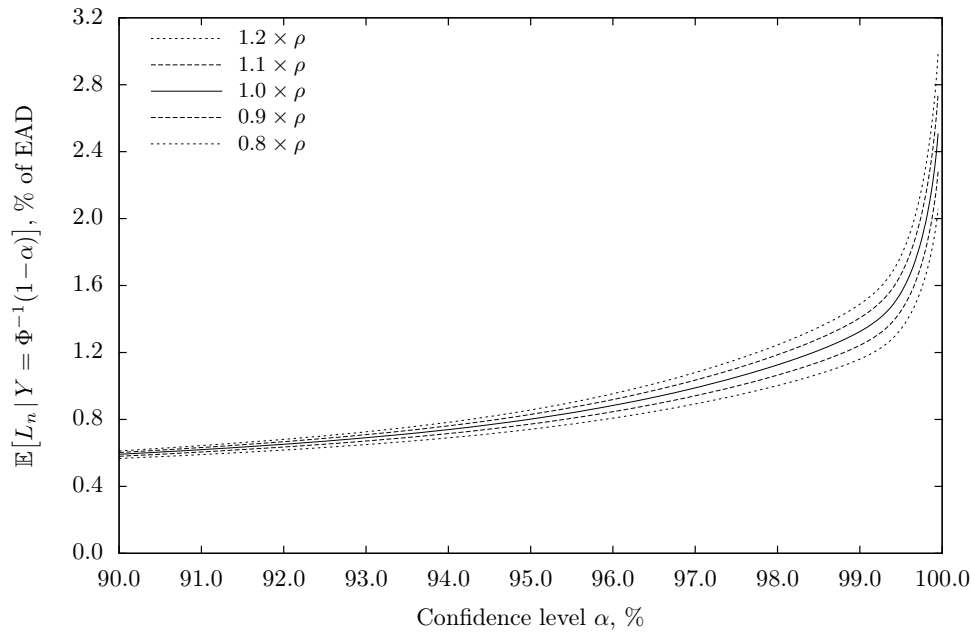
\begin{figure}[!t]
	\caption{${\E\big[L_{n} \given Y=\Phi^{-1}(1\!-\!\alpha)\big]}$, ${90.0\% \leq \alpha < 100.0\%}$, for the representative credit portfolio described in Table~\ref{tbl_rep_port_char}.  The sensitivity of ${\E\big[L_{n} \given Y=\Phi^{-1}(1\!-\!\alpha)\big]}$ to default dependence structure is measured by adjusting the asset correlation assigned to constituent obligors by $\pm10\%$ and $\pm20\%$.}
	\label{fig_var_qntl_asset_corr}
	\scalebox{0.90}{
		\input{plot_var_qntl_asset_corr}
	}
\end{figure}

Recall from Section~\ref{sect_empir_data} that the IRB approach models asset correlation as a constant for residential mortgages and retail qualified revolving credit exposures; a function of unconditional PD for corporate, bank, sovereign and other retail credit exposures; and a function of firm size and unconditional PD for SME credit exposures.  Parameters of the asset correlation functions prescribed by the IRB approach were derived from an analysis of times series collected by G10 supervisors.  Even if the dependence described by parameters derived from this time series analysis continues to hold, the actual asset correlation will lie in some distribution around the estimate given by the model specification of the IRB approach.  So, some measure of the sensitivity of credit risk capital to the error in asset correlation estimates would be informative.   

Figure~\ref{fig_var_qntl_asset_corr} plots the sensitivity of ${\E\big[L_{n} \given Y=\Phi^{-1}(1\!-\!\alpha)\big]}$, and hence credit risk capital, to dependence structure as described by asset correlations.  Here, we adopt the credit risk capital interpretation of Definition~\ref{def_cap_cred_risk}, which applies to portfolios that adequately satisfy the asymptotic granularity condition of Definition~\ref{def_asymp_port}.  ${\E\big[L_{n} \given Y=\Phi^{-1}(1\!-\!\alpha)\big]}$, given by the first term on the right-hand side of \eqref{eqn_cap_cred_gauss}, is calculated for the representative credit portfolio described in Table~\ref{tbl_rep_port_char}, and then recalculated after adjusting the asset correlation assigned to constituent obligors by $\pm10\%$ and $\pm20\%$.  Since solvency assessment is concerned with the tail of the portfolio loss distribution, we examine the confidence interval bounded by ${90.0\% \leq \alpha < 100.0\%}$.  Conditional expectation of portfolio percentage loss becomes more sensitive to asset correlation as one moves further into the tail of the portfolio loss distribution.  At the $99.9\%$ confidence level the relative error in asset correlation estimates affects the calculation of credit risk capital by a similar magnitude.  Naturally, this heuristic for the sensitivity of credit risk capital to asset correlation only applies deep in the tail of the loss distribution of near asymptotic portfolios with characteristics (viz., unconditional PDs and asset correlations) not very different from those described in Table~\ref{tbl_rep_port_char}.
 
\subsection{Dependence Structure Modelled by Elliptical Copulas}\label{sect_sens_ellip_cop}
While it's generally acknowledged that models which assume that financial data follow a Gaussian distribution tend to underestimate tail risk, the IRB approach to solvency assessment does apply the one-factor Gaussian copula.  As discussed in Section~\ref{sect_ellip_cop}, Gaussian copulas do not exhibit tail dependence --- the tendency for extreme observations (i.e., credit defaults) to occur simultaneously for all random variables.  We examine the effect of tail dependence by measuring the sensitivity of credit risk capital to dependence structure as modelled by elliptical copulas, including Gaussian and Student's $t$ copulas.  The latter, which we abbreviate by $t$-copula, admits tail dependence with fewer degrees of freedom producing stronger dependence (Figure~\ref{fig_bivar_ellip_cop}). 

\begin{figure}[!t]
	\caption{$\VaR[\alpha]{L_{n}}$, ${90.0\% \leq \alpha < 100.0\%}$, for the representative credit portfolio described in Table~\ref{tbl_rep_port_char}.  The portfolio is constructed such that no credit accounts for more than one basis point exposure.  The sensitivity of $\VaR[\alpha]{L_{n}}$ to default dependence structure is measured by generating empirical loss distributions using elliptical copulas and assuming that asset values follow a log-normal distribution.}
	\label{fig_var_qntl_ellip_cop}
	\scalebox{0.90}{
		\input{plot_var_qntl_ellip_cop}
	}
\end{figure}
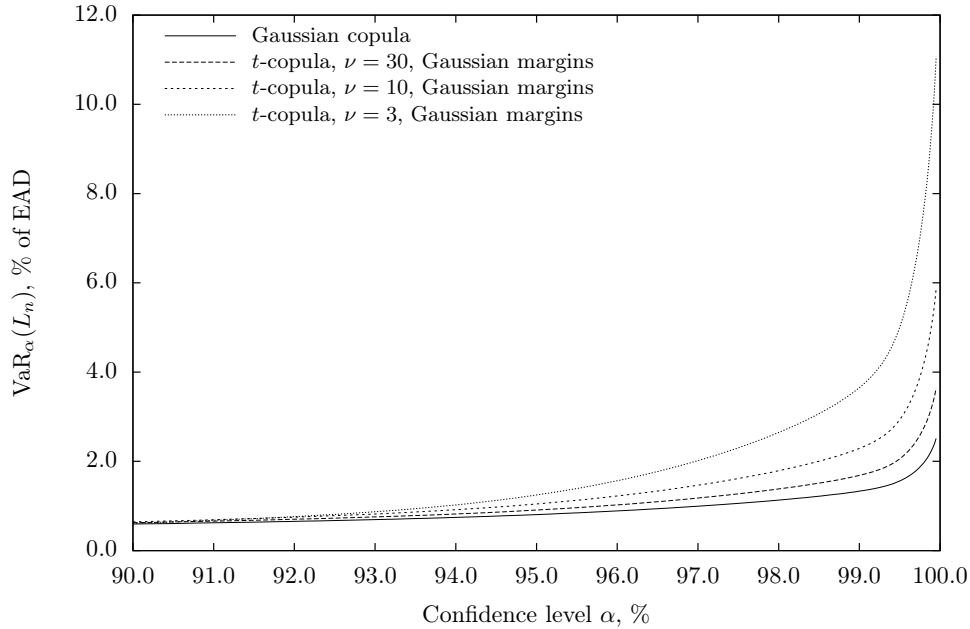

In order to evaluate the sensitivity of credit risk capital to default dependence structure, we generate empirical loss distributions of the representative credit portfolio described in Table~\ref{tbl_rep_port_char} using elliptical copulas: one-factor Gaussian copula, and $t$-copulas with 30, 10 and 3 degrees of freedom and Gaussian margins.  Again, the portfolio is constructed such that no credit accounts for more than one basis point exposure.  Assuming that asset values follow a log-normal distribution, distributional differences are attributed to the copula function modelling default dependence.  The one-factor Gaussian copula is implemented by simulation of \eqref{eqn_cond_port_loss} parameterised by \eqref{eqn_zeta_y_gauss}, and generates the empirical loss distribution plotted in Figure~\ref{fig_empir_loss_dist}. $t$-copulas with Gaussian margins are implemented by simulation of \eqref{eqn_port_loss_t_cop} parameterised by \eqref{eqn_zeta_y_t_cop}.\footnote{
Normally and chi-square distributed random variables are generated using GNU Scientific Library routines \texttt{gsl\_ran\_gaussian} and \texttt{gsl\_ran\_chisq}, respectively. 
}  Monte Carlo simulation performs 1,000,000 iterations.

Figure~\ref{fig_var_qntl_ellip_cop} plots the sensitivity of credit VaR, and hence credit risk capital, to default dependence structure as modelled by the one-factor Gaussian copula and $t$-copulas with Gaussian margins.  As in Section~\ref{sect_sens_asset_corr}, we examine the confidence interval bounded by ${90.0\% \leq \alpha < 100.0\%}$.  At the $90.0\%$ confidence level there is little difference in estimates of credit VaR computed by the respective elliptical copulas.  However, as one moves further into the tail of the empirical loss distribution estimates of credit VaR diverge at an accelerating rate.  At the $99.9\%$ confidence level credit VaR computed by the $t$-copula with $\nu = 10$ is more than double that computed by the Gaussian copula; and the $t$-copula with $\nu = 3$ estimates credit VaR to be more than four times the estimate produced by the Gaussian copula.  Unfortunately, the calibration of degrees of freedom is not easy and to some extent subjective, which may explain why the one-factor Gaussian copula prevails in practice.  The sensitivity of credit risk capital to the choice of elliptical copula can be much greater than its sensitivity to asset correlation.

\section{Conclusion}\label{sect_concl}
The academic contribution of this paper is both theoretical and empirical.  We derive the theoretical foundations of the IRB approach to capital adequacy for credit risk.  Adapting the single asset model of \citet{MRC74} to a portfolio of credits, \citet{VO02} derived a function that transforms unconditional PDs into PDs conditional on a single systematic risk factor.  It is the kernel of ASRF model prescribed by the IRB approach.  The ASRF model rests on the proposition, due to \citet{GMB03}, that quantiles of the distribution of conditional expectation of portfolio percentage loss may be substituted for quantiles of the portfolio loss distribution.  Here, our theoretical contribution is to present a more compact proof of this proposition starting from weaker assumptions, and extend the ASRF model to a more general setting than the usual Gaussian case.

We proceed to conduct an empirical analysis of the ASRF model using internal bank data collected by the prudential regulator.  Access to these data distinguishes our research from other empirical studies on the IRB approach.  Firstly, we demonstrate that Gordy's proposition holds for a portfolio that is representative of the IRB credit exposures of the major Australian banks, and exhibits sufficient granularity.  Quantiles of the distribution of conditional expectation of portfolio percentage loss are readily calculated by the analytical ASRF model, while quantiles of the portfolio loss distribution are determined by a computationally intensive simulation of the Gaussian conditional independence model.  For our representative portfolio, estimates of credit risk capital from the ASRF model and Monte Carlo simulation are within one basis point of one another.  Then, we measure the rate of convergence, in terms of number of obligors, of empirical loss distributions to the asymptotic portfolio loss distribution (i.e., the distribution of conditional expectation of portfolio percentage loss representing an infinitely fine-grained portfolio).  For portfolios comprising 1,000 obligors or more, without concentration in a few names dominating the rest of the portfolio, the conditional expectation of portfolio percentage loss provides a statistically accurate estimate of portfolio percentage loss in the tail of the portfolio loss distribution.  Assuming that a Gaussian conditional independence model accurately describes default dependence and credit portfolios exhibit sufficient granularity, we find empirical support for the IRB approach.

In the context of solvency assessment, we empirically evaluate the sensitivity of credit risk capital (in the tail of the portfolio loss distribution) to dependence structure.  The IRB approach applies the one-factor Gaussian copula with default dependence described by the matrix of pairwise correlations between obligors' asset values.  The relative error in asset correlation estimates affects the calculation of credit risk capital by a similar magnitude.  Assuming that asset values follow a log-normal distribution, the portfolio loss distribution is determined by the copula function chosen to combine its margins.  The sensitivity of credit risk capital to dependence structure as modelled by elliptical copulas, including Gaussian and Student's $t$ copulas, can be much greater than its sensitivity to asset correlations.  

A couple of future research directions emerge from this paper.  In a related research paper \citet{RT14a} take measurements from the ASRF model of the prevailing state of Australia's economy and the level of capitalisation of its banking sector.  They find general agreement with macroeconomic indicators, financial statistics and external credit ratings.  Given the range of economic conditions, from mild contraction to moderate expansion, experienced in Australia since the implementation of Basel~II, their empirical findings support a favourable assessment of the ASRF model for purposes of capital allocation, performance attribution and risk monitoring.  Evaluating the ASRF model for the purpose of solvency assessment, would involve taking readings from north Atlantic banking jurisdictions that experienced the full force of the financial crisis of 2007--09, which precipitated the worst global recession since the Great Depression of the 1930s.

Portfolios of large banks contain financial instruments that are influenced by a multitude of underlying risk factors.  Internal and external risk controllers, such as the prudential regulator, require risks to be measured over varying holding periods for different risk classes (e.g., one-day or 10-day for market risk, and one-year for credit and operational risks).  The path dependent nature of many risks and the requirement to measure portfolio risk over different time horizons leads to a multi-period, or dynamic, simulation.  It is practical then to simulate all variables, including defaults and survivals, in each time period.  In order to consistently represent dependence between multivariate default times in each period of the simulation, \citet{BMS13} model default times using a continuous-time Markov chain with multivariate exponential distribution (i.e., Marshall-Olkin copula).  An efficient implementation of this Markovian credit risk model remains a challenge in building a feasible multi-period simulation for measuring risk in massive portfolios.

\section*{Acknowledgements}
This research has been conducted with the support of the Australian Prudential Regulation Authority (APRA).  It has provided access to professionals with expertise in financial regulation, risk modelling and data management, along with access to internal bank data collected by APRA from the institutions that it supervises.  We acknowledge with pleasure the support and encouragement from Charles Littrell, Executive General Manager at APRA, as well as his constructive comments on early drafts of this paper.  Anthony Coleman and Guy Eastwood of the Credit Risk Analytics team offered valuable input in the design of the empirical analysis and review of the findings.  The Statistics team, led by Steve Davies, provided data support for the empirical analysis, and reviewed the paper for compliance with confidentiality agreements.  Finally, we thank the Research team, headed by Bruce Arnold, for its contribution in distilling research topics that are pertinent to prudential regulators.   Financial support for this research is gratefully received in the form of an Australian Postgraduate Award, and a scholarship sponsored by the Capital Markets Cooperative Research Centre and APRA.

\appendix
\section{Proof of Proposition~\ref{prop_cond_port_loss}}\label{appx_cond_port_loss}
\citet[Proposition~1]{GMB03} established that, conditional on a single systematic risk factor, the portfolio percentage loss converges, almost surely, to its conditional expectation as the portfolio approaches asymptotic granularity.

\condportloss*	

\begin{proof}
It relies on a variant of the strong law of large numbers \citep[see, e.g.,][Theorem~7.5.1]{GS01} based on Kolmogorov's convergence criterion and Kronecker's lemma.

For a credit portfolio comprising $n$ obligors, $w_{i} = \ead_{i} / \EAD_{n}$ is the exposure weight of obligor~$i$.  Firstly, observe that given realisation~$y \in \R$ of systematic risk factor~$Y$, conditional variances satisfy
\begin{equation*}
	\sum_{i=1}^{\infty} \Var{\frac{\ead_{i}\lgd_{i}\One{Z_{i} < \zeta_{i}(y)}}{\EAD_{i}}} = \sum_{i=1}^{\infty} \left(\frac{\ead_{i}}{\EAD_{i}}\right)^{2}\lgd_{i}^{2}p_{i}(y)(1-p_{i}(y)) \leq \sum_{i=1}^{\infty} \left(\frac{\ead_{i}}{\EAD_{i}}\right)^{2} < \infty,
\end{equation*}
where $\Var{\One{Z_{i} < \zeta_{i}(y)}} = p_{i}(y)(1-p_{i}(y))$.  Recall that $\E\left[\One{Z_{i} < \zeta_{i}(y)}\right] = p_{i}(y)$, and the sequence of random variables ${\{\ead_{k}\lgd_{k}\One{Z_{k} < \zeta_{k}(y)} / \EAD_{k}\}}$ is independent with respect to $\P_{y}$.  Then,
\begin{equation}\label{eqn_kolmo_inf_series}
	\sum_{i=1}^{\infty} \frac{\ead_{i}\lgd_{i}}{\EAD_{i}}\big(\One{Z_{i} < \zeta_{i}(y)} - p_{i}(y)\big)
\end{equation}
converges $\P_{y}$-almost surely by Kolmogorov's convergence criterion \citep[Theorem~6.5.2]{GA05}.  Notice that ${\{\ead_{k}\lgd_{k}(\One{Z_{k} < \zeta_{k}(y)} - p_{k}(y))\}}$ constitutes a sequence of random variables, and the sequence of real numbers $\{\EAD_{k}\}$ is positive and strictly increasing to infinity.  Next, applying Kronecker's lemma \citep[Lemma~6.5.1]{GA05} to infinite series~\eqref{eqn_kolmo_inf_series} yields 
\begin{equation*}
	\frac{1}{\EAD_{n}}\sum_{i=1}^{n}\ead_{i}\lgd_{i}\big(\One{Z_{i} < \zeta_{i}(y)} - p_{i}(y)\big) \stackrel{a.s.}{\longrightarrow} 0 \textit{ as } n\rightarrow\infty,
\end{equation*}
which implies that
\begin{equation*}
	\lim_{n\rightarrow\infty} \left(L_{n} - \sum_{i = 1}^{n} w_{i}\lgd_{i}p_{i}(y)\right) = 0,
\end{equation*}
$\P_{y}$-almost surely for all $y \in \R$.  Accordingly, \eqref{eqn_cond_port_loss_lim} holds $\P$-almost surely.
\end{proof}

\section{Proof of Corollary~\ref{cor_lim_loss_dist}}\label{appx_lim_loss_dist}
We generalise the \citet{VO02} formulation of the loss distribution function of an asymptotic, homogeneous portfolio, which models default dependence as a multivariate Gaussian process.

\limlossdist*	

\begin{proof}
For the homogeneous case of Definition~\ref{def_homo_port}, let
\begin{equation*}
	\zeta(y) = \frac{F^{-1}(p) - \gamma y}{\sqrts{1-\gamma^{2}}} = G^{-1}\big(p(y)\big).
\end{equation*} 
Then, conditional on realisation~$y \in \R$ of systematic risk factor~$Y$, the portfolio percentage loss is calculated as
\begin{equation}\label{eqn_cond_port_loss_homo}
	L_{n} = \sum_{i = 1}^{n} w_{i}\lgd\One{Z_{i} < \zeta(y)},
\end{equation}
where $\One{Z_{i} < \zeta(y)}$ are independent identically distributed Bernoulli random variables with finite mean and variance.  By the strong law of large numbers \citep[see, e.g.,][Theorem~7.5.1]{GS01}, $\One{Z_{i} < \zeta(y)}$ converges to its conditional expectation $p(y)$ as $n\rightarrow\infty$, $\P_{y}$-almost surely for all $y\in\R$.  Observing that $\sum_{i=1}^{n}w_{i} = 1$ for all $n \in \N$, and $\lgd \in [0,1]$,  \eqref{eqn_cond_lim_port_loss} holds $\P$-almost surely.

It follows from~\eqref{eqn_cond_lim_port_loss} that given $Y=y$,
\begin{equation*}
	\lim_{n\rightarrow\infty} L_{n} = \lgd p(y),
\end{equation*}
which is a deterministic quantity.  Hence,
\begin{equation*}
	\lim_{n\rightarrow\infty} \P(L_{n} \leq \el \given Y=y) =  \One{0 < \lgd p(y) \leq \el} =  \One{p^{-1}(\el/\lgd) \leq y < \infty}.
\end{equation*}
Then, integrating over systematic risk factor~$Y$ with density function~$h(y)$, yields the limiting form of the portfolio loss distribution:
\begin{IEEEeqnarray*}{rCl}
	\lim_{n\rightarrow\infty} \P(L_{n} \leq \el) & = & \lim_{n\rightarrow\infty} \int_{-\infty}^{\infty} \P(L_{n} \leq \el, Y=y) \dee{y} =  \lim_{n\rightarrow\infty} \int_{-\infty}^{\infty} \P(L_{n} \leq \el \given Y=y) h(y)\dee{y} \\	
	& = & \int_{-\infty}^{\infty} \lim_{n\rightarrow\infty} \P(L_{n} \leq \el \given Y=y) h(y)\dee{y}= \int_{-\infty}^{\infty} \One{p^{-1}(\el/\lgd) \leq y < \infty}h(y)\dee{y} \\
	& = & \int_{p^{-1}(\el/\lgd)}^{\infty} h(y)\dee{y} = 1 - \int_{-\infty}^{p^{-1}(\el/\lgd)} h(y)\dee{y} \\
	& = & 1 - H\left(p^{-1}(\el/\lgd)\right) = 1 - H\left(\frac{F^{-1}(p) - \sqrts{1-\gamma^{2}}G^{-1}(\el/\lgd)}{\gamma}\right).
\end{IEEEeqnarray*}
The dominated convergence theorem \citep[see, e.g.,][Theorem~1.4.9]{SSE04} provides conditions under which the limit of integrals of a sequence of functions is the integral of the limiting function.  It justifies the third equality, while the last equality follows from \eqref{eqn_inv_cond_pd}.  
\end{proof}

\section{Proposition due to \citet{GMB03}}\label{appx_gordy_prop_5}
The Basel~II IRB approach rests on Proposition~5 of \citet{GMB03} that quantiles of the distribution of conditional expectation of portfolio percentage loss may be substituted for quantiles of the portfolio loss distribution.  It leads to an analytical approximation of credit VaR.  In Section~\ref{sect_cred_var} we present a version of this proposition that relaxes the technical conditions imposed by Gordy, resulting in a more compact, or parsimonious, proof.  Here, the statement and proof of Proposition~\ref{prop_gordy_prop_5} closely follow that of \citet{GMB03}.

\begin{proposition}\label{prop_gordy_prop_5}
Consider a credit portfolio comprising $n$ obligors, and denote by $L_{n}$ the portfolio percentage loss.  Let $Y$ be a random variable with continuous and strictly increasing distribution function~$H$, and denote by $\varphi_{n}(Y)$ the conditional expectation of portfolio percentage loss $\E[L_{n} \given Y]$.  Assume that the following conditions hold:
\begin{enumerate}[label=(\arabic*)]
	\item\label{item_loss_cvrg_cond_exp}
		$\displaystyle\lim_{n\rightarrow\infty} \left(L_{n} - \sum_{i = 1}^{n} \varphi_{n}(Y)\right) = 0, \quad\P\text{-a.s.}$
	\item\label{item_cond_exp_diffbl}  There is an open interval $I$ containing $H^{-1}(1-\alpha)$, $\alpha\in(0,1)$, and $N_{0}\in\N$ such that whenever $n > N_{0}$ the conditional expectation of portfolio percentage loss, ${\varphi_{n}(y)=\E[L_{n} \given Y=y]}$, is strictly decreasing in $y$ and differentiable on $I$.
	 \item\label{item_cond_exp_bndd_deriv} There is an $N_{0}\in\N$ such that whenever $n > N_{0}$,
	\begin{equation}\label{eqn_cond_exp_bndd_deriv}
		-\infty < -\ub \leq \varphi_{n}^{\prime}(y) \leq -\lb < 0
	\end{equation}
	for all $y \in I$, with $\lb > 0$ and $\ub > 0$ independent of $n$, and where $\varphi_{n}^{\prime}(y)$ denotes the derivative of $\E[L_{n} \given Y=y]$ with respect to $y$.
\end{enumerate}
Then, 
\begin{equation}\label{eqn_port_loss_alpha}
	\lim_{n\rightarrow\infty} \P\left({L_{n}} \leq \varphi_{n}\big(H^{-1}(1\!-\!\alpha)\big)\right) = \alpha,
\end{equation}
and
\begin{equation}\label{eqn_cvrg_var_cond_exp}
	\lim_{n\rightarrow\infty} \big|\VaR[\alpha]{L_{n}} - \varphi_{n}\big(H^{-1}(1\!-\!\alpha)\big)\big| = 0.
\end{equation}
\end{proposition}

\begin{remark}\label{rem_gordy_prop_gen_app}
Proposition~\ref{prop_gordy_prop_5} applies more generally than to the conditional independence model of Definition~\ref{def_cond_indep_model_port_loss}.  However, we are not aware of competing models for which the conditional expectation function $\varphi_{n}$ satisfies Condition~\ref{item_loss_cvrg_cond_exp}.  
\end{remark}

\begin{remark}\label{rem_reiter_bnds_pd_corr}
Condition~\ref{item_loss_cvrg_cond_exp} of Proposition~\ref{prop_gordy_prop_5} postulates that portfolio percentage loss converges, almost surely, to its conditional expectation as the portfolio approaches asymptotic granularity.  The proof of Proposition~\ref{prop_gordy_prop_5} only requires convergence in probability, but we assume almost sure convergence consistent with Proposition~\ref{prop_cond_port_loss}.  For an asymptotic credit portfolio (Definition~\ref{def_asymp_port}), the conditional independence model of Definition~\ref{def_cond_indep_model_port_loss} satisfies Condition~\ref{item_loss_cvrg_cond_exp} by Proposition~\ref{prop_cond_port_loss}.  

Conditions~\ref{item_cond_exp_diffbl} and~\ref{item_cond_exp_bndd_deriv} are more technical in nature.  Condition~\ref{item_cond_exp_diffbl} postulates that the conditional expectation of portfolio percentage loss rises ``smoothly'' as the economy deteriorates for a set of states of the economy associated with the tail of the portfolio loss distribution, which is the concern of solvency assessment.  In relation to the conditional independence model of Definition~\ref{def_cond_indep_model_port_loss}, we quite reasonably assume that the conditional expectation function~$\varphi_{n}$ satisfies Condition~\ref{item_cond_exp_diffbl}.  By inspection of~\eqref{eqn_cond_pd} we observe that Condition~\ref{item_cond_exp_bndd_deriv} holds if $p_{i}$ and $\gamma_{i}^{2}$ are bounded away from zero and one for ${i = 1, \ldots, n}$.
\end{remark}

The proof of Proposition~\ref{prop_gordy_prop_5}, presented below, requires a result due to \citet{PVV95}, which we proceed to state.  Note that in Lemma~\ref{lem_dist_abs_diff}, variables $X$ and $Z$ and functions $F$ and $G$ denote arbitrary random variables and distribution functions, and in Proposition~\ref{prop_gordy_prop_5}, $F_{n}$ and $G_{n}$ denote distribution functions distinct from the notation adopted in Section~\ref{sect_gen_set}.

\begin{lemma}\label{lem_dist_abs_diff}
Let $X$ and $Z$ be random variables defined on a common probability space with distribution functions F and G, respectively.  For all $a \in \R$ and $\varepsilon > 0$,
\begin{equation}\label{eqn_dist_abs_diff}
	\big|F(a) - G(a)\big| \leq \P\big(|X - Z| > \varepsilon\big) + \max\left\{G(a\!+\!\varepsilon) - G(a), G(a) - G(a\!-\!\varepsilon)\right\}.
\end{equation}
\end{lemma}

\begin{proof}
See \citet[Lemma~1.8]{PVV95}.
\end{proof}

\begin{proof}[Proof of Proposition~\ref{prop_gordy_prop_5}]
Denote by $F_{n}$ and $G_{n}$ the distribution functions of $L_{n}$ and $\varphi_{n}(Y)$, respectively.  Firstly, by appealing to Lemma~\ref{lem_dist_abs_diff}, we show that
\begin{equation*}
	\big|F_{n}\big(\varphi_{n}(y)\big) - G_{n}\big(\varphi_{n}(y)\big)\big| \rightarrow 0 \textnormal{ as } n\rightarrow\infty,
\end{equation*}
for all $y \in I$, and deduce~\eqref{eqn_port_loss_alpha}.  Set $X = L_n$, $Z = \varphi_{n}(Y)$ and $a = \varphi_{n}(y)$ in Lemma~\ref{lem_dist_abs_diff}.  Then, for any $\varepsilon > 0$,
\begin{IEEEeqnarray*}{rCl} 
	\IEEEeqnarraymulticol{3}{l}{\big|F_{n}\big(\varphi_{n}(y)\big) - G_{n}\big(\varphi_{n}(y)\big)\big|} \\ \qquad
	& \leq & \P\big(\big|L_{n} - \varphi_{n}(Y)\big| > \varepsilon\big) \\
	& + & \max\left\{G_{n}\big(\varphi_{n}(y) + \varepsilon\big) - G_{n}\big(\varphi_{n}(y)\big), G_{n}\big(\varphi_{n}(y)\big) - G_{n}\big(\varphi_{n}(y) - \varepsilon\big)\right\}.\IEEEyesnumber\label{eqn_cdfs_abs_diff}
\end{IEEEeqnarray*}

The almost sure convergence asserted by Condition~\ref{item_loss_cvrg_cond_exp} implies convergence in probability \citep[Theorem~7.2.3]{GS01}.  By the definition of convergence in probability \citep[Definition~7.2.1]{GS01}, for any $\xi > 0$ and $\varepsilon > 0$, choose an $N_{1} \in \N$ (which in general depends on $\varepsilon$ and $\xi$) such that $n > N_{1}$ implies
\begin{equation}\label{eqn_cvrg_in_prob_eps}
	\P\big(\big|L_{n} - \E[L_n \given Y]\big| > \varepsilon\big) < \frac{\xi}{2}
\end{equation}
for all $y \in I$.

The convergence of $G_{n}$ in the neighbourhood of $\varphi_{n}\big(H^{-1}(1\!-\!\alpha)\big)$ does not immediately follow from the assumption that $\varphi_{n}(y)$ is differentiable on $I$.  By hypothesis, $\varphi_{n}(y)$ is differentiable on I and $\varphi_{n}^{\prime}(y)$ has an upper bound $-\lb$ on $I$ whenever $n > N_{0}$.  By the mean value theorem \citep[see, e.g.,][Theorem~4.15]{WWR04}, for any $\varepsilon > 0$ satisfying ${(y - \varepsilon/\lb, y + \varepsilon/\lb) \subset I}$, there is a ${y^{*} \in (y - \varepsilon/\lb, y)}$ such that
\begin{equation*}
	\varphi_{n}(y - \varepsilon/\lb) - \varphi_{n}(y) = \varphi_{n}^{\prime}(y^{*})(-\varepsilon/\lb) \geq -\lb(-\varepsilon/\lb) = \varepsilon
\end{equation*}
for $n >N_{0}$.  Similarly, there is a  $y^{*} \in (y, y + \varepsilon/\lb)$ such that
\begin{equation*}
	\varphi_{n}(y) - \varphi_{n}(y + \varepsilon/\lb) = \varphi_{n}^{\prime}(y^{*})(-\varepsilon/\lb) \geq -\lb(-\varepsilon/\lb) = \varepsilon
\end{equation*}
for $n > N_{0}$.  By hypothesis, $\varphi_{n}(y)$ is strictly decreasing in~$y$ on $I$ whenever $n > N_{0}$, implying that ${\varphi_{n}(Y) \leq \varphi_{n}(y)}$ if and only if $Y \geq y$.  Moreover, $H$ is continuous and strictly increasing.  Accordingly, 
\begin{equation}\label{eqn_cdfs_an_y}
	G_{n}\big(\varphi_{n}(y)\big) = \P(Y \geq y) = 1 - H(y).
\end{equation}
Hence,
\begin{equation*}
	G_{n}\big(\varphi_{n}(y) + \varepsilon\big) \leq G_{n}\big(\varphi_{n}(y - \varepsilon/\lb)\big) = \P(Y \geq y -  \varepsilon/\lb)
\end{equation*}
and
\begin{equation*}
	G_{n}\big(\varphi_{n}(y) - \varepsilon\big) \geq G_{n}\big(\varphi_{n}(y + \varepsilon/\lb)\big) = \P(Y \geq y +  \varepsilon/\lb)
\end{equation*}
for $n > N_{0}$.  It follows that the continuity of $H(y)$ on $I$ implies convergence.  That is, for any $\xi > 0$ and $\varepsilon > 0$, there is an $N_{0} \in \N$ (which in general depends on $\varepsilon$ and~$\xi$) such that $n > N_{0}$ implies
\begin{IEEEeqnarray*}{rCl}
	\IEEEeqnarraymulticol{3}{l}{\max\left\{G_{n}\big(\varphi_{n}(y) + \varepsilon\big) - G_{n}\big(\varphi_{n}(y)\big), G_{n}\big(\varphi_{n}(y)\big) - G_{n}\big(\varphi_{n}(y) - \varepsilon\big)\right\}} \\ \qquad
	& \leq & \max\left\{\P(Y \geq y-\varepsilon/\lb) - \P(Y \geq y),  \P(Y \geq y) - \P(Y \geq y+\varepsilon/\lb)\right\} \\ \qquad
	& = & \max\left\{H(y) - H(y-\varepsilon/\lb), H(y+\varepsilon/\lb) -H(y)\right\} < \frac{\xi}{2}\IEEEyesnumber\label{eqn_max_cdfs_cvrg}
\end{IEEEeqnarray*}
for any $y \in I$.  Combining \eqref{eqn_cvrg_in_prob_eps} and \eqref{eqn_max_cdfs_cvrg} to evaluate \eqref{eqn_cdfs_abs_diff}, we claim that for any $\xi > 0$, whenever ${n > N = \max\{N_{0}, N_{1}\}}$,
\begin{equation}\label{eqn_cdfs_cvrg_eps}
	\big|F_{n}\big(\varphi_{n}(y)\big) - G_{n}\big(\varphi_{n}(y)\big)\big| = \big|\P\big(L_{n} \leq \varphi_{n}(y)\big) - \P\big(\varphi_{n}(Y) \leq \varphi_{n}(y)\big)\big| < \frac{\xi}{2} + \frac{\xi}{2} = \xi
\end{equation}
for any $y \in I$.  Substituting for $G_{n}\big(\varphi_{n}(y)\big)$ in~\eqref{eqn_cdfs_cvrg_eps} from~\eqref{eqn_cdfs_an_y}, and setting $y = H^{-1}(1\!-\!\alpha) \in I$, \eqref{eqn_cdfs_cvrg_eps} may be expressed as 
\begin{equation}\label{eqn_cdfs_cvrg_lim}
	\lim_{n\rightarrow\infty} \big|\P\left(L_{n} \leq \varphi_{n}\big(H^{-1}(1\!-\!\alpha)\big)\right) - \P\big(Y \geq H^{-1}(1\!-\!\alpha)\big)\big| = 0.
\end{equation}
Now, observing that $\P\big(Y \geq H^{-1}(1\!-\!\alpha)\big) = \alpha$ establishes~\eqref{eqn_port_loss_alpha}, and completes the first part of the proof.

It remains to deduce~\eqref{eqn_cvrg_var_cond_exp}.  Notice that whenever $n > N$, $(y - \varepsilon/\ub, y + \varepsilon/\ub) \subseteq (y - \varepsilon/\lb, y + \varepsilon/\lb)$ and $\varphi_{n}^{\prime}(y)$ has a lower bound $-\ub$ on $I$.  Again, by the mean value theorem, for any $\varepsilon > 0$ satisfying ${(y - \varepsilon/\ub, y + \varepsilon/\ub) \subset I}$, there is a ${y^{*} \in (y - \varepsilon/\ub, y)}$ such that
\begin{equation*}
	\varphi_{n}(y - \varepsilon/\ub) - \varphi_{n}(y) = \varphi_{n}^{\prime}(y^{*})(-\varepsilon/\ub) \leq -\ub(-\varepsilon/\ub) = \varepsilon
\end{equation*}
for $n > N$.  Similarly, there is a ${y^{*} \in (y, y + \varepsilon/\ub)}$ such that
\begin{equation*}
	\varphi_{n}(y) - \varphi_{n}(y + \varepsilon/\ub) = \varphi_{n}^{\prime}(y^{*})(-\varepsilon/\ub) \leq -\ub(-\varepsilon/\ub) = \varepsilon
\end{equation*}
for $n > N$.  Setting ${y = H^{-1}(1\!-\!\alpha) \in I}$, corresponding to the $\alpha$ confidence level in Definition~\ref{def_cred_var}, yields
\begin{IEEEeqnarray*}{rCl}
	F_{n}\big(\varphi_{n}(H^{-1}(1\!-\!\alpha)) + \varepsilon\big) & \geq & F_{n}\big(\varphi_{n}(H^{-1}(1\!-\!\alpha) - \varepsilon/\ub)\big)	\\
	& = & G_{n}\big(\varphi_{n}(H^{-1}(1\!-\!\alpha) - \varepsilon/\ub)\big)	\\
	& = & \P\big(Y \geq H^{-1}(1\!-\!\alpha) - \varepsilon/\ub\big)	\\
	& > & \alpha.\IEEEyesnumber\label{eqn_appx_cdf_alpha_cvrg_above}
\end{IEEEeqnarray*}
The first equality is a consequence of~\eqref{eqn_cdfs_cvrg_lim}, the second of~\eqref{eqn_cdfs_an_y}.  It follows from Definition~\ref{def_cred_var} and~\eqref{eqn_appx_cdf_alpha_cvrg_above} that whenever $n > N$,
\begin{equation}\label{eqn_appx_var_alpha_cvrg_above}
	\VaR[\alpha]{L_{n}} < \varphi_{n}\big(H^{-1}(1\!-\!\alpha)\big) + \varepsilon.
\end{equation}
By a parallel argument, whenever $n > N$,
\begin{IEEEeqnarray*}{rCl}
	F_{n}\big(\varphi_{n}(H^{-1}(1\!-\!\alpha)) - \varepsilon\big) & \leq & F_{n}\big(\varphi_{n}(H^{-1}(1\!-\!\alpha) + \varepsilon/\ub)\big)	\\
	& = & G_{n}\big(\varphi_{n}(H^{-1}(1\!-\!\alpha) + \varepsilon/\ub)\big)	\\
	& = & \P\big(Y \geq H^{-1}(1\!-\!\alpha) + \varepsilon/\ub\big)	\\
	& < & \alpha,\IEEEyesnumber\label{eqn_appx_cdf_alpha_cvrg_below}
\end{IEEEeqnarray*}
and
\begin{equation}\label{eqn_appx_var_alpha_cvrg_below}
	\VaR[\alpha]{L_{n}} > \varphi_{n}\big(H^{-1}(1\!-\!\alpha)\big) - \varepsilon.
\end{equation}
Together, \eqref{eqn_appx_var_alpha_cvrg_above} and \eqref{eqn_appx_var_alpha_cvrg_below} may be expressed as
\begin{equation*}
	\big|\VaR[\alpha]{L_{n}} - \varphi_{n}\big(H^{-1}(1\!-\!\alpha)\big)\big| < \varepsilon.
\end{equation*}
Setting $\varepsilon$ arbitrarily close to zero establishes~\eqref{eqn_cvrg_var_cond_exp}, and completes the proof.
\end{proof}

\section{Proof of Lemma~\ref{lem_sfc_model} and Corollary~\ref{cor_ofgc}}\label{appx_sfc_model_ofgc}
Lemma~\ref{lem_sfc_model} derives the single-factor copula model describing default dependence for the general case.  Then, Corollary~\ref{cor_ofgc} deals with the special case of the one-factor Gaussian copula, the most commonly applied copula function in credit risk modelling.

\sfcmodel*		

\begin{proof}
Integrating conditional marginal distribution functions over systematic risk factor~$Y$, copula~\eqref{eqn_dflt_cop} may be expressed as
\begin{IEEEeqnarray*}{rCl}
	C(u_{1}, \ldots, u_{n}) & = & \P\big(W_{1} < F_{1}^{-1}(u_{1}), \ldots, W_{n} < F_{n}^{-1}(u_{n})\big) \\
	& = & \int_{-\infty}^{\infty} \P\big(W_{1}<F_{1}^{-1}(u_{1}), \ldots, W_{n}<F_{n}^{-1}(u_{n}) \given Y=y\big) \dee{H(y)} \\
	& = & \int_{-\infty}^{\infty} \P\big(W_{1}<F_{1}^{-1}(u_{1}) \given Y=y\big) \ldots \P\big(W_{n}<F_{n}^{-1}(u_{n}) \given Y=y\big) \dee{H(y)}. \\
	& & \IEEEyesnumber\label{eqn_sfc_model_proof}
\end{IEEEeqnarray*}
The integrand in the last equality is the product of obligor PDs conditional on realisation~$y \in \R$ of systematic risk factor~$Y$.  Substituting the probability statements in this integrand with the expression for conditional PD in~\eqref{eqn_cond_pd}, and observing that  $u_{i} = p_{i}$ for $i = 1, \ldots, n$ establishes~\eqref{eqn_sfc_model}.
\end{proof}

\ofgc*	

\begin{proof}
Substituting the operands of the product operator in~\eqref{eqn_sfc_model} with the expression for conditional PD in \eqref{eqn_cond_pd_gauss}, and observing that $\dee{\Phi(y)} = \phi(y)\dee{y}$ establishes~\eqref{eqn_ofgc}.
\end{proof}

\section{Proof of Proposition~\ref{prop_cap_cred_asymp}}\label{appx_cap_cred_asymp}
Appealing to Definition~\ref{def_unexp_loss} and Proposition~\ref{prop_qntl_loss_dist_subst}, we deduce a function for calculating capital held against unexpected losses, $K_{\alpha}(L_{n})$.

\capcredasymp*		

\begin{proof}
By Definition~\ref{def_cap_cred_asymp}, capital held against unexpected losses on an asymptotic credit portfolio may be expressed as
\begin{equation}\label{eqn_cap_cred_var_asymp}
	\lim_{n\rightarrow\infty}K_{\alpha}(L_{n}) = \lim_{n\rightarrow\infty}\VaR[\alpha]{L_{n}} - \lim_{n\rightarrow\infty}\E[L_{n}],
\end{equation}
assuming that the limits on the right-hand side are well defined.  By hypothesis, defaults are modelled as conditionally independent random variables given systematic risk factor~$Y$.  So, adding and subtracting the $\alpha$ quantile of the distribution of conditional expectation of portfolio percentage loss yields
\begin{IEEEeqnarray*}{rCl}
	\lim_{n\rightarrow\infty}K_{\alpha}(L_{n}) & = & \lim_{n\rightarrow\infty}\left(\VaR[\alpha]{L_{n}} - \E\big[L_{n} \given Y=H^{-1}(1\!-\!\alpha)\big]\right)	\\
	& & + \lim_{n\rightarrow\infty}\E\big[L_{n} \given Y=H^{-1}(1\!-\!\alpha)\big] - \lim_{n\rightarrow\infty}\E[L_{n}]	\\
	& = & \lim_{n\rightarrow\infty}\E\big[L_{n} \given Y=H^{-1}(1\!-\!\alpha)\big] - \lim_{n\rightarrow\infty}\E[L_{n}],\IEEEyesnumber\label{eqn_cap_cred_var_subst}
\end{IEEEeqnarray*}
where the limits on the right-hand side of~\eqref{eqn_cap_cred_var_subst} are well defined by hypothesis.  Recall that in probabilistic terms, $\VaR[\alpha]{L_{n}}$ is the $\alpha$ quantile of the portfolio loss distribution, $\q{\alpha}{L_{n}}$.  Hence, the second equality in~\eqref{eqn_cap_cred_var_subst} follows from Proposition~\ref{prop_qntl_loss_dist_subst}, with the $\alpha$ quantile of the distribution of $\E[L_{n} \given Y]$, in effect, being substituted for $\VaR[\alpha]{L_{n}}$.  Expanding~\eqref{eqn_cap_cred_var_subst} using~\eqref{eqn_exp_port_loss} and~\eqref{eqn_cond_exp_port_loss} yields
\begin{IEEEeqnarray*}{rCl}
	\lim_{n\rightarrow\infty}K_{\alpha}(L_{n}) & = & \lim_{n\rightarrow\infty}\sum_{i = 1}^{n} w_{i}\lgd_{i}p_{i}\big(H^{-1}(1\!-\!\alpha)\big) - \lim_{n\rightarrow\infty}\sum_{i = 1}^{n} w_{i}\lgd_{i}p_{i}.\IEEEyesnumber\label{eqn_cap_cred_cond_exp}
\end{IEEEeqnarray*}
Finally, setting $y=H^{-1}(1\!-\!\alpha)$ in \eqref{eqn_cond_pd} and substituting into~\eqref{eqn_cap_cred_cond_exp} establishes~\eqref{eqn_cap_cred_asymp}.
\end{proof}



\end{document}

%% file: plot_gauss_cop.tex
\begingroup
  \makeatletter
  \providecommand\color[2][]{%
    \GenericError{(gnuplot) \space\space\space\@spaces}{%
      Package color not loaded in conjunction with
      terminal option `colourtext'%
    }{See the gnuplot documentation for explanation.%
    }{Either use 'blacktext' in gnuplot or load the package
      color.sty in LaTeX.}%
    \renewcommand\color[2][]{}%
  }%
  \providecommand\includegraphics[2][]{%
    \GenericError{(gnuplot) \space\space\space\@spaces}{%
      Package graphicx or graphics not loaded%
    }{See the gnuplot documentation for explanation.%
    }{The gnuplot epslatex terminal needs graphicx.sty or graphics.sty.}%
    \renewcommand\includegraphics[2][]{}%
  }%
  \providecommand\rotatebox[2]{#2}%
  \@ifundefined{ifGPcolor}{%
    \newif\ifGPcolor
    \GPcolorfalse
  }{}%
  \@ifundefined{ifGPblacktext}{%
    \newif\ifGPblacktext
    \GPblacktexttrue
  }{}%
  \let\gplgaddtomacro\g@addto@macro
  \gdef\gplbacktext{}%
  \gdef\gplfronttext{}%
  \makeatother
  \ifGPblacktext
    \def\colorrgb#1{}%
    \def\colorgray#1{}%
  \else
    \ifGPcolor
      \def\colorrgb#1{\color[rgb]{#1}}%
      \def\colorgray#1{\color[gray]{#1}}%
      \expandafter\def\csname LTw\endcsname{\color{white}}%
      \expandafter\def\csname LTb\endcsname{\color{black}}%
      \expandafter\def\csname LTa\endcsname{\color{black}}%
      \expandafter\def\csname LT0\endcsname{\color[rgb]{1,0,0}}%
      \expandafter\def\csname LT1\endcsname{\color[rgb]{0,1,0}}%
      \expandafter\def\csname LT2\endcsname{\color[rgb]{0,0,1}}%
      \expandafter\def\csname LT3\endcsname{\color[rgb]{1,0,1}}%
      \expandafter\def\csname LT4\endcsname{\color[rgb]{0,1,1}}%
      \expandafter\def\csname LT5\endcsname{\color[rgb]{1,1,0}}%
      \expandafter\def\csname LT6\endcsname{\color[rgb]{0,0,0}}%
      \expandafter\def\csname LT7\endcsname{\color[rgb]{1,0.3,0}}%
      \expandafter\def\csname LT8\endcsname{\color[rgb]{0.5,0.5,0.5}}%
    \else
      \def\colorrgb#1{\color{black}}%
      \def\colorgray#1{\color[gray]{#1}}%
      \expandafter\def\csname LTw\endcsname{\color{white}}%
      \expandafter\def\csname LTb\endcsname{\color{black}}%
      \expandafter\def\csname LTa\endcsname{\color{black}}%
      \expandafter\def\csname LT0\endcsname{\color{black}}%
      \expandafter\def\csname LT1\endcsname{\color{black}}%
      \expandafter\def\csname LT2\endcsname{\color{black}}%
      \expandafter\def\csname LT3\endcsname{\color{black}}%
      \expandafter\def\csname LT4\endcsname{\color{black}}%
      \expandafter\def\csname LT5\endcsname{\color{black}}%
      \expandafter\def\csname LT6\endcsname{\color{black}}%
      \expandafter\def\csname LT7\endcsname{\color{black}}%
      \expandafter\def\csname LT8\endcsname{\color{black}}%
    \fi
  \fi
  \setlength{\unitlength}{0.0500bp}%
  \begin{picture}(5442.00,5442.00)%
    \gplgaddtomacro\gplbacktext{%
      \csname LTb\endcsname%
      \put(946,704){\makebox(0,0)[r]{\strut{}-5.0}}%
      \put(946,1112){\makebox(0,0)[r]{\strut{}-4.0}}%
      \put(946,1519){\makebox(0,0)[r]{\strut{}-3.0}}%
      \put(946,1927){\makebox(0,0)[r]{\strut{}-2.0}}%
      \put(946,2335){\makebox(0,0)[r]{\strut{}-1.0}}%
      \put(946,2743){\makebox(0,0)[r]{\strut{}0.0}}%
      \put(946,3150){\makebox(0,0)[r]{\strut{}1.0}}%
      \put(946,3558){\makebox(0,0)[r]{\strut{}2.0}}%
      \put(946,3966){\makebox(0,0)[r]{\strut{}3.0}}%
      \put(946,4373){\makebox(0,0)[r]{\strut{}4.0}}%
      \put(946,4781){\makebox(0,0)[r]{\strut{}5.0}}%
      \put(1078,484){\makebox(0,0){\strut{}-5.0}}%
      \put(1475,484){\makebox(0,0){\strut{}-4.0}}%
      \put(1871,484){\makebox(0,0){\strut{}-3.0}}%
      \put(2268,484){\makebox(0,0){\strut{}-2.0}}%
      \put(2665,484){\makebox(0,0){\strut{}-1.0}}%
      \put(3062,484){\makebox(0,0){\strut{}0.0}}%
      \put(3458,484){\makebox(0,0){\strut{}1.0}}%
      \put(3855,484){\makebox(0,0){\strut{}2.0}}%
      \put(4252,484){\makebox(0,0){\strut{}3.0}}%
      \put(4648,484){\makebox(0,0){\strut{}4.0}}%
      \put(5045,484){\makebox(0,0){\strut{}5.0}}%
      \put(176,2742){\rotatebox{-270}{\makebox(0,0){\strut{}$X_{2}$}}}%
      \put(3061,154){\makebox(0,0){\strut{}$X_{1}$}}%
      \put(3061,5111){\makebox(0,0){\strut{}One-factor Gaussian copula}}%
    }%
    \gplgaddtomacro\gplfronttext{%
    }%
    \gplbacktext
    \put(0,0){\includegraphics{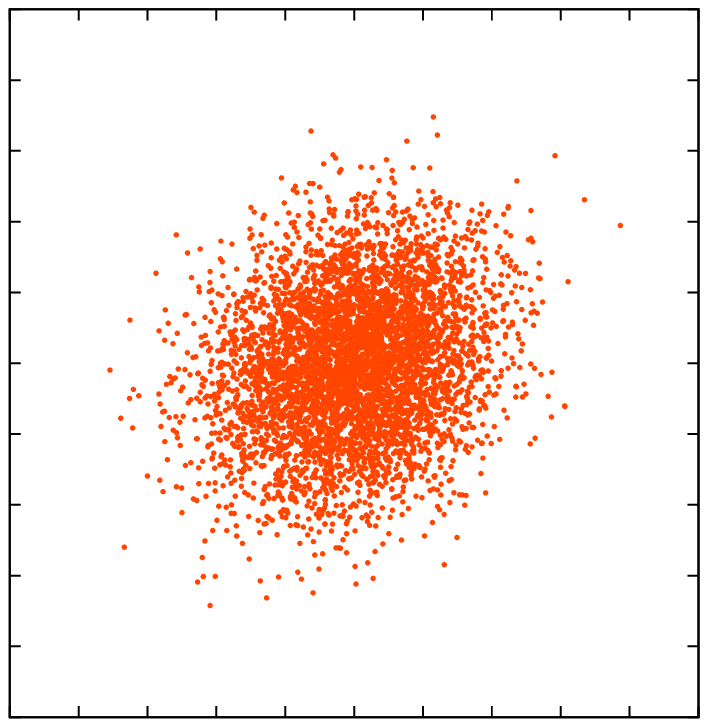}}%
    \gplfronttext
  \end{picture}%
\endgroup

%% file: plot_t30_cop_norm.tex
\begingroup
  \makeatletter
  \providecommand\color[2][]{%
    \GenericError{(gnuplot) \space\space\space\@spaces}{%
      Package color not loaded in conjunction with
      terminal option `colourtext'%
    }{See the gnuplot documentation for explanation.%
    }{Either use 'blacktext' in gnuplot or load the package
      color.sty in LaTeX.}%
    \renewcommand\color[2][]{}%
  }%
  \providecommand\includegraphics[2][]{%
    \GenericError{(gnuplot) \space\space\space\@spaces}{%
      Package graphicx or graphics not loaded%
    }{See the gnuplot documentation for explanation.%
    }{The gnuplot epslatex terminal needs graphicx.sty or graphics.sty.}%
    \renewcommand\includegraphics[2][]{}%
  }%
  \providecommand\rotatebox[2]{#2}%
  \@ifundefined{ifGPcolor}{%
    \newif\ifGPcolor
    \GPcolorfalse
  }{}%
  \@ifundefined{ifGPblacktext}{%
    \newif\ifGPblacktext
    \GPblacktexttrue
  }{}%
  \let\gplgaddtomacro\g@addto@macro
  \gdef\gplbacktext{}%
  \gdef\gplfronttext{}%
  \makeatother
  \ifGPblacktext
    \def\colorrgb#1{}%
    \def\colorgray#1{}%
  \else
    \ifGPcolor
      \def\colorrgb#1{\color[rgb]{#1}}%
      \def\colorgray#1{\color[gray]{#1}}%
      \expandafter\def\csname LTw\endcsname{\color{white}}%
      \expandafter\def\csname LTb\endcsname{\color{black}}%
      \expandafter\def\csname LTa\endcsname{\color{black}}%
      \expandafter\def\csname LT0\endcsname{\color[rgb]{1,0,0}}%
      \expandafter\def\csname LT1\endcsname{\color[rgb]{0,1,0}}%
      \expandafter\def\csname LT2\endcsname{\color[rgb]{0,0,1}}%
      \expandafter\def\csname LT3\endcsname{\color[rgb]{1,0,1}}%
      \expandafter\def\csname LT4\endcsname{\color[rgb]{0,1,1}}%
      \expandafter\def\csname LT5\endcsname{\color[rgb]{1,1,0}}%
      \expandafter\def\csname LT6\endcsname{\color[rgb]{0,0,0}}%
      \expandafter\def\csname LT7\endcsname{\color[rgb]{1,0.3,0}}%
      \expandafter\def\csname LT8\endcsname{\color[rgb]{0.5,0.5,0.5}}%
    \else
      \def\colorrgb#1{\color{black}}%
      \def\colorgray#1{\color[gray]{#1}}%
      \expandafter\def\csname LTw\endcsname{\color{white}}%
      \expandafter\def\csname LTb\endcsname{\color{black}}%
      \expandafter\def\csname LTa\endcsname{\color{black}}%
      \expandafter\def\csname LT0\endcsname{\color{black}}%
      \expandafter\def\csname LT1\endcsname{\color{black}}%
      \expandafter\def\csname LT2\endcsname{\color{black}}%
      \expandafter\def\csname LT3\endcsname{\color{black}}%
      \expandafter\def\csname LT4\endcsname{\color{black}}%
      \expandafter\def\csname LT5\endcsname{\color{black}}%
      \expandafter\def\csname LT6\endcsname{\color{black}}%
      \expandafter\def\csname LT7\endcsname{\color{black}}%
      \expandafter\def\csname LT8\endcsname{\color{black}}%
    \fi
  \fi
  \setlength{\unitlength}{0.0500bp}%
  \begin{picture}(5442.00,5442.00)%
    \gplgaddtomacro\gplbacktext{%
      \csname LTb\endcsname%
      \put(946,704){\makebox(0,0)[r]{\strut{}-5.0}}%
      \put(946,1112){\makebox(0,0)[r]{\strut{}-4.0}}%
      \put(946,1519){\makebox(0,0)[r]{\strut{}-3.0}}%
      \put(946,1927){\makebox(0,0)[r]{\strut{}-2.0}}%
      \put(946,2335){\makebox(0,0)[r]{\strut{}-1.0}}%
      \put(946,2743){\makebox(0,0)[r]{\strut{}0.0}}%
      \put(946,3150){\makebox(0,0)[r]{\strut{}1.0}}%
      \put(946,3558){\makebox(0,0)[r]{\strut{}2.0}}%
      \put(946,3966){\makebox(0,0)[r]{\strut{}3.0}}%
      \put(946,4373){\makebox(0,0)[r]{\strut{}4.0}}%
      \put(946,4781){\makebox(0,0)[r]{\strut{}5.0}}%
      \put(1078,484){\makebox(0,0){\strut{}-5.0}}%
      \put(1475,484){\makebox(0,0){\strut{}-4.0}}%
      \put(1871,484){\makebox(0,0){\strut{}-3.0}}%
      \put(2268,484){\makebox(0,0){\strut{}-2.0}}%
      \put(2665,484){\makebox(0,0){\strut{}-1.0}}%
      \put(3062,484){\makebox(0,0){\strut{}0.0}}%
      \put(3458,484){\makebox(0,0){\strut{}1.0}}%
      \put(3855,484){\makebox(0,0){\strut{}2.0}}%
      \put(4252,484){\makebox(0,0){\strut{}3.0}}%
      \put(4648,484){\makebox(0,0){\strut{}4.0}}%
      \put(5045,484){\makebox(0,0){\strut{}5.0}}%
      \put(176,2742){\rotatebox{-270}{\makebox(0,0){\strut{}$X_{2}$}}}%
      \put(3061,154){\makebox(0,0){\strut{}$X_{1}$}}%
      \put(3061,5111){\makebox(0,0){\strut{}$t$-copula with $\nu = 30$ and Gaussian margins}}%
    }%
    \gplgaddtomacro\gplfronttext{%
    }%
    \gplbacktext
    \put(0,0){\includegraphics{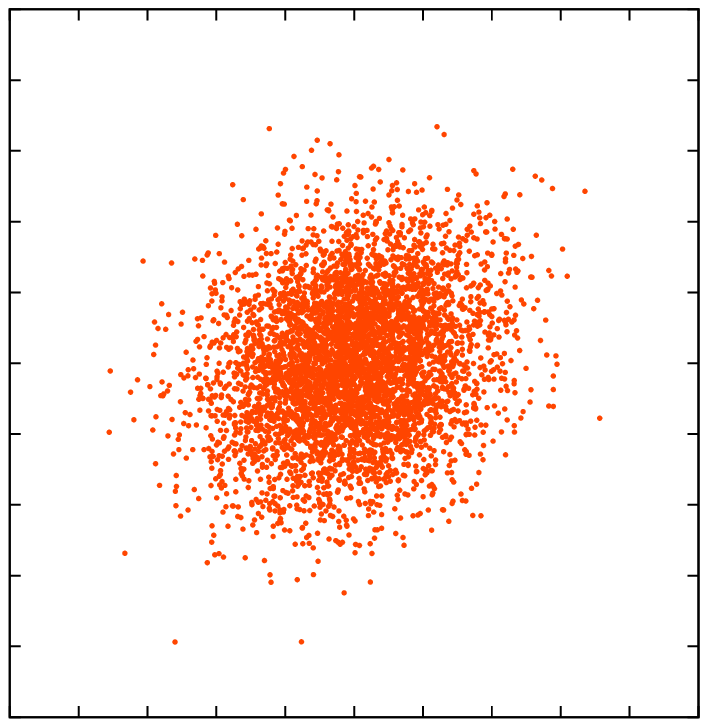}}%
    \gplfronttext
  \end{picture}%
\endgroup

%% file: plot_t10_cop_norm.tex
\begingroup
  \makeatletter
  \providecommand\color[2][]{%
    \GenericError{(gnuplot) \space\space\space\@spaces}{%
      Package color not loaded in conjunction with
      terminal option `colourtext'%
    }{See the gnuplot documentation for explanation.%
    }{Either use 'blacktext' in gnuplot or load the package
      color.sty in LaTeX.}%
    \renewcommand\color[2][]{}%
  }%
  \providecommand\includegraphics[2][]{%
    \GenericError{(gnuplot) \space\space\space\@spaces}{%
      Package graphicx or graphics not loaded%
    }{See the gnuplot documentation for explanation.%
    }{The gnuplot epslatex terminal needs graphicx.sty or graphics.sty.}%
    \renewcommand\includegraphics[2][]{}%
  }%
  \providecommand\rotatebox[2]{#2}%
  \@ifundefined{ifGPcolor}{%
    \newif\ifGPcolor
    \GPcolorfalse
  }{}%
  \@ifundefined{ifGPblacktext}{%
    \newif\ifGPblacktext
    \GPblacktexttrue
  }{}%
  \let\gplgaddtomacro\g@addto@macro
  \gdef\gplbacktext{}%
  \gdef\gplfronttext{}%
  \makeatother
  \ifGPblacktext
    \def\colorrgb#1{}%
    \def\colorgray#1{}%
  \else
    \ifGPcolor
      \def\colorrgb#1{\color[rgb]{#1}}%
      \def\colorgray#1{\color[gray]{#1}}%
      \expandafter\def\csname LTw\endcsname{\color{white}}%
      \expandafter\def\csname LTb\endcsname{\color{black}}%
      \expandafter\def\csname LTa\endcsname{\color{black}}%
      \expandafter\def\csname LT0\endcsname{\color[rgb]{1,0,0}}%
      \expandafter\def\csname LT1\endcsname{\color[rgb]{0,1,0}}%
      \expandafter\def\csname LT2\endcsname{\color[rgb]{0,0,1}}%
      \expandafter\def\csname LT3\endcsname{\color[rgb]{1,0,1}}%
      \expandafter\def\csname LT4\endcsname{\color[rgb]{0,1,1}}%
      \expandafter\def\csname LT5\endcsname{\color[rgb]{1,1,0}}%
      \expandafter\def\csname LT6\endcsname{\color[rgb]{0,0,0}}%
      \expandafter\def\csname LT7\endcsname{\color[rgb]{1,0.3,0}}%
      \expandafter\def\csname LT8\endcsname{\color[rgb]{0.5,0.5,0.5}}%
    \else
      \def\colorrgb#1{\color{black}}%
      \def\colorgray#1{\color[gray]{#1}}%
      \expandafter\def\csname LTw\endcsname{\color{white}}%
      \expandafter\def\csname LTb\endcsname{\color{black}}%
      \expandafter\def\csname LTa\endcsname{\color{black}}%
      \expandafter\def\csname LT0\endcsname{\color{black}}%
      \expandafter\def\csname LT1\endcsname{\color{black}}%
      \expandafter\def\csname LT2\endcsname{\color{black}}%
      \expandafter\def\csname LT3\endcsname{\color{black}}%
      \expandafter\def\csname LT4\endcsname{\color{black}}%
      \expandafter\def\csname LT5\endcsname{\color{black}}%
      \expandafter\def\csname LT6\endcsname{\color{black}}%
      \expandafter\def\csname LT7\endcsname{\color{black}}%
      \expandafter\def\csname LT8\endcsname{\color{black}}%
    \fi
  \fi
  \setlength{\unitlength}{0.0500bp}%
  \begin{picture}(5442.00,5442.00)%
    \gplgaddtomacro\gplbacktext{%
      \csname LTb\endcsname%
      \put(946,704){\makebox(0,0)[r]{\strut{}-5.0}}%
      \put(946,1112){\makebox(0,0)[r]{\strut{}-4.0}}%
      \put(946,1519){\makebox(0,0)[r]{\strut{}-3.0}}%
      \put(946,1927){\makebox(0,0)[r]{\strut{}-2.0}}%
      \put(946,2335){\makebox(0,0)[r]{\strut{}-1.0}}%
      \put(946,2743){\makebox(0,0)[r]{\strut{}0.0}}%
      \put(946,3150){\makebox(0,0)[r]{\strut{}1.0}}%
      \put(946,3558){\makebox(0,0)[r]{\strut{}2.0}}%
      \put(946,3966){\makebox(0,0)[r]{\strut{}3.0}}%
      \put(946,4373){\makebox(0,0)[r]{\strut{}4.0}}%
      \put(946,4781){\makebox(0,0)[r]{\strut{}5.0}}%
      \put(1078,484){\makebox(0,0){\strut{}-5.0}}%
      \put(1475,484){\makebox(0,0){\strut{}-4.0}}%
      \put(1871,484){\makebox(0,0){\strut{}-3.0}}%
      \put(2268,484){\makebox(0,0){\strut{}-2.0}}%
      \put(2665,484){\makebox(0,0){\strut{}-1.0}}%
      \put(3062,484){\makebox(0,0){\strut{}0.0}}%
      \put(3458,484){\makebox(0,0){\strut{}1.0}}%
      \put(3855,484){\makebox(0,0){\strut{}2.0}}%
      \put(4252,484){\makebox(0,0){\strut{}3.0}}%
      \put(4648,484){\makebox(0,0){\strut{}4.0}}%
      \put(5045,484){\makebox(0,0){\strut{}5.0}}%
      \put(176,2742){\rotatebox{-270}{\makebox(0,0){\strut{}$X_{2}$}}}%
      \put(3061,154){\makebox(0,0){\strut{}$X_{1}$}}%
      \put(3061,5111){\makebox(0,0){\strut{}$t$-copula with $\nu = 10$ and Gaussian margins}}%
    }%
    \gplgaddtomacro\gplfronttext{%
    }%
    \gplbacktext
    \put(0,0){\includegraphics{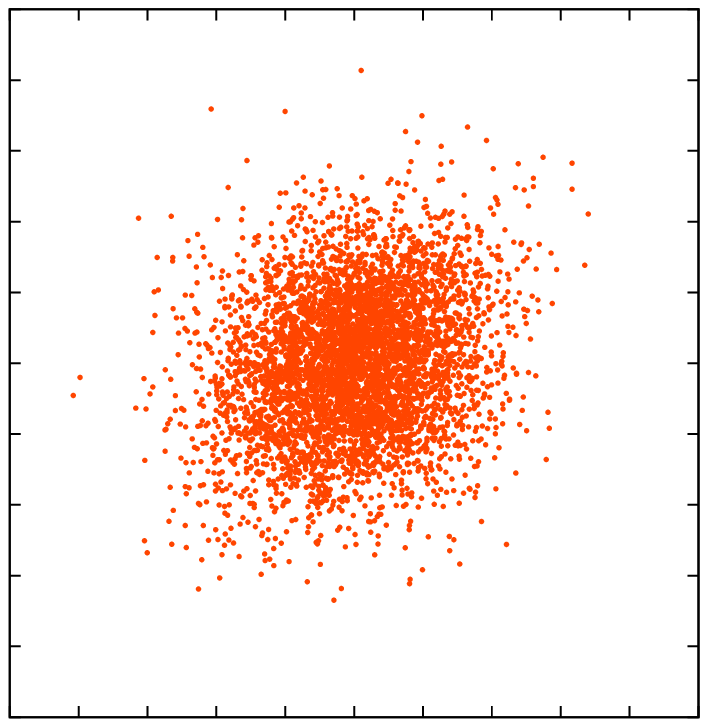}}%
    \gplfronttext
  \end{picture}%
\endgroup

%% file: plot_prod_cop.tex
\begingroup
  \makeatletter
  \providecommand\color[2][]{%
    \GenericError{(gnuplot) \space\space\space\@spaces}{%
      Package color not loaded in conjunction with
      terminal option `colourtext'%
    }{See the gnuplot documentation for explanation.%
    }{Either use 'blacktext' in gnuplot or load the package
      color.sty in LaTeX.}%
    \renewcommand\color[2][]{}%
  }%
  \providecommand\includegraphics[2][]{%
    \GenericError{(gnuplot) \space\space\space\@spaces}{%
      Package graphicx or graphics not loaded%
    }{See the gnuplot documentation for explanation.%
    }{The gnuplot epslatex terminal needs graphicx.sty or graphics.sty.}%
    \renewcommand\includegraphics[2][]{}%
  }%
  \providecommand\rotatebox[2]{#2}%
  \@ifundefined{ifGPcolor}{%
    \newif\ifGPcolor
    \GPcolorfalse
  }{}%
  \@ifundefined{ifGPblacktext}{%
    \newif\ifGPblacktext
    \GPblacktexttrue
  }{}%
  \let\gplgaddtomacro\g@addto@macro
  \gdef\gplbacktext{}%
  \gdef\gplfronttext{}%
  \makeatother
  \ifGPblacktext
    \def\colorrgb#1{}%
    \def\colorgray#1{}%
  \else
    \ifGPcolor
      \def\colorrgb#1{\color[rgb]{#1}}%
      \def\colorgray#1{\color[gray]{#1}}%
      \expandafter\def\csname LTw\endcsname{\color{white}}%
      \expandafter\def\csname LTb\endcsname{\color{black}}%
      \expandafter\def\csname LTa\endcsname{\color{black}}%
      \expandafter\def\csname LT0\endcsname{\color[rgb]{1,0,0}}%
      \expandafter\def\csname LT1\endcsname{\color[rgb]{0,1,0}}%
      \expandafter\def\csname LT2\endcsname{\color[rgb]{0,0,1}}%
      \expandafter\def\csname LT3\endcsname{\color[rgb]{1,0,1}}%
      \expandafter\def\csname LT4\endcsname{\color[rgb]{0,1,1}}%
      \expandafter\def\csname LT5\endcsname{\color[rgb]{1,1,0}}%
      \expandafter\def\csname LT6\endcsname{\color[rgb]{0,0,0}}%
      \expandafter\def\csname LT7\endcsname{\color[rgb]{1,0.3,0}}%
      \expandafter\def\csname LT8\endcsname{\color[rgb]{0.5,0.5,0.5}}%
    \else
      \def\colorrgb#1{\color{black}}%
      \def\colorgray#1{\color[gray]{#1}}%
      \expandafter\def\csname LTw\endcsname{\color{white}}%
      \expandafter\def\csname LTb\endcsname{\color{black}}%
      \expandafter\def\csname LTa\endcsname{\color{black}}%
      \expandafter\def\csname LT0\endcsname{\color{black}}%
      \expandafter\def\csname LT1\endcsname{\color{black}}%
      \expandafter\def\csname LT2\endcsname{\color{black}}%
      \expandafter\def\csname LT3\endcsname{\color{black}}%
      \expandafter\def\csname LT4\endcsname{\color{black}}%
      \expandafter\def\csname LT5\endcsname{\color{black}}%
      \expandafter\def\csname LT6\endcsname{\color{black}}%
      \expandafter\def\csname LT7\endcsname{\color{black}}%
      \expandafter\def\csname LT8\endcsname{\color{black}}%
    \fi
  \fi
  \setlength{\unitlength}{0.0500bp}%
  \begin{picture}(5442.00,5442.00)%
    \gplgaddtomacro\gplbacktext{%
      \csname LTb\endcsname%
      \put(946,704){\makebox(0,0)[r]{\strut{}-5.0}}%
      \put(946,1112){\makebox(0,0)[r]{\strut{}-4.0}}%
      \put(946,1519){\makebox(0,0)[r]{\strut{}-3.0}}%
      \put(946,1927){\makebox(0,0)[r]{\strut{}-2.0}}%
      \put(946,2335){\makebox(0,0)[r]{\strut{}-1.0}}%
      \put(946,2743){\makebox(0,0)[r]{\strut{}0.0}}%
      \put(946,3150){\makebox(0,0)[r]{\strut{}1.0}}%
      \put(946,3558){\makebox(0,0)[r]{\strut{}2.0}}%
      \put(946,3966){\makebox(0,0)[r]{\strut{}3.0}}%
      \put(946,4373){\makebox(0,0)[r]{\strut{}4.0}}%
      \put(946,4781){\makebox(0,0)[r]{\strut{}5.0}}%
      \put(1078,484){\makebox(0,0){\strut{}-5.0}}%
      \put(1475,484){\makebox(0,0){\strut{}-4.0}}%
      \put(1871,484){\makebox(0,0){\strut{}-3.0}}%
      \put(2268,484){\makebox(0,0){\strut{}-2.0}}%
      \put(2665,484){\makebox(0,0){\strut{}-1.0}}%
      \put(3062,484){\makebox(0,0){\strut{}0.0}}%
      \put(3458,484){\makebox(0,0){\strut{}1.0}}%
      \put(3855,484){\makebox(0,0){\strut{}2.0}}%
      \put(4252,484){\makebox(0,0){\strut{}3.0}}%
      \put(4648,484){\makebox(0,0){\strut{}4.0}}%
      \put(5045,484){\makebox(0,0){\strut{}5.0}}%
      \put(176,2742){\rotatebox{-270}{\makebox(0,0){\strut{}$X_{2}$}}}%
      \put(3061,154){\makebox(0,0){\strut{}$X_{1}$}}%
      \put(3061,5111){\makebox(0,0){\strut{}Product copula with Gaussian margins}}%
    }%
    \gplgaddtomacro\gplfronttext{%
    }%
    \gplbacktext
    \put(0,0){\includegraphics{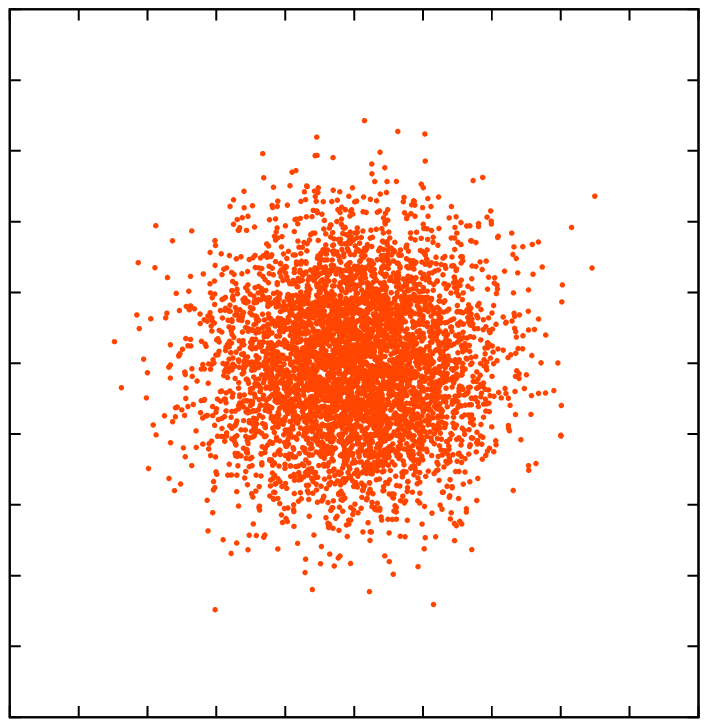}}%
    \gplfronttext
  \end{picture}%
\endgroup

%% file: plot_t3_cop_norm.tex
\begingroup
  \makeatletter
  \providecommand\color[2][]{%
    \GenericError{(gnuplot) \space\space\space\@spaces}{%
      Package color not loaded in conjunction with
      terminal option `colourtext'%
    }{See the gnuplot documentation for explanation.%
    }{Either use 'blacktext' in gnuplot or load the package
      color.sty in LaTeX.}%
    \renewcommand\color[2][]{}%
  }%
  \providecommand\includegraphics[2][]{%
    \GenericError{(gnuplot) \space\space\space\@spaces}{%
      Package graphicx or graphics not loaded%
    }{See the gnuplot documentation for explanation.%
    }{The gnuplot epslatex terminal needs graphicx.sty or graphics.sty.}%
    \renewcommand\includegraphics[2][]{}%
  }%
  \providecommand\rotatebox[2]{#2}%
  \@ifundefined{ifGPcolor}{%
    \newif\ifGPcolor
    \GPcolorfalse
  }{}%
  \@ifundefined{ifGPblacktext}{%
    \newif\ifGPblacktext
    \GPblacktexttrue
  }{}%
  \let\gplgaddtomacro\g@addto@macro
  \gdef\gplbacktext{}%
  \gdef\gplfronttext{}%
  \makeatother
  \ifGPblacktext
    \def\colorrgb#1{}%
    \def\colorgray#1{}%
  \else
    \ifGPcolor
      \def\colorrgb#1{\color[rgb]{#1}}%
      \def\colorgray#1{\color[gray]{#1}}%
      \expandafter\def\csname LTw\endcsname{\color{white}}%
      \expandafter\def\csname LTb\endcsname{\color{black}}%
      \expandafter\def\csname LTa\endcsname{\color{black}}%
      \expandafter\def\csname LT0\endcsname{\color[rgb]{1,0,0}}%
      \expandafter\def\csname LT1\endcsname{\color[rgb]{0,1,0}}%
      \expandafter\def\csname LT2\endcsname{\color[rgb]{0,0,1}}%
      \expandafter\def\csname LT3\endcsname{\color[rgb]{1,0,1}}%
      \expandafter\def\csname LT4\endcsname{\color[rgb]{0,1,1}}%
      \expandafter\def\csname LT5\endcsname{\color[rgb]{1,1,0}}%
      \expandafter\def\csname LT6\endcsname{\color[rgb]{0,0,0}}%
      \expandafter\def\csname LT7\endcsname{\color[rgb]{1,0.3,0}}%
      \expandafter\def\csname LT8\endcsname{\color[rgb]{0.5,0.5,0.5}}%
    \else
      \def\colorrgb#1{\color{black}}%
      \def\colorgray#1{\color[gray]{#1}}%
      \expandafter\def\csname LTw\endcsname{\color{white}}%
      \expandafter\def\csname LTb\endcsname{\color{black}}%
      \expandafter\def\csname LTa\endcsname{\color{black}}%
      \expandafter\def\csname LT0\endcsname{\color{black}}%
      \expandafter\def\csname LT1\endcsname{\color{black}}%
      \expandafter\def\csname LT2\endcsname{\color{black}}%
      \expandafter\def\csname LT3\endcsname{\color{black}}%
      \expandafter\def\csname LT4\endcsname{\color{black}}%
      \expandafter\def\csname LT5\endcsname{\color{black}}%
      \expandafter\def\csname LT6\endcsname{\color{black}}%
      \expandafter\def\csname LT7\endcsname{\color{black}}%
      \expandafter\def\csname LT8\endcsname{\color{black}}%
    \fi
  \fi
  \setlength{\unitlength}{0.0500bp}%
  \begin{picture}(5442.00,5442.00)%
    \gplgaddtomacro\gplbacktext{%
      \csname LTb\endcsname%
      \put(946,704){\makebox(0,0)[r]{\strut{}-5.0}}%
      \put(946,1112){\makebox(0,0)[r]{\strut{}-4.0}}%
      \put(946,1519){\makebox(0,0)[r]{\strut{}-3.0}}%
      \put(946,1927){\makebox(0,0)[r]{\strut{}-2.0}}%
      \put(946,2335){\makebox(0,0)[r]{\strut{}-1.0}}%
      \put(946,2743){\makebox(0,0)[r]{\strut{}0.0}}%
      \put(946,3150){\makebox(0,0)[r]{\strut{}1.0}}%
      \put(946,3558){\makebox(0,0)[r]{\strut{}2.0}}%
      \put(946,3966){\makebox(0,0)[r]{\strut{}3.0}}%
      \put(946,4373){\makebox(0,0)[r]{\strut{}4.0}}%
      \put(946,4781){\makebox(0,0)[r]{\strut{}5.0}}%
      \put(1078,484){\makebox(0,0){\strut{}-5.0}}%
      \put(1475,484){\makebox(0,0){\strut{}-4.0}}%
      \put(1871,484){\makebox(0,0){\strut{}-3.0}}%
      \put(2268,484){\makebox(0,0){\strut{}-2.0}}%
      \put(2665,484){\makebox(0,0){\strut{}-1.0}}%
      \put(3062,484){\makebox(0,0){\strut{}0.0}}%
      \put(3458,484){\makebox(0,0){\strut{}1.0}}%
      \put(3855,484){\makebox(0,0){\strut{}2.0}}%
      \put(4252,484){\makebox(0,0){\strut{}3.0}}%
      \put(4648,484){\makebox(0,0){\strut{}4.0}}%
      \put(5045,484){\makebox(0,0){\strut{}5.0}}%
      \put(176,2742){\rotatebox{-270}{\makebox(0,0){\strut{}$X_{2}$}}}%
      \put(3061,154){\makebox(0,0){\strut{}$X_{1}$}}%
      \put(3061,5111){\makebox(0,0){\strut{}$t$-copula with $\nu = 3$ and Gaussian margins}}%
    }%
    \gplgaddtomacro\gplfronttext{%
    }%
    \gplbacktext
    \put(0,0){\includegraphics{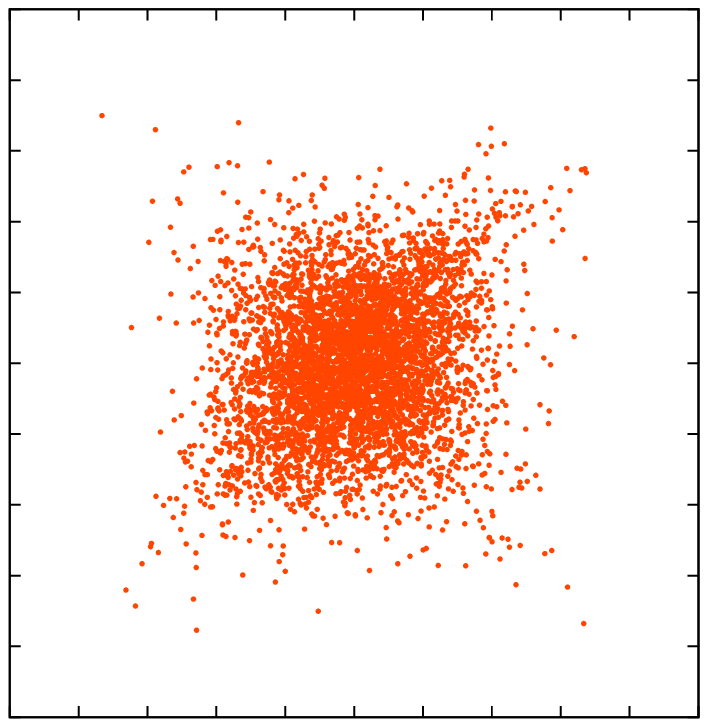}}%
    \gplfronttext
  \end{picture}%
\endgroup

%% file: plot_t10_cop.tex
\begingroup
  \makeatletter
  \providecommand\color[2][]{%
    \GenericError{(gnuplot) \space\space\space\@spaces}{%
      Package color not loaded in conjunction with
      terminal option `colourtext'%
    }{See the gnuplot documentation for explanation.%
    }{Either use 'blacktext' in gnuplot or load the package
      color.sty in LaTeX.}%
    \renewcommand\color[2][]{}%
  }%
  \providecommand\includegraphics[2][]{%
    \GenericError{(gnuplot) \space\space\space\@spaces}{%
      Package graphicx or graphics not loaded%
    }{See the gnuplot documentation for explanation.%
    }{The gnuplot epslatex terminal needs graphicx.sty or graphics.sty.}%
    \renewcommand\includegraphics[2][]{}%
  }%
  \providecommand\rotatebox[2]{#2}%
  \@ifundefined{ifGPcolor}{%
    \newif\ifGPcolor
    \GPcolorfalse
  }{}%
  \@ifundefined{ifGPblacktext}{%
    \newif\ifGPblacktext
    \GPblacktexttrue
  }{}%
  \let\gplgaddtomacro\g@addto@macro
  \gdef\gplbacktext{}%
  \gdef\gplfronttext{}%
  \makeatother
  \ifGPblacktext
    \def\colorrgb#1{}%
    \def\colorgray#1{}%
  \else
    \ifGPcolor
      \def\colorrgb#1{\color[rgb]{#1}}%
      \def\colorgray#1{\color[gray]{#1}}%
      \expandafter\def\csname LTw\endcsname{\color{white}}%
      \expandafter\def\csname LTb\endcsname{\color{black}}%
      \expandafter\def\csname LTa\endcsname{\color{black}}%
      \expandafter\def\csname LT0\endcsname{\color[rgb]{1,0,0}}%
      \expandafter\def\csname LT1\endcsname{\color[rgb]{0,1,0}}%
      \expandafter\def\csname LT2\endcsname{\color[rgb]{0,0,1}}%
      \expandafter\def\csname LT3\endcsname{\color[rgb]{1,0,1}}%
      \expandafter\def\csname LT4\endcsname{\color[rgb]{0,1,1}}%
      \expandafter\def\csname LT5\endcsname{\color[rgb]{1,1,0}}%
      \expandafter\def\csname LT6\endcsname{\color[rgb]{0,0,0}}%
      \expandafter\def\csname LT7\endcsname{\color[rgb]{1,0.3,0}}%
      \expandafter\def\csname LT8\endcsname{\color[rgb]{0.5,0.5,0.5}}%
    \else
      \def\colorrgb#1{\color{black}}%
      \def\colorgray#1{\color[gray]{#1}}%
      \expandafter\def\csname LTw\endcsname{\color{white}}%
      \expandafter\def\csname LTb\endcsname{\color{black}}%
      \expandafter\def\csname LTa\endcsname{\color{black}}%
      \expandafter\def\csname LT0\endcsname{\color{black}}%
      \expandafter\def\csname LT1\endcsname{\color{black}}%
      \expandafter\def\csname LT2\endcsname{\color{black}}%
      \expandafter\def\csname LT3\endcsname{\color{black}}%
      \expandafter\def\csname LT4\endcsname{\color{black}}%
      \expandafter\def\csname LT5\endcsname{\color{black}}%
      \expandafter\def\csname LT6\endcsname{\color{black}}%
      \expandafter\def\csname LT7\endcsname{\color{black}}%
      \expandafter\def\csname LT8\endcsname{\color{black}}%
    \fi
  \fi
  \setlength{\unitlength}{0.0500bp}%
  \begin{picture}(5442.00,5442.00)%
    \gplgaddtomacro\gplbacktext{%
      \csname LTb\endcsname%
      \put(946,704){\makebox(0,0)[r]{\strut{}-5.0}}%
      \put(946,1112){\makebox(0,0)[r]{\strut{}-4.0}}%
      \put(946,1519){\makebox(0,0)[r]{\strut{}-3.0}}%
      \put(946,1927){\makebox(0,0)[r]{\strut{}-2.0}}%
      \put(946,2335){\makebox(0,0)[r]{\strut{}-1.0}}%
      \put(946,2743){\makebox(0,0)[r]{\strut{}0.0}}%
      \put(946,3150){\makebox(0,0)[r]{\strut{}1.0}}%
      \put(946,3558){\makebox(0,0)[r]{\strut{}2.0}}%
      \put(946,3966){\makebox(0,0)[r]{\strut{}3.0}}%
      \put(946,4373){\makebox(0,0)[r]{\strut{}4.0}}%
      \put(946,4781){\makebox(0,0)[r]{\strut{}5.0}}%
      \put(1078,484){\makebox(0,0){\strut{}-5.0}}%
      \put(1475,484){\makebox(0,0){\strut{}-4.0}}%
      \put(1871,484){\makebox(0,0){\strut{}-3.0}}%
      \put(2268,484){\makebox(0,0){\strut{}-2.0}}%
      \put(2665,484){\makebox(0,0){\strut{}-1.0}}%
      \put(3062,484){\makebox(0,0){\strut{}0.0}}%
      \put(3458,484){\makebox(0,0){\strut{}1.0}}%
      \put(3855,484){\makebox(0,0){\strut{}2.0}}%
      \put(4252,484){\makebox(0,0){\strut{}3.0}}%
      \put(4648,484){\makebox(0,0){\strut{}4.0}}%
      \put(5045,484){\makebox(0,0){\strut{}5.0}}%
      \put(176,2742){\rotatebox{-270}{\makebox(0,0){\strut{}$X_{2}$}}}%
      \put(3061,154){\makebox(0,0){\strut{}$X_{1}$}}%
      \put(3061,5111){\makebox(0,0){\strut{}$t$-copula with $\nu = 10$ and $t$-distributed margins}}%
    }%
    \gplgaddtomacro\gplfronttext{%
    }%
    \gplbacktext
    \put(0,0){\includegraphics{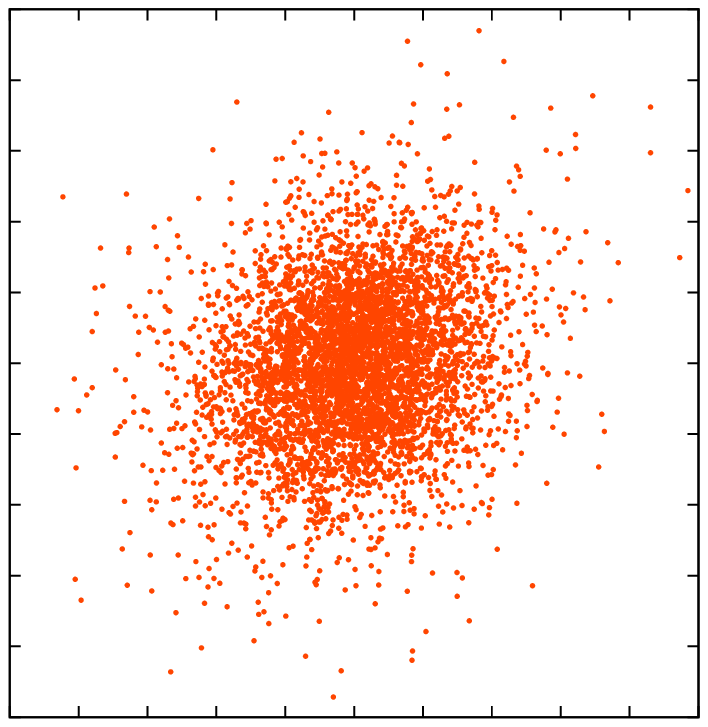}}%
    \gplfronttext
  \end{picture}%
\endgroup

%% file: plot_irb_ead_121231.tex
\begingroup
  \makeatletter
  \providecommand\color[2][]{%
    \GenericError{(gnuplot) \space\space\space\@spaces}{%
      Package color not loaded in conjunction with
      terminal option `colourtext'%
    }{See the gnuplot documentation for explanation.%
    }{Either use 'blacktext' in gnuplot or load the package
      color.sty in LaTeX.}%
    \renewcommand\color[2][]{}%
  }%
  \providecommand\includegraphics[2][]{%
    \GenericError{(gnuplot) \space\space\space\@spaces}{%
      Package graphicx or graphics not loaded%
    }{See the gnuplot documentation for explanation.%
    }{The gnuplot epslatex terminal needs graphicx.sty or graphics.sty.}%
    \renewcommand\includegraphics[2][]{}%
  }%
  \providecommand\rotatebox[2]{#2}%
  \@ifundefined{ifGPcolor}{%
    \newif\ifGPcolor
    \GPcolorfalse
  }{}%
  \@ifundefined{ifGPblacktext}{%
    \newif\ifGPblacktext
    \GPblacktexttrue
  }{}%
  \let\gplgaddtomacro\g@addto@macro
  \gdef\gplbacktext{}%
  \gdef\gplfronttext{}%
  \makeatother
  \ifGPblacktext
    \def\colorrgb#1{}%
    \def\colorgray#1{}%
  \else
    \ifGPcolor
      \def\colorrgb#1{\color[rgb]{#1}}%
      \def\colorgray#1{\color[gray]{#1}}%
      \expandafter\def\csname LTw\endcsname{\color{white}}%
      \expandafter\def\csname LTb\endcsname{\color{black}}%
      \expandafter\def\csname LTa\endcsname{\color{black}}%
      \expandafter\def\csname LT0\endcsname{\color[rgb]{1,0,0}}%
      \expandafter\def\csname LT1\endcsname{\color[rgb]{0,1,0}}%
      \expandafter\def\csname LT2\endcsname{\color[rgb]{0,0,1}}%
      \expandafter\def\csname LT3\endcsname{\color[rgb]{1,0,1}}%
      \expandafter\def\csname LT4\endcsname{\color[rgb]{0,1,1}}%
      \expandafter\def\csname LT5\endcsname{\color[rgb]{1,1,0}}%
      \expandafter\def\csname LT6\endcsname{\color[rgb]{0,0,0}}%
      \expandafter\def\csname LT7\endcsname{\color[rgb]{1,0.3,0}}%
      \expandafter\def\csname LT8\endcsname{\color[rgb]{0.5,0.5,0.5}}%
    \else
      \def\colorrgb#1{\color{black}}%
      \def\colorgray#1{\color[gray]{#1}}%
      \expandafter\def\csname LTw\endcsname{\color{white}}%
      \expandafter\def\csname LTb\endcsname{\color{black}}%
      \expandafter\def\csname LTa\endcsname{\color{black}}%
      \expandafter\def\csname LT0\endcsname{\color{black}}%
      \expandafter\def\csname LT1\endcsname{\color{black}}%
      \expandafter\def\csname LT2\endcsname{\color{black}}%
      \expandafter\def\csname LT3\endcsname{\color{black}}%
      \expandafter\def\csname LT4\endcsname{\color{black}}%
      \expandafter\def\csname LT5\endcsname{\color{black}}%
      \expandafter\def\csname LT6\endcsname{\color{black}}%
      \expandafter\def\csname LT7\endcsname{\color{black}}%
      \expandafter\def\csname LT8\endcsname{\color{black}}%
    \fi
  \fi
  \setlength{\unitlength}{0.0500bp}%
  \begin{picture}(5442.00,5442.00)%
    \gplgaddtomacro\gplbacktext{%
      \csname LTb\endcsname%
      \put(2721,4691){\makebox(0,0){\strut{}\large{Exposure at default by sector}}}%
      \put(3447,4129){\makebox(0,0)[l]{\strut{}Business (35.8\%)}}%
    }%
    \gplgaddtomacro\gplfronttext{%
    }%
    \gplgaddtomacro\gplbacktext{%
      \csname LTb\endcsname%
      \put(2721,4691){\makebox(0,0){\strut{}\large{Exposure at default by sector}}}%
      \put(1379,3625){\makebox(0,0)[r]{\strut{}Government (7.8\%)}}%
    }%
    \gplgaddtomacro\gplfronttext{%
    }%
    \gplgaddtomacro\gplbacktext{%
      \csname LTb\endcsname%
      \put(2721,4691){\makebox(0,0){\strut{}\large{Exposure at default by sector}}}%
      \put(2385,962){\makebox(0,0)[r]{\strut{}Household (56.4\%)}}%
    }%
    \gplgaddtomacro\gplfronttext{%
    }%
    \gplbacktext
    \put(0,0){\includegraphics{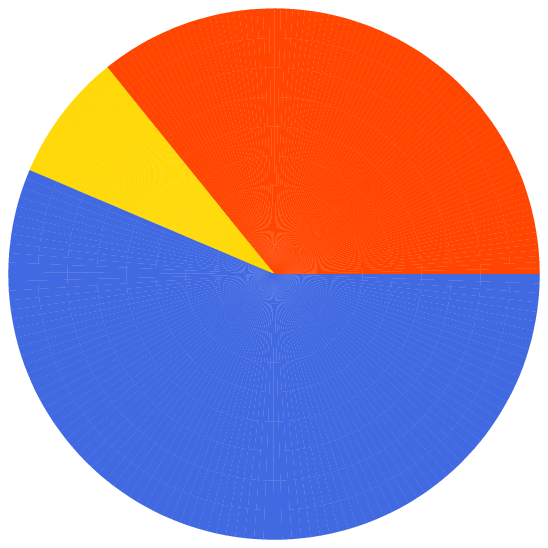}}%
    \gplfronttext
  \end{picture}%
\endgroup

%% file: plot_irb_rwa_121231.tex
\begingroup
  \makeatletter
  \providecommand\color[2][]{%
    \GenericError{(gnuplot) \space\space\space\@spaces}{%
      Package color not loaded in conjunction with
      terminal option `colourtext'%
    }{See the gnuplot documentation for explanation.%
    }{Either use 'blacktext' in gnuplot or load the package
      color.sty in LaTeX.}%
    \renewcommand\color[2][]{}%
  }%
  \providecommand\includegraphics[2][]{%
    \GenericError{(gnuplot) \space\space\space\@spaces}{%
      Package graphicx or graphics not loaded%
    }{See the gnuplot documentation for explanation.%
    }{The gnuplot epslatex terminal needs graphicx.sty or graphics.sty.}%
    \renewcommand\includegraphics[2][]{}%
  }%
  \providecommand\rotatebox[2]{#2}%
  \@ifundefined{ifGPcolor}{%
    \newif\ifGPcolor
    \GPcolorfalse
  }{}%
  \@ifundefined{ifGPblacktext}{%
    \newif\ifGPblacktext
    \GPblacktexttrue
  }{}%
  \let\gplgaddtomacro\g@addto@macro
  \gdef\gplbacktext{}%
  \gdef\gplfronttext{}%
  \makeatother
  \ifGPblacktext
    \def\colorrgb#1{}%
    \def\colorgray#1{}%
  \else
    \ifGPcolor
      \def\colorrgb#1{\color[rgb]{#1}}%
      \def\colorgray#1{\color[gray]{#1}}%
      \expandafter\def\csname LTw\endcsname{\color{white}}%
      \expandafter\def\csname LTb\endcsname{\color{black}}%
      \expandafter\def\csname LTa\endcsname{\color{black}}%
      \expandafter\def\csname LT0\endcsname{\color[rgb]{1,0,0}}%
      \expandafter\def\csname LT1\endcsname{\color[rgb]{0,1,0}}%
      \expandafter\def\csname LT2\endcsname{\color[rgb]{0,0,1}}%
      \expandafter\def\csname LT3\endcsname{\color[rgb]{1,0,1}}%
      \expandafter\def\csname LT4\endcsname{\color[rgb]{0,1,1}}%
      \expandafter\def\csname LT5\endcsname{\color[rgb]{1,1,0}}%
      \expandafter\def\csname LT6\endcsname{\color[rgb]{0,0,0}}%
      \expandafter\def\csname LT7\endcsname{\color[rgb]{1,0.3,0}}%
      \expandafter\def\csname LT8\endcsname{\color[rgb]{0.5,0.5,0.5}}%
    \else
      \def\colorrgb#1{\color{black}}%
      \def\colorgray#1{\color[gray]{#1}}%
      \expandafter\def\csname LTw\endcsname{\color{white}}%
      \expandafter\def\csname LTb\endcsname{\color{black}}%
      \expandafter\def\csname LTa\endcsname{\color{black}}%
      \expandafter\def\csname LT0\endcsname{\color{black}}%
      \expandafter\def\csname LT1\endcsname{\color{black}}%
      \expandafter\def\csname LT2\endcsname{\color{black}}%
      \expandafter\def\csname LT3\endcsname{\color{black}}%
      \expandafter\def\csname LT4\endcsname{\color{black}}%
      \expandafter\def\csname LT5\endcsname{\color{black}}%
      \expandafter\def\csname LT6\endcsname{\color{black}}%
      \expandafter\def\csname LT7\endcsname{\color{black}}%
      \expandafter\def\csname LT8\endcsname{\color{black}}%
    \fi
  \fi
  \setlength{\unitlength}{0.0500bp}%
  \begin{picture}(5442.00,5442.00)%
    \gplgaddtomacro\gplbacktext{%
      \csname LTb\endcsname%
      \put(2721,4691){\makebox(0,0){\strut{}\large{Risk-weighted assets by sector}}}%
      \put(2112,4179){\makebox(0,0)[r]{\strut{}Business (61.8\%)}}%
    }%
    \gplgaddtomacro\gplfronttext{%
    }%
    \gplgaddtomacro\gplbacktext{%
      \csname LTb\endcsname%
      \put(2721,4691){\makebox(0,0){\strut{}\large{Risk-weighted assets by sector}}}%
      \put(1531,1421){\makebox(0,0)[r]{\strut{}Government (1.4\%)}}%
    }%
    \gplgaddtomacro\gplfronttext{%
    }%
    \gplgaddtomacro\gplbacktext{%
      \csname LTb\endcsname%
      \put(2721,4691){\makebox(0,0){\strut{}\large{Risk-weighted assets by sector}}}%
      \put(3399,1071){\makebox(0,0)[l]{\strut{}Household (36.8\%)}}%
    }%
    \gplgaddtomacro\gplfronttext{%
    }%
    \gplbacktext
    \put(0,0){\includegraphics{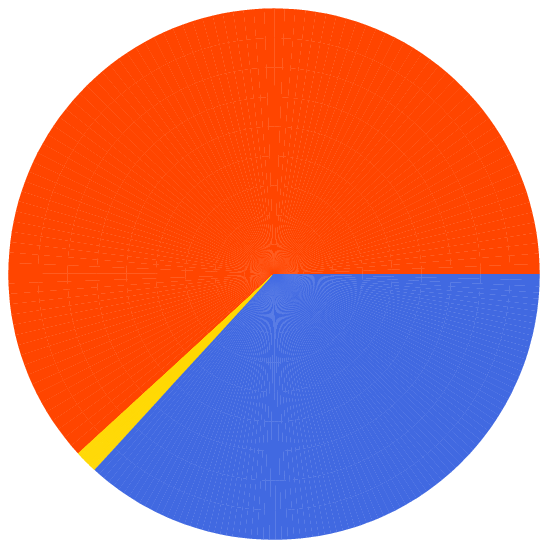}}%
    \gplfronttext
  \end{picture}%
\endgroup

%% file: plot_empir_loss_dist.tex
\begingroup
  \makeatletter
  \providecommand\color[2][]{%
    \GenericError{(gnuplot) \space\space\space\@spaces}{%
      Package color not loaded in conjunction with
      terminal option `colourtext'%
    }{See the gnuplot documentation for explanation.%
    }{Either use 'blacktext' in gnuplot or load the package
      color.sty in LaTeX.}%
    \renewcommand\color[2][]{}%
  }%
  \providecommand\includegraphics[2][]{%
    \GenericError{(gnuplot) \space\space\space\@spaces}{%
      Package graphicx or graphics not loaded%
    }{See the gnuplot documentation for explanation.%
    }{The gnuplot epslatex terminal needs graphicx.sty or graphics.sty.}%
    \renewcommand\includegraphics[2][]{}%
  }%
  \providecommand\rotatebox[2]{#2}%
  \@ifundefined{ifGPcolor}{%
    \newif\ifGPcolor
    \GPcolorfalse
  }{}%
  \@ifundefined{ifGPblacktext}{%
    \newif\ifGPblacktext
    \GPblacktexttrue
  }{}%
  \let\gplgaddtomacro\g@addto@macro
  \gdef\gplbacktext{}%
  \gdef\gplfronttext{}%
  \makeatother
  \ifGPblacktext
    \def\colorrgb#1{}%
    \def\colorgray#1{}%
  \else
    \ifGPcolor
      \def\colorrgb#1{\color[rgb]{#1}}%
      \def\colorgray#1{\color[gray]{#1}}%
      \expandafter\def\csname LTw\endcsname{\color{white}}%
      \expandafter\def\csname LTb\endcsname{\color{black}}%
      \expandafter\def\csname LTa\endcsname{\color{black}}%
      \expandafter\def\csname LT0\endcsname{\color[rgb]{1,0,0}}%
      \expandafter\def\csname LT1\endcsname{\color[rgb]{0,1,0}}%
      \expandafter\def\csname LT2\endcsname{\color[rgb]{0,0,1}}%
      \expandafter\def\csname LT3\endcsname{\color[rgb]{1,0,1}}%
      \expandafter\def\csname LT4\endcsname{\color[rgb]{0,1,1}}%
      \expandafter\def\csname LT5\endcsname{\color[rgb]{1,1,0}}%
      \expandafter\def\csname LT6\endcsname{\color[rgb]{0,0,0}}%
      \expandafter\def\csname LT7\endcsname{\color[rgb]{1,0.3,0}}%
      \expandafter\def\csname LT8\endcsname{\color[rgb]{0.5,0.5,0.5}}%
    \else
      \def\colorrgb#1{\color{black}}%
      \def\colorgray#1{\color[gray]{#1}}%
      \expandafter\def\csname LTw\endcsname{\color{white}}%
      \expandafter\def\csname LTb\endcsname{\color{black}}%
      \expandafter\def\csname LTa\endcsname{\color{black}}%
      \expandafter\def\csname LT0\endcsname{\color{black}}%
      \expandafter\def\csname LT1\endcsname{\color{black}}%
      \expandafter\def\csname LT2\endcsname{\color{black}}%
      \expandafter\def\csname LT3\endcsname{\color{black}}%
      \expandafter\def\csname LT4\endcsname{\color{black}}%
      \expandafter\def\csname LT5\endcsname{\color{black}}%
      \expandafter\def\csname LT6\endcsname{\color{black}}%
      \expandafter\def\csname LT7\endcsname{\color{black}}%
      \expandafter\def\csname LT8\endcsname{\color{black}}%
    \fi
  \fi
  \setlength{\unitlength}{0.0500bp}%
  \begin{picture}(8162.00,5442.00)%
    \gplgaddtomacro\gplbacktext{%
      \csname LTb\endcsname%
      \put(418,704){\makebox(0,0)[r]{\strut{}}}%
      \put(418,1263){\makebox(0,0)[r]{\strut{}}}%
      \put(418,1822){\makebox(0,0)[r]{\strut{}}}%
      \put(418,2381){\makebox(0,0)[r]{\strut{}}}%
      \put(418,2941){\makebox(0,0)[r]{\strut{}}}%
      \put(418,3500){\makebox(0,0)[r]{\strut{}}}%
      \put(418,4059){\makebox(0,0)[r]{\strut{}}}%
      \put(418,4618){\makebox(0,0)[r]{\strut{}}}%
      \put(418,5177){\makebox(0,0)[r]{\strut{}}}%
      \put(550,484){\makebox(0,0){\strut{}0.0}}%
      \put(1151,484){\makebox(0,0){\strut{}0.2}}%
      \put(1753,484){\makebox(0,0){\strut{}0.4}}%
      \put(2354,484){\makebox(0,0){\strut{}0.6}}%
      \put(2955,484){\makebox(0,0){\strut{}0.8}}%
      \put(3556,484){\makebox(0,0){\strut{}1.0}}%
      \put(4158,484){\makebox(0,0){\strut{}1.2}}%
      \put(4759,484){\makebox(0,0){\strut{}1.4}}%
      \put(5360,484){\makebox(0,0){\strut{}1.6}}%
      \put(5961,484){\makebox(0,0){\strut{}1.8}}%
      \put(6562,484){\makebox(0,0){\strut{}2.0}}%
      \put(7164,484){\makebox(0,0){\strut{}2.2}}%
      \put(7765,484){\makebox(0,0){\strut{}2.4}}%
      \put(176,2940){\rotatebox{-270}{\makebox(0,0){\strut{}Frequency}}}%
      \put(4157,154){\makebox(0,0){\strut{}Portfolio percentage loss $L_{n}$, \% of EAD}}%
    }%
    \gplgaddtomacro\gplfronttext{%
    }%
    \gplbacktext
    \put(0,0){\includegraphics{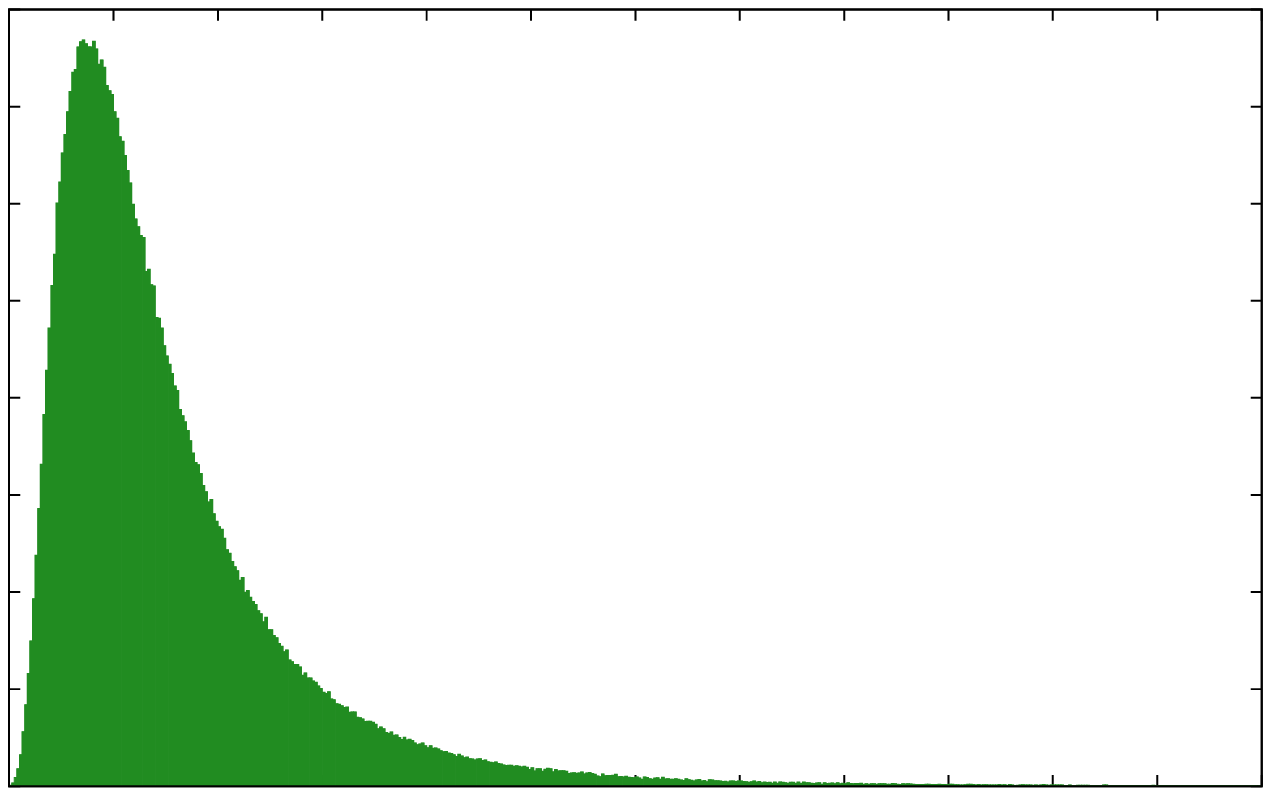}}%
    \gplfronttext
  \end{picture}%
\endgroup

%% file: plot_cvrg_asymp_dist.tex
\begingroup
  \makeatletter
  \providecommand\color[2][]{%
    \GenericError{(gnuplot) \space\space\space\@spaces}{%
      Package color not loaded in conjunction with
      terminal option `colourtext'%
    }{See the gnuplot documentation for explanation.%
    }{Either use 'blacktext' in gnuplot or load the package
      color.sty in LaTeX.}%
    \renewcommand\color[2][]{}%
  }%
  \providecommand\includegraphics[2][]{%
    \GenericError{(gnuplot) \space\space\space\@spaces}{%
      Package graphicx or graphics not loaded%
    }{See the gnuplot documentation for explanation.%
    }{The gnuplot epslatex terminal needs graphicx.sty or graphics.sty.}%
    \renewcommand\includegraphics[2][]{}%
  }%
  \providecommand\rotatebox[2]{#2}%
  \@ifundefined{ifGPcolor}{%
    \newif\ifGPcolor
    \GPcolorfalse
  }{}%
  \@ifundefined{ifGPblacktext}{%
    \newif\ifGPblacktext
    \GPblacktexttrue
  }{}%
  \let\gplgaddtomacro\g@addto@macro
  \gdef\gplbacktext{}%
  \gdef\gplfronttext{}%
  \makeatother
  \ifGPblacktext
    \def\colorrgb#1{}%
    \def\colorgray#1{}%
  \else
    \ifGPcolor
      \def\colorrgb#1{\color[rgb]{#1}}%
      \def\colorgray#1{\color[gray]{#1}}%
      \expandafter\def\csname LTw\endcsname{\color{white}}%
      \expandafter\def\csname LTb\endcsname{\color{black}}%
      \expandafter\def\csname LTa\endcsname{\color{black}}%
      \expandafter\def\csname LT0\endcsname{\color[rgb]{1,0,0}}%
      \expandafter\def\csname LT1\endcsname{\color[rgb]{0,1,0}}%
      \expandafter\def\csname LT2\endcsname{\color[rgb]{0,0,1}}%
      \expandafter\def\csname LT3\endcsname{\color[rgb]{1,0,1}}%
      \expandafter\def\csname LT4\endcsname{\color[rgb]{0,1,1}}%
      \expandafter\def\csname LT5\endcsname{\color[rgb]{1,1,0}}%
      \expandafter\def\csname LT6\endcsname{\color[rgb]{0,0,0}}%
      \expandafter\def\csname LT7\endcsname{\color[rgb]{1,0.3,0}}%
      \expandafter\def\csname LT8\endcsname{\color[rgb]{0.5,0.5,0.5}}%
    \else
      \def\colorrgb#1{\color{black}}%
      \def\colorgray#1{\color[gray]{#1}}%
      \expandafter\def\csname LTw\endcsname{\color{white}}%
      \expandafter\def\csname LTb\endcsname{\color{black}}%
      \expandafter\def\csname LTa\endcsname{\color{black}}%
      \expandafter\def\csname LT0\endcsname{\color{black}}%
      \expandafter\def\csname LT1\endcsname{\color{black}}%
      \expandafter\def\csname LT2\endcsname{\color{black}}%
      \expandafter\def\csname LT3\endcsname{\color{black}}%
      \expandafter\def\csname LT4\endcsname{\color{black}}%
      \expandafter\def\csname LT5\endcsname{\color{black}}%
      \expandafter\def\csname LT6\endcsname{\color{black}}%
      \expandafter\def\csname LT7\endcsname{\color{black}}%
      \expandafter\def\csname LT8\endcsname{\color{black}}%
    \fi
  \fi
  \setlength{\unitlength}{0.0500bp}%
  \begin{picture}(8162.00,5442.00)%
    \gplgaddtomacro\gplbacktext{%
      \csname LTb\endcsname%
      \put(814,704){\makebox(0,0)[r]{\strut{}3.0}}%
      \put(814,1450){\makebox(0,0)[r]{\strut{}4.0}}%
      \put(814,2195){\makebox(0,0)[r]{\strut{}5.0}}%
      \put(814,2941){\makebox(0,0)[r]{\strut{}6.0}}%
      \put(814,3686){\makebox(0,0)[r]{\strut{}7.0}}%
      \put(814,4432){\makebox(0,0)[r]{\strut{}8.0}}%
      \put(814,5177){\makebox(0,0)[r]{\strut{}9.0}}%
      \put(946,484){\makebox(0,0){\strut{}99.0}}%
      \put(1628,484){\makebox(0,0){\strut{}99.1}}%
      \put(2310,484){\makebox(0,0){\strut{}99.2}}%
      \put(2992,484){\makebox(0,0){\strut{}99.3}}%
      \put(3674,484){\makebox(0,0){\strut{}99.4}}%
      \put(4355,484){\makebox(0,0){\strut{}99.5}}%
      \put(5037,484){\makebox(0,0){\strut{}99.6}}%
      \put(5719,484){\makebox(0,0){\strut{}99.7}}%
      \put(6401,484){\makebox(0,0){\strut{}99.8}}%
      \put(7083,484){\makebox(0,0){\strut{}99.9}}%
      \put(7765,484){\makebox(0,0){\strut{}100.0}}%
      \put(176,2940){\rotatebox{-270}{\makebox(0,0){\strut{}VaR$_{\alpha}(L_{n})$, \% of EAD}}}%
      \put(4355,154){\makebox(0,0){\strut{}Confidence level $\alpha$, \%}}%
    }%
    \gplgaddtomacro\gplfronttext{%
      \csname LTb\endcsname%
      \put(2508,4694){\makebox(0,0)[l]{\strut{}\small{50 obligors}}}%
      \csname LTb\endcsname%
      \put(2508,4474){\makebox(0,0)[l]{\strut{}\small{100 obligors}}}%
      \csname LTb\endcsname%
      \put(2508,4254){\makebox(0,0)[l]{\strut{}\small{200 obligors}}}%
      \csname LTb\endcsname%
      \put(2508,4034){\makebox(0,0)[l]{\strut{}\small{500 obligors}}}%
      \csname LTb\endcsname%
      \put(2508,3814){\makebox(0,0)[l]{\strut{}\small{1000 obligors}}}%
      \csname LTb\endcsname%
      \put(2508,3594){\makebox(0,0)[l]{\strut{}\small{2000 obligors}}}%
      \csname LTb\endcsname%
      \put(2508,3374){\makebox(0,0)[l]{\strut{}\small{$\E[L_{n} \given Y = \Phi^{-1}(1\!-\!\alpha)]$}}}%
    }%
    \gplbacktext
    \put(0,0){\includegraphics{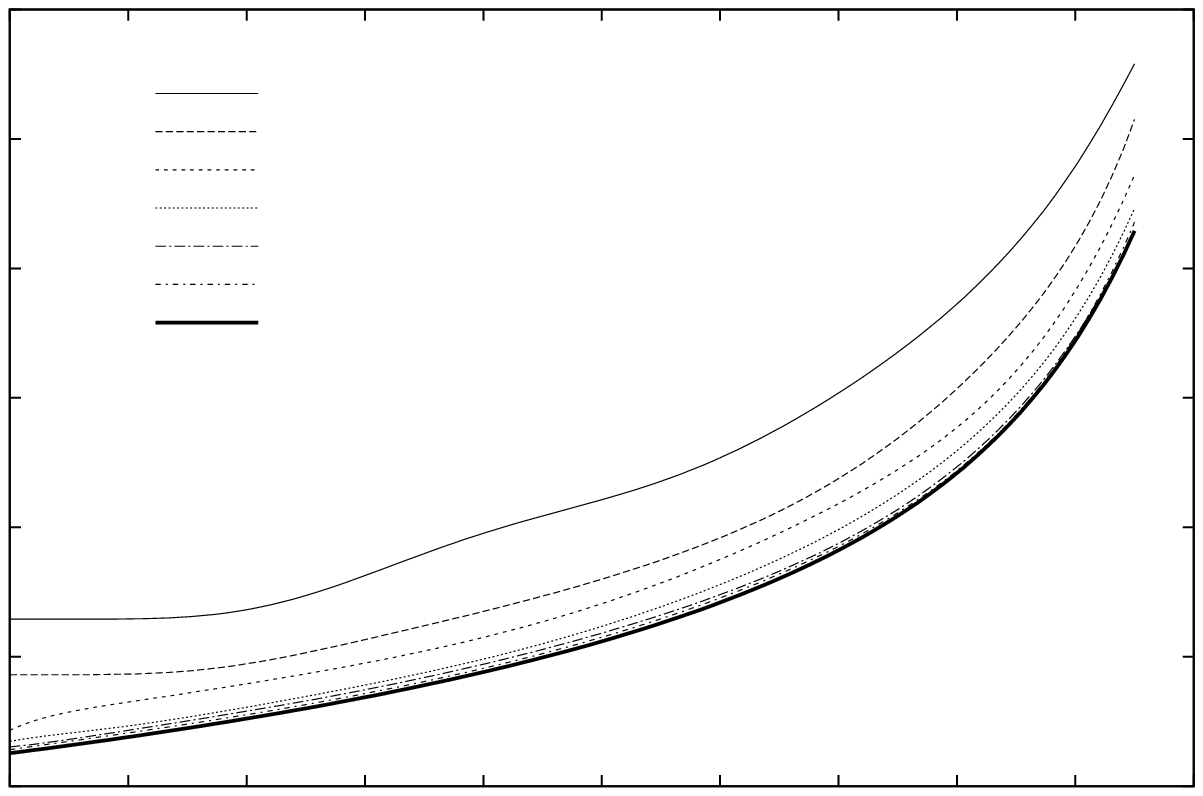}}%
    \gplfronttext
  \end{picture}%
\endgroup

%% file: plot_var_qntl_asset_corr.tex
\begingroup
  \makeatletter
  \providecommand\color[2][]{%
    \GenericError{(gnuplot) \space\space\space\@spaces}{%
      Package color not loaded in conjunction with
      terminal option `colourtext'%
    }{See the gnuplot documentation for explanation.%
    }{Either use 'blacktext' in gnuplot or load the package
      color.sty in LaTeX.}%
    \renewcommand\color[2][]{}%
  }%
  \providecommand\includegraphics[2][]{%
    \GenericError{(gnuplot) \space\space\space\@spaces}{%
      Package graphicx or graphics not loaded%
    }{See the gnuplot documentation for explanation.%
    }{The gnuplot epslatex terminal needs graphicx.sty or graphics.sty.}%
    \renewcommand\includegraphics[2][]{}%
  }%
  \providecommand\rotatebox[2]{#2}%
  \@ifundefined{ifGPcolor}{%
    \newif\ifGPcolor
    \GPcolorfalse
  }{}%
  \@ifundefined{ifGPblacktext}{%
    \newif\ifGPblacktext
    \GPblacktexttrue
  }{}%
  \let\gplgaddtomacro\g@addto@macro
  \gdef\gplbacktext{}%
  \gdef\gplfronttext{}%
  \makeatother
  \ifGPblacktext
    \def\colorrgb#1{}%
    \def\colorgray#1{}%
  \else
    \ifGPcolor
      \def\colorrgb#1{\color[rgb]{#1}}%
      \def\colorgray#1{\color[gray]{#1}}%
      \expandafter\def\csname LTw\endcsname{\color{white}}%
      \expandafter\def\csname LTb\endcsname{\color{black}}%
      \expandafter\def\csname LTa\endcsname{\color{black}}%
      \expandafter\def\csname LT0\endcsname{\color[rgb]{1,0,0}}%
      \expandafter\def\csname LT1\endcsname{\color[rgb]{0,1,0}}%
      \expandafter\def\csname LT2\endcsname{\color[rgb]{0,0,1}}%
      \expandafter\def\csname LT3\endcsname{\color[rgb]{1,0,1}}%
      \expandafter\def\csname LT4\endcsname{\color[rgb]{0,1,1}}%
      \expandafter\def\csname LT5\endcsname{\color[rgb]{1,1,0}}%
      \expandafter\def\csname LT6\endcsname{\color[rgb]{0,0,0}}%
      \expandafter\def\csname LT7\endcsname{\color[rgb]{1,0.3,0}}%
      \expandafter\def\csname LT8\endcsname{\color[rgb]{0.5,0.5,0.5}}%
    \else
      \def\colorrgb#1{\color{black}}%
      \def\colorgray#1{\color[gray]{#1}}%
      \expandafter\def\csname LTw\endcsname{\color{white}}%
      \expandafter\def\csname LTb\endcsname{\color{black}}%
      \expandafter\def\csname LTa\endcsname{\color{black}}%
      \expandafter\def\csname LT0\endcsname{\color{black}}%
      \expandafter\def\csname LT1\endcsname{\color{black}}%
      \expandafter\def\csname LT2\endcsname{\color{black}}%
      \expandafter\def\csname LT3\endcsname{\color{black}}%
      \expandafter\def\csname LT4\endcsname{\color{black}}%
      \expandafter\def\csname LT5\endcsname{\color{black}}%
      \expandafter\def\csname LT6\endcsname{\color{black}}%
      \expandafter\def\csname LT7\endcsname{\color{black}}%
      \expandafter\def\csname LT8\endcsname{\color{black}}%
    \fi
  \fi
  \setlength{\unitlength}{0.0500bp}%
  \begin{picture}(8162.00,5442.00)%
    \gplgaddtomacro\gplbacktext{%
      \csname LTb\endcsname%
      \put(814,704){\makebox(0,0)[r]{\strut{}0.0}}%
      \put(814,1263){\makebox(0,0)[r]{\strut{}0.4}}%
      \put(814,1822){\makebox(0,0)[r]{\strut{}0.8}}%
      \put(814,2381){\makebox(0,0)[r]{\strut{}1.2}}%
      \put(814,2941){\makebox(0,0)[r]{\strut{}1.6}}%
      \put(814,3500){\makebox(0,0)[r]{\strut{}2.0}}%
      \put(814,4059){\makebox(0,0)[r]{\strut{}2.4}}%
      \put(814,4618){\makebox(0,0)[r]{\strut{}2.8}}%
      \put(814,5177){\makebox(0,0)[r]{\strut{}3.2}}%
      \put(946,484){\makebox(0,0){\strut{}90.0}}%
      \put(1628,484){\makebox(0,0){\strut{}91.0}}%
      \put(2310,484){\makebox(0,0){\strut{}92.0}}%
      \put(2992,484){\makebox(0,0){\strut{}93.0}}%
      \put(3674,484){\makebox(0,0){\strut{}94.0}}%
      \put(4356,484){\makebox(0,0){\strut{}95.0}}%
      \put(5037,484){\makebox(0,0){\strut{}96.0}}%
      \put(5719,484){\makebox(0,0){\strut{}97.0}}%
      \put(6401,484){\makebox(0,0){\strut{}98.0}}%
      \put(7083,484){\makebox(0,0){\strut{}99.0}}%
      \put(7765,484){\makebox(0,0){\strut{}100.0}}%
      \put(176,2940){\rotatebox{-270}{\makebox(0,0){\strut{}${\E\big[L_{n} \given Y=\Phi^{-1}(1\!-\!\alpha)\big]}$, \% of EAD}}}%
      \put(4355,154){\makebox(0,0){\strut{}Confidence level $\alpha$, \%}}%
    }%
    \gplgaddtomacro\gplfronttext{%
      \csname LTb\endcsname%
      \put(1933,5004){\makebox(0,0)[l]{\strut{}\small{$1.2 \times \rho$}}}%
      \csname LTb\endcsname%
      \put(1933,4784){\makebox(0,0)[l]{\strut{}\small{$1.1 \times \rho$}}}%
      \csname LTb\endcsname%
      \put(1933,4564){\makebox(0,0)[l]{\strut{}\small{$1.0 \times \rho$}}}%
      \csname LTb\endcsname%
      \put(1933,4344){\makebox(0,0)[l]{\strut{}\small{$0.9 \times \rho$}}}%
      \csname LTb\endcsname%
      \put(1933,4124){\makebox(0,0)[l]{\strut{}\small{$0.8 \times \rho$}}}%
    }%
    \gplbacktext
    \put(0,0){\includegraphics{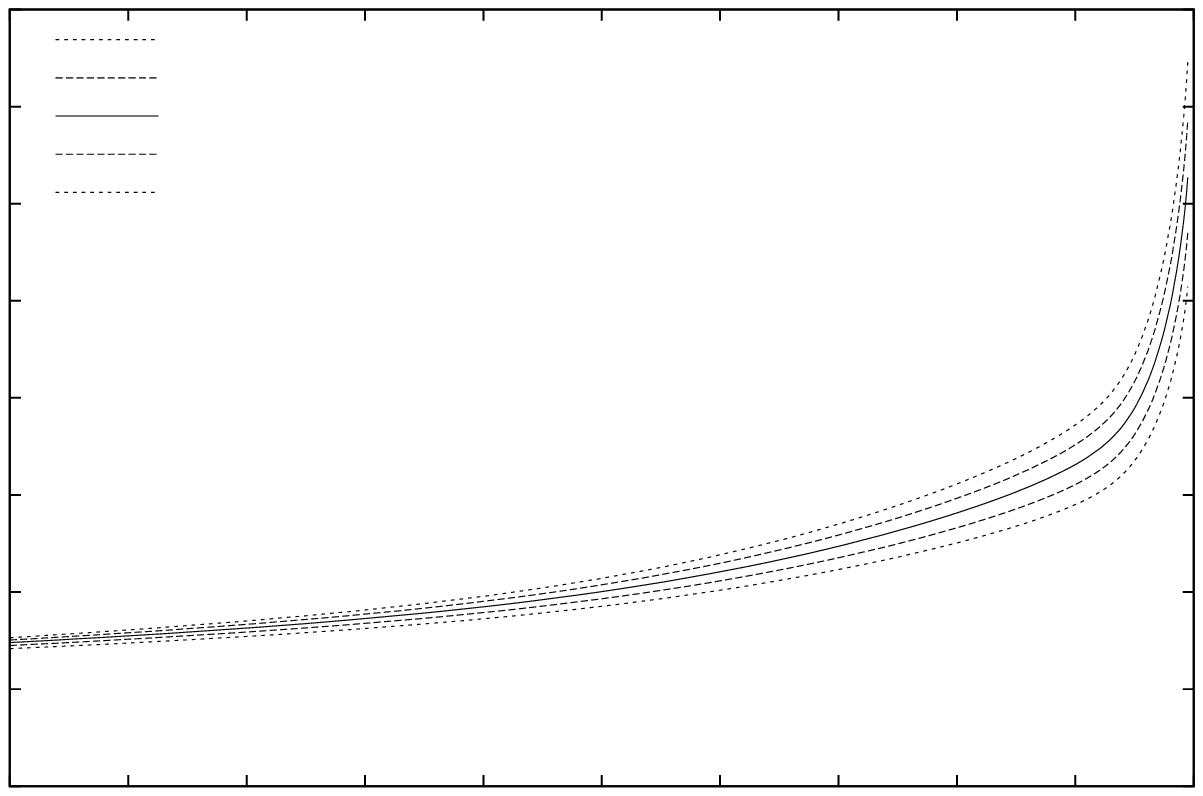}}%
    \gplfronttext
  \end{picture}%
\endgroup

%% file: plot_var_qntl_ellip_cop.tex
\begingroup
  \makeatletter
  \providecommand\color[2][]{%
    \GenericError{(gnuplot) \space\space\space\@spaces}{%
      Package color not loaded in conjunction with
      terminal option `colourtext'%
    }{See the gnuplot documentation for explanation.%
    }{Either use 'blacktext' in gnuplot or load the package
      color.sty in LaTeX.}%
    \renewcommand\color[2][]{}%
  }%
  \providecommand\includegraphics[2][]{%
    \GenericError{(gnuplot) \space\space\space\@spaces}{%
      Package graphicx or graphics not loaded%
    }{See the gnuplot documentation for explanation.%
    }{The gnuplot epslatex terminal needs graphicx.sty or graphics.sty.}%
    \renewcommand\includegraphics[2][]{}%
  }%
  \providecommand\rotatebox[2]{#2}%
  \@ifundefined{ifGPcolor}{%
    \newif\ifGPcolor
    \GPcolorfalse
  }{}%
  \@ifundefined{ifGPblacktext}{%
    \newif\ifGPblacktext
    \GPblacktexttrue
  }{}%
  \let\gplgaddtomacro\g@addto@macro
  \gdef\gplbacktext{}%
  \gdef\gplfronttext{}%
  \makeatother
  \ifGPblacktext
    \def\colorrgb#1{}%
    \def\colorgray#1{}%
  \else
    \ifGPcolor
      \def\colorrgb#1{\color[rgb]{#1}}%
      \def\colorgray#1{\color[gray]{#1}}%
      \expandafter\def\csname LTw\endcsname{\color{white}}%
      \expandafter\def\csname LTb\endcsname{\color{black}}%
      \expandafter\def\csname LTa\endcsname{\color{black}}%
      \expandafter\def\csname LT0\endcsname{\color[rgb]{1,0,0}}%
      \expandafter\def\csname LT1\endcsname{\color[rgb]{0,1,0}}%
      \expandafter\def\csname LT2\endcsname{\color[rgb]{0,0,1}}%
      \expandafter\def\csname LT3\endcsname{\color[rgb]{1,0,1}}%
      \expandafter\def\csname LT4\endcsname{\color[rgb]{0,1,1}}%
      \expandafter\def\csname LT5\endcsname{\color[rgb]{1,1,0}}%
      \expandafter\def\csname LT6\endcsname{\color[rgb]{0,0,0}}%
      \expandafter\def\csname LT7\endcsname{\color[rgb]{1,0.3,0}}%
      \expandafter\def\csname LT8\endcsname{\color[rgb]{0.5,0.5,0.5}}%
    \else
      \def\colorrgb#1{\color{black}}%
      \def\colorgray#1{\color[gray]{#1}}%
      \expandafter\def\csname LTw\endcsname{\color{white}}%
      \expandafter\def\csname LTb\endcsname{\color{black}}%
      \expandafter\def\csname LTa\endcsname{\color{black}}%
      \expandafter\def\csname LT0\endcsname{\color{black}}%
      \expandafter\def\csname LT1\endcsname{\color{black}}%
      \expandafter\def\csname LT2\endcsname{\color{black}}%
      \expandafter\def\csname LT3\endcsname{\color{black}}%
      \expandafter\def\csname LT4\endcsname{\color{black}}%
      \expandafter\def\csname LT5\endcsname{\color{black}}%
      \expandafter\def\csname LT6\endcsname{\color{black}}%
      \expandafter\def\csname LT7\endcsname{\color{black}}%
      \expandafter\def\csname LT8\endcsname{\color{black}}%
    \fi
  \fi
  \setlength{\unitlength}{0.0500bp}%
  \begin{picture}(8162.00,5442.00)%
    \gplgaddtomacro\gplbacktext{%
      \csname LTb\endcsname%
      \put(946,704){\makebox(0,0)[r]{\strut{}0.0}}%
      \put(946,1450){\makebox(0,0)[r]{\strut{}2.0}}%
      \put(946,2195){\makebox(0,0)[r]{\strut{}4.0}}%
      \put(946,2941){\makebox(0,0)[r]{\strut{}6.0}}%
      \put(946,3686){\makebox(0,0)[r]{\strut{}8.0}}%
      \put(946,4432){\makebox(0,0)[r]{\strut{}10.0}}%
      \put(946,5177){\makebox(0,0)[r]{\strut{}12.0}}%
      \put(1078,484){\makebox(0,0){\strut{}90.0}}%
      \put(1747,484){\makebox(0,0){\strut{}91.0}}%
      \put(2415,484){\makebox(0,0){\strut{}92.0}}%
      \put(3084,484){\makebox(0,0){\strut{}93.0}}%
      \put(3753,484){\makebox(0,0){\strut{}94.0}}%
      \put(4422,484){\makebox(0,0){\strut{}95.0}}%
      \put(5090,484){\makebox(0,0){\strut{}96.0}}%
      \put(5759,484){\makebox(0,0){\strut{}97.0}}%
      \put(6428,484){\makebox(0,0){\strut{}98.0}}%
      \put(7096,484){\makebox(0,0){\strut{}99.0}}%
      \put(7765,484){\makebox(0,0){\strut{}100.0}}%
      \put(176,2940){\rotatebox{-270}{\makebox(0,0){\strut{}VaR$_{\alpha}(L_{n})$, \% of EAD}}}%
      \put(4421,154){\makebox(0,0){\strut{}Confidence level $\alpha$, \%}}%
    }%
    \gplgaddtomacro\gplfronttext{%
      \csname LTb\endcsname%
      \put(2065,5004){\makebox(0,0)[l]{\strut{}\small{Gaussian copula}}}%
      \csname LTb\endcsname%
      \put(2065,4784){\makebox(0,0)[l]{\strut{}\small{$t$-copula, $\nu = 30$, Gaussian margins}}}%
      \csname LTb\endcsname%
      \put(2065,4564){\makebox(0,0)[l]{\strut{}\small{$t$-copula, $\nu = 10$, Gaussian margins}}}%
      \csname LTb\endcsname%
      \put(2065,4344){\makebox(0,0)[l]{\strut{}\small{$t$-copula, $\nu = 3$, Gaussian margins}}}%
    }%
    \gplbacktext
    \put(0,0){\includegraphics{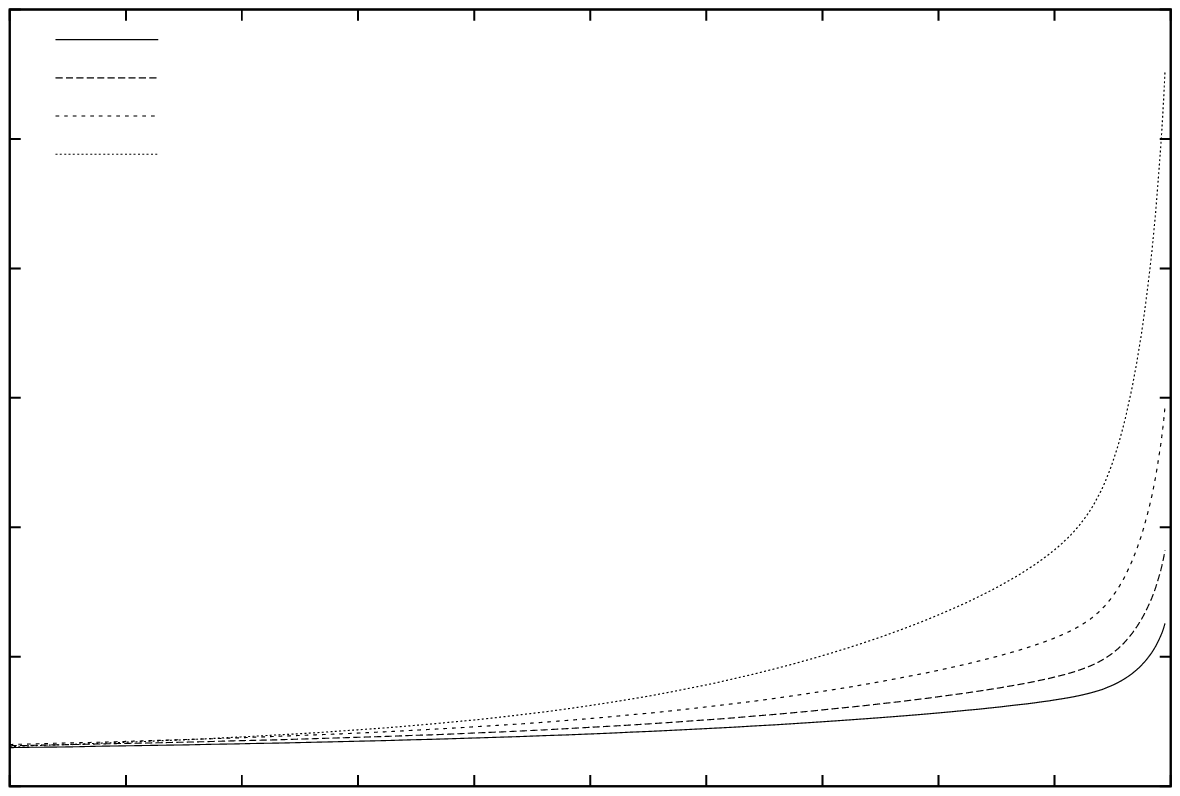}}%
    \gplfronttext
  \end{picture}%
\endgroup